%% file: main.tex
\def\llncs{0}
\def\fullpage{1}
\def\anonymous{0}
\def\authnote{0}
\def\notxfont{0}
\def\submission{0}
\def\cameraready{0}

\def\pure{1}

\ifnum\submission=1
\def\llncs{1}
\fi

\ifnum\llncs=1
\documentclass[envcountsect,a4paper,runningheads,10pt]{llncs}
\else
	\documentclass[letterpaper,hmargin=1.05in,vmargin=1.05in,11pt]{article}
			\ifnum\fullpage=1
		\usepackage{fullpage}
		\fi
\fi

\input{preamble_usepackages.tex}
\input{macros}

\title{
Quantum Public-Key Encryption 
with
Tamper-Resilient 
Public Keys
from One-Way Functions}

\ifnum\anonymous=1
\ifnum\llncs=1
\author{\empty}\institute{\empty}
\else
\author{}
\fi
\else
\ifnum\llncs=1
\author{
Fuyuki Kitagawa\inst{2,3} \and
	Tomoyuki Morimae\inst{1} \and Ryo Nishimaki\inst{2,3}\and Takashi Yamakawa\inst{1,2,3}
}
\institute{
Yukawa Institute for Theoretical Physics, Kyoto University, Kyoto, Japan \and NTT Social Informatics Laboratories, Tokyo, Japan \and NTT Research Center for Theoretical Quantum Information, Atsugi, Japan
}
\else
\author[1,3]{Fuyuki Kitagawa}
\author[2]{Tomoyuki Morimae}
\author[1,3]{Ryo Nishimaki}
\author[1,2,3]{Takashi Yamakawa}
\affil[1]{{\small NTT Social Informatics Laboratories, Tokyo, Japan}\authorcr{\small 
\{fuyuki.kitagawa,ryo.nishimaki,takashi.yamakawa\}@ntt.com}}
\affil[2]{{\small Yukawa Institute for Theoretical Physics, Kyoto University, Kyoto, Japan}\authorcr{\small tomoyuki.morimae@yukawa.kyoto-u.ac.jp}}
\affil[3]{{\small NTT Research Center for Theoretical Quantum Information, Atsugi, Japan}}

\fi 
\fi

\date{}

\begin{document}

\maketitle

\begin{abstract}
We construct quantum public-key encryption from one-way functions.
In our construction, public keys are quantum, but ciphertexts are classical.
Quantum public-key encryption from one-way functions (or weaker primitives such as
pseudorandom function-like states) are also proposed in some recent works [Morimae-Yamakawa, eprint:2022/1336; Coladangelo, eprint:2023/282; Barooti-Grilo-Malavolta-Sattath-Vu-Walter, TCC 2023]. However, they have a huge drawback: they are secure only when quantum public keys can be transmitted to the sender (who runs the encryption algorithm) without being tampered with by the adversary, which seems to require unsatisfactory physical setup assumptions such as secure quantum channels.
Our construction is free from such a drawback: it guarantees the secrecy of the encrypted messages even if we assume only unauthenticated quantum channels. Thus, the encryption is done with adversarially tampered quantum public keys.
Our construction is the first quantum public-key encryption that achieves the goal of classical public-key encryption, namely, to establish secure communication over insecure channels, based only on one-way functions.  
Moreover, we show a generic compiler to upgrade security against chosen plaintext attacks (CPA security) into security against chosen ciphertext attacks (CCA security) only using one-way functions.
As a result, we obtain CCA secure quantum public-key encryption based only on one-way functions.
\end{abstract}

\input{introduction}

\input{technical_overview}

\input{Preliminaries}
\input{definitions}

\input{construction}
\input{security_detailed}
\ifnum\llncs=0
\input{def_cca}

\input{transformations_cca}

\else
\input{submission_CCA}
\fi
\input{reusable}

\ifnum\anonymous=1
\else
\paragraph{Acknowledgments.}
TM is supported by
JST CREST JPMJCR23I3,
JST Moonshot R\verb|&|D JPMJMS2061-5-1-1, 
JST FOREST, 
MEXT QLEAP, 
the Grant-in-Aid for Scientific Research (B) No.JP19H04066, 
the Grant-in Aid for Transformative Research Areas (A) 21H05183,
and 
the Grant-in-Aid for Scientific Research (A) No.22H00522.
\fi

\ifnum\submission=0
\bibliographystyle{alpha} 
\else
\bibliographystyle{splncs04}
\fi
\bibliography{abbrev3,crypto,reference}

\ifnum\submission=1
\else
\appendix 
\fi

\ifnum\cameraready=0
\ifnum\pure=0
\else
\appendix
\ifnum\llncs=1
\input{omitted_related_work}

\input{def_cca}
\input{transformations_cca}
\fi
\input{pure}
\fi
\fi

\end{document}

%% file: preamble_usepackages.tex
\usepackage[%
  colorlinks=true,
  citecolor=blue,
  pagebackref=true
]{hyperref}

\usepackage{amsmath, amsfonts, amssymb, mathtools,amscd}

\usepackage{amsthm}

\allowdisplaybreaks 
\usepackage{lmodern}
\usepackage[T1]{fontenc}
\usepackage[utf8]{inputenc}

\usepackage{etex}

\usepackage{arydshln} 
\usepackage{url}
\usepackage{ifthen}
\usepackage{bm}
\usepackage{multirow}
\usepackage[dvips]{graphicx}
\usepackage[usenames]{color}
\usepackage{xcolor,colortbl} 
\usepackage{threeparttable}
\usepackage{comment}
\usepackage{paralist,verbatim}
\usepackage{cases}
\usepackage{booktabs}
\usepackage{braket}
\usepackage{cancel} 
\usepackage{ascmac} 
\usepackage{framed}
\usepackage{authblk}
\usepackage{pifont}
\usepackage{qcircuit}
\definecolor{darkblue}{rgb}{0,0,0.6}
\definecolor{darkgreen}{rgb}{0,0.5,0}
\definecolor{maroon}{rgb}{0.5,0.1,0.1}
\definecolor{dpurple}{rgb}{0.2,0,0.65}

\usepackage[capitalise,noabbrev]{cleveref}
\usepackage[absolute]{textpos}
\usepackage[final]{microtype}
\usepackage[absolute]{textpos}
\usepackage{everypage}
\DeclareMathAlphabet{\mathpzc}{OT1}{pzc}{m}{it}

\usepackage{algorithmic}
\usepackage{algorithm}
\usepackage{here}

\newtheoremstyle{thicktheorem}%
{\topsep}
{\topsep}
{\itshape}{}%
{\bfseries}%
{.}
{ }%
{\thmname{#1}\thmnumber{ #2}%
		\thmnote{ (#3)}%
}

\newtheoremstyle{remark}
{\topsep}
{\topsep}
	{}
	{}
	{}
	{.}
	{ }
	{\textit{\thmname{#1}}\thmnumber{ #2}
			\thmnote{ (#3)}%
	}

\ifnum\llncs=0
	\theoremstyle{thicktheorem}
	\newtheorem{theorem}{Theorem}[section]
	\newtheorem{lemma}[theorem]{Lemma}

	\newtheorem{definition}[theorem]{Definition}

	\theoremstyle{remark}
	
	\newtheorem{remark}[theorem]{Remark}

\else
\fi

\Crefname{MyClaim}{Claim}{Claims}

	\crefname{theorem}{Theorem}{Theorems}
	\crefname{assumption}{Assumption}{Assumptions}
	\crefname{construction}{Construction}{Constructions}
	\crefname{corollary}{Corollary}{Corollaries}
	\crefname{conjecture}{Conjecture}{Conjectures}
	\crefname{definition}{Definition}{Definitions}
	\crefname{exmaple}{Example}{Examples}
	\crefname{experiment}{Experiment}{Experiments}
	\crefname{counterexample}{Counterexample}{Counterexamples}
	\crefname{lemma}{Lemma}{Lemmata}
	\crefname{observation}{Observation}{Observations}
	\crefname{proposition}{Proposition}{Propositions}
	\crefname{remark}{Remark}{Remarks}
	\crefname{claim}{Claim}{Claims}
	\crefname{fact}{Fact}{Facts}
	\crefname{note}{Note}{Notes}

\ifnum\llncs=1
 \crefname{appendix}{App.}{Appendices}
 \crefname{section}{Sec.}{Sections}
\else
\fi

\ifnum\llncs=1
\renewcommand*{\backref}[1]{}
\else
	\renewcommand*{\backref}[1]{(Cited on page~#1.)}
	\ifnum\notxfont=1
	\else
		\usepackage{newtxtext}
	\fi
\fi

%% file: macros.tex
\ifnum\authnote=0  
\newcommand{\mor}[1]{}
\newcommand{\ryo}[1]{}
\newcommand{\takashi}[1]{}
\newcommand{\fuyuki}[1]{}

\else
\newcommand{\mor}[1]{$\ll$\textsf{\color{red} Tomoyuki: { #1}}$\gg$}
\newcommand{\takashi}[1]{$\ll$\textsf{\color{orange} Takashi: { #1}}$\gg$}
\newcommand{\ryo}[1]{$\ll$\textsf{\color{darkgreen} Ryo: { #1}}$\gg$}
\newcommand{\fuyuki}[1]{$\ll$\textsf{\color{darkblue} Fuyuki: { #1}}$\gg$}

\fi

\newcommand{\calY}{\mathcal{Y}}
\newcommand{\calX}{\mathcal{X}}

\newcommand{\la}{\leftarrow}

\newcommand{\seteq}{\coloneqq}

\newcommand{\concat}{\|}

\newcommand{\abs}[1]{|#1|}

\newcommand{\cA}{\mathcal{A}}
\newcommand{\cB}{\mathcal{B}}
\newcommand{\cC}{\mathcal{C}}

\newcommand{\cX}{\mathcal{X}}
\newcommand{\cY}{\mathcal{Y}}

\def\makeuppercase#1{
\expandafter\newcommand\csname tl#1\endcsname{\widetilde{#1}}
}

\def\makelowercase#1{
\expandafter\newcommand\csname tl#1\endcsname{\widetilde{#1}}
}

\newcommand{\regC}{\mathbf{C}}
\newcommand{\regE}{\mathbf{E}}
\newcommand{\regR}{\mathbf{R}}

\newcommand{\regD}{\mathbf{D}}

\newcommand{\regB}{\mathbf{B}}
\newcommand{\regA}{\mathbf{A}}

\newcommand{\secp}{\lambda}
\newcommand{\sig}{\sigma}

\newcommand{\A}{\entity{A}}
\newcommand{\B}{\entity{B}}

\newcommand{\CCA}{\textrm{CCA}}

\newcommand{\hybi}[1]{\mathsf{Hyb}_{#1}}

\newcommand*{\sk}{\keys{sk}}
\newcommand*{\pk}{\keys{pk}}

\newcommand*{\rk}{\keys{rk}}

\newcommand{\ct}{\keys{ct}}

\newcommand*{\vk}{\keys{vk}}

\newcommand*{\msg}{\keys{msg}}

\newcommand*{\keys}[1]{\mathsf{#1}}

\newcommand*{\algo}[1]{\ensuremath{\mathsf{#1}}}

\newcommand*{\entity}[1]{\mathcal{#1}}

\newenvironment{boxfig}[2]{\begin{figure}[#1]\fbox{\begin{minipage}{0.97\linewidth}
                        \vspace{0.2em}
                        \makebox[0.025\linewidth]{}
                        \begin{minipage}{0.95\linewidth}
            {{
                        #2 }}
                        \end{minipage}
                        \vspace{0.2em}
                        \end{minipage}}}{\end{figure}}

\newcommand{\pprotocol}[4]{
\begin{boxfig}{h}{\footnotesize 
\centering{\textbf{#1}}
    #4
\vspace{0.2em} } \caption{\label{#3} #2}
\end{boxfig}
}

\newcommand{\protocol}[4]{
\pprotocol{#1}{#2}{#3}{#4} }

\newcommand{\bit}{\{0,1\}}

\newcommand{\Gen}{\algo{Gen}}

\newcommand{\SKGen}{\algo{SKGen}}
\newcommand{\PKGen}{\algo{PKGen}}

\newcommand{\Enc}{\algo{Enc}}
\newcommand{\rEnc}{\algo{rEnc}}
\newcommand{\Dec}{\algo{Dec}}

\newcommand{\Sign}{\algo{Sign}}

\newcommand{\sign}{\algo{Sign}}

\newcommand{\Ver}{\algo{Ver}}

\newcommand{\st}{\algo{st}}

\newcommand{\PRF}{\algo{PRF}}

\newcommand\SKE{\algo{SKE}}

\newcommand{\ske}{\mathsf{ske}}

\newcommand{\QPKE}{\algo{QPKE}}

\newcommand\SIG{\algo{SIG}}

\newcommand{\negl}{{\mathsf{negl}}}

\newcommand{\poly}{{\mathrm{poly}}}

\makeatletter
\DeclareRobustCommand
  \myvdots{\vbox{\baselineskip4\p@ \lineskiplimit\z@
    \hbox{.}\hbox{.}\hbox{.}}}
\makeatother

%% file: introduction.tex
\section{Introduction}
\label{sec:introduction}

\subsection{Background}
Quantum physics provides several advantages in cryptography.
For instance, statistically-secure key exchange, which is impossible in classical cryptography, becomes possible
if quantum states are transmitted~\cite{BB84}. 
Additionally, oblivious transfers and multiparty computations are possible only from one-way functions (OWFs) in the quantum world~\cite{C:BCKM21b,EC:GLSV21}. Those cryptographic primitives are believed to require stronger structured assumptions in classical cryptography~\cite{STOC:ImpRud89,FOCS:GKMRV00}.
Furthermore, it has been shown that several cryptographic tasks, such as (non-interactive) commitments, digital signatures, secret-key encryption, quantum money, and multiparty computations, are possible based on new primitives such as pseudorandom states generators, pseudorandom function-like states generators, one-way states generators, and EFI, which seem to be weaker than OWFs~\cite{C:JiLiuSon18,Kre21,C:MorYam22,C:AnaQiaYue22,ITCS:BraCanQia23,TCC:AGQY22,EPRINT:CaoXue22b,EPRINT:MorYam22c,STOC:KQST23}. 

\paragraph{\bf Quantum public key encryption from OWFs.}
Despite these developments, it is still an open problem whether public-key encryption (PKE) is possible with 
only OWFs (or the above weaker primitives) in the quantum world. 
PKE from OWFs is impossible (in a black-box way) in the classical cryptography~\cite{C:ImpRud88}. However, it could be possible if quantum states are transmitted or local operations are quantum.
In fact, some recent works~\cite{EPRINT:MorYam22c,cryptoeprint:2023/282,TCC:BGHMSVW23} independently constructed quantum PKE (QPKE) with quantum public keys based on OWFs or pseudorandom function-like states generators.
However, the constructions proposed in those works have a huge drawback as explained below, and thus we still do not have a satisfactory solution to the problem of ``QPKE from OWFs''.

\paragraph{\bf How to certify quantum public keys?}
When we study public key cryptographic primitives, we have to care about how to certify the public keys, that is, how a sender (who encrypts messages) can check if a given public key is a valid public key under which the secrecy of the encrypted messages is guaranteed.
When the public keys are classical strings, we can easily certify them using digital signature schemes.
However, in the case where the public keys are quantum states, we cannot use digital signature schemes to achieve this goal in general\footnote{
There is a general impossibility result for signing quantum states~\cite{signingquantumstates}.
}, and it is unclear how to certify them.

As stated above, some recent works~\cite{EPRINT:MorYam22c,cryptoeprint:2023/282,TCC:BGHMSVW23} realized QPKE with quantum public keys from OWFs or even weaker assumptions. However, those works did not tackle this quantum public key certification problem very much.
In fact, as far as we understand, to use the primitives proposed in those works meaningfully, we need to use secure quantum channels to transmit the quantum public keys so that a sender can use an intact quantum public key.
This is a huge drawback since the goal of PKE is to transmit a message \emph{without assuming secure channels}.
If the sender can establish a secure channel to obtain the quantum public key, the sender could use it to transmit the message in the first place, and there is no advantage to using the PKE scheme.

\paragraph{\bf QPKE with tamper-resilient quantum public keys.}
One of our goals in this work is to solve this issue and develop a more reasonable notion of QPKE with quantum public keys.
Especially, we consider the setting with the following two natural conditions.
First, we assume that every quantum state (that is, quantum public keys in this work) is sent via an unauthenticated channel, and thus it can be tampered with by an adversary.
If we do not assume secure quantum channels, we have to take such a tampering attack into account since authentication generally requires secrecy for quantum channels~\cite{FOCS:BCGST02}.
Second, we assume that every classical string is sent via an authenticated channel.
This is the same assumption in classical PKE and can be achieved using digital signatures.
Note that the security of the constructions proposed in the above works~\cite{EPRINT:MorYam22c,cryptoeprint:2023/282,TCC:BGHMSVW23} is broken in this natural setting.
In this work, we tackle whether we can realize QPKE with quantum public keys that provides a security guarantee in this natural setting, especially from OWFs.

%

\subsection{Our Results}\label{sec:our_results}
We affirmatively answer the above question. We realize the first QPKE scheme based only on OWFs that achieves the goal of classical PKE, which is to establish secure communication over insecure channels.
We define the notions of QPKE that can be used in the above setting where unauthenticated quantum channels and classical authenticated channels are available.
Then, we propose constructions satisfying the definitions from OWFs.
Below, we state each result in detail.

\paragraph{\bf Definitional work.}
We redefine the syntax of QPKE.
The difference from the previous definitions is that the key generation algorithm outputs a classical verification key together with the secret key.
Also, this verification key is given to the encryption algorithm with a quantum public key and a message so that the encryption algorithm can check the validity of the given quantum public key.
We require ciphertexts to be classical.\footnote{We could also consider QPKE schemes with quantum ciphertexts if we only consider IND-pkT-CPA security. However, it is unclear how we should define IND-pkT-CCA security for such schemes because the decryption oracle cannot check if a given ciphertext is equivalent to the challenge ciphertext. Thus, we focus on schemes with classical ciphertexts in this paper.}
We require a QPKE scheme to satisfy the following two basic security notions.
\begin{itemize}
\item \textbf{Indistinguishability against public key tempering chosen plaintext attacks (IND-pkT-CPA security).}
Roughly speaking, it guarantees that indistinguishability holds even if messages are encrypted by a public key tampered with by an adversary.
More specifically, it guarantees that no efficient adversary can guess the challenge bit $b$ with a probability significantly better than random guessing given $\Enc(\vk,\pk^\prime,\msg_b)$, where $\vk$ is the correct verification key and $(\pk^\prime,\msg_0,\msg_1)$ are generated by the adversary who is given the verification key $\vk$ and multiple copies of the correctly generated quantum public keys.\footnote{The tampered quantum public key $\pk'$ can
be entangled with the adversary's internal state.}
IND-pkT-CPA security captures the setting where the classical verification key is sent via a classical authenticated channel. Thus, everyone can obtain the correct verification key. However, a quantum public key is sent via an unauthenticated quantum channel and thus can be tampered with by an adversary.
\item \textbf{Decryption error detectability.}  
In our setting, an adversary may try to cause a decryption error by tampering with the quantum public key.
To address this issue, we introduce a security notion that we call \emph{decryption error detectability}.
It roughly guarantees that a legitimate receiver of a ciphertext can notice if the decrypted message is different from the message intended by the sender.
\end{itemize}
IND-pkT-CPA security considers adversaries that may tamper with quantum public keys but only passively observe ciphertext. For classical PKE, the golden standard security notion is indistinguishability against chosen ciphertext attacks (IND-CCA security) that considers active adversaries that may see decryption results of any (possibly malformed) ciphertexts. Thus, we also define its analog for QPKE.  In \Cref{sec:discussion}, we discuss its importance in a natural application scenario.  
\begin{itemize} 
\item \textbf{Indistinguishability against public key tempering chosen ciphertext attacks (IND-pkT-CCA security).}
This is similar to IND-pkT-CPA security except that the adversary is given access to the decryption oracle that returns a decryption result on any ciphertext other than the challenge ciphertext.\footnote{Recall that ciphertexts are classical in our definition, and thus this is well-defined.} Moreover, we allow the adversary to learn one-bit information indicating if the challenge ciphertext is decrypted to $\bot$ or not. We note that it is redundant for classical PKE since the challenge ciphertext is always decrypted to the challenge message, which is not $\bot$, by decryption correctness. On the other hand, it may give more power to the adversary for QPKE since if the adversary tempers with the public key that is used to generate the challenge ciphertext, decryption correctness may no longer hold.  
\end{itemize} 

\paragraph{\bf IND-pkT-CPA secure construction from OWFs.}
We propose a QPKE scheme satisfying IND-pkT-CPA security from a digital signature scheme that can be constructed from OWFs.
Our construction is inspired by the duality between distinguishing and swapping shown by Aaronson, Atia, and Susskind~\cite{AAS20} and its cryptographic applications by Hhan, Morimae, and Yamakawa~\cite{EC:HhaMorYam23}.
Our construction has quantum public keys and classical ciphertexts.
We also propose a general transformation that adds decryption error detectability.
The transformation uses only a digital signature scheme.

\paragraph{\bf Upgrading to IND-pkT-CCA security.} 
We show a generic compiler that upgrades IND-pkT-CPA security into IND-pkT-CCA security while preserving decryption error detectability only using OWFs. 
It is worth mentioning that constructing such a generic CPA-to-CCA compiler is a long-standing open problem for classical PKE, and thus we make crucial use of the fact that public keys are quantum for constructing our compiler.     
By plugging our IND-pkT-CPA secure construction into the compiler, we obtain a QPKE scheme that satisfies IND-pkT-CCA security and decryption error detectability only based on OWFs.

\paragraph{\bf Recyclable variant.}
Our above definitions for QPKE assume each quantum public key is used to encrypt only a single message and might be consumed. 
We also introduce a notion of \emph{recyclable QPKE} where the encryption algorithm given a quantum public key outputs a ciphertext together with a classical state that can be used to encrypt a message many times. 
Then, we show that any standard IND-pkT-CPA (resp. IND-pkT-CCA) secure QPKE scheme with classical ciphertexts can be transformed into a recyclable one with IND-pkT-CPA (resp. IND-pkT-CCA) security while preserving decryption error detectability.
The transformation uses only a CPA (resp. CCA) secure classical symmetric key encryption scheme that is implied by OWFs.
Thus, by combining the transformation with the above results, we obtain a recyclable IND-pkT-CCA QPKE scheme with decryption error detectability from OWFs.

\subsection{Discussion}\label{sec:discussion}
\paragraph{\bf Pure State Public Keys vs. Mixed State Public Keys.} 
The quantum public keys of our QPKE schemes are mixed states.
Some recent works~\cite{cryptoeprint:2023/282,TCC:BGHMSVW23}  that studied QPKE explicitly require that a quantum public key of QPKE be a pure quantum state.
The reason is related to the quantum public key certification problem, which is this work's  main focus.
Barooti et al.~\cite{TCC:BGHMSVW23} claimed that a sender can check the validity of given quantum public keys by using SWAP test if they are pure states, but not if they are mixed states.
However, as far as we understand, this claim implicitly requires that at least one intact quantum public key be transmitted via secure quantum channels where an adversary cannot touch it at all\footnote{
More precisely, their model seems to require a physical setup assumption that enables a sender to obtain at least one intact quantum public key, such as secure quantum channels or tamper-proof quantum hardware.
}, which is an unsatisfactory assumption that makes QPKE less valuable.
It is unclear how a sender can check the validity of a given quantum public key in the constructions proposed in \cite{TCC:BGHMSVW23} without assuming such secure transmission of intact quantum public keys.

We believe that it is not important whether the quantum public keys are pure states or mixed states, and what is really important is whether a sender can check the validity of given quantum public keys without assuming unsatisfactory setups such as quantum secure channels.
Although our QPKE schemes have mixed state quantum public keys, they provide such a validity checking of quantum public keys by a sender without assuming any unsatisfactory setups.
\ifnum\pure=0
\else
In addition, we can easily extend our construction into one with pure state quantum public keys. 
We provide the variant in  
\ifnum\cameraready=1
the full version. 
\else
\cref{sec:pure_pk_qpke}.
\fi
\fi

\input{related_work}

%% file: related_work.tex
\subsection{Related Works} 
The possibility that QPKE can be achieved from weaker assumptions was first pointed out by Gottesman~\cite{GottesmanPKE}, though he did not give any concrete construction.
The first concrete construction of QPKE was proposed by Kawachi, Koshiba, Nishimura, and Yamakami~\cite{EC:KKNY05}.
They formally defined the notion of QPKE with quantum public keys, and provided a construction satisfying it from a distinguishing problem of two quantum states.
Recently, Morimae and Yamakawa~\cite{EPRINT:MorYam22c} pointed out that QPKE defined by~\cite{EC:KKNY05} can be achieved from any classical or quantum symmetric key encryption almost trivially.
The constructions proposed in these two works have mixed state quantum public keys.
Then, subsequent works~\cite{cryptoeprint:2023/282,cryptoeprint:BGHMSVW23} independently studied the question whether QPKE with pure state quantum public keys can be constructed from OWFs or even weaker assumptions.

The definition of QPKE studied in the above works essentially assume that a sender can obtain intact quantum public keys.
As far as we understand, this requires unsatisfactory physical setup assumptions such as secure quantum channels or tamper-proof quantum hardware, regardless of whether the quantum public keys are pure states or mixed states.
In our natural setting where an adversary can touch the quantum channel where quantum public keys are sent, the adversary can easily attack the previous constructions by simply replacing the quantum public key on the channel with the one generated by itself that the adversary knows the corresponding secret key.
We need to take such adversarial behavior into consideration, unless we assume physical setup assumptions that deliver intact quantum public keys to the sender.
Our work is the first one that proposes a QPKE scheme secure in this natural setting assuming only classical authenticated channels that is the same assumption as classical PKE and can be implemented by digital signature schemes.
It is unclear if we could solve the problem in the previous constructions by using classical authenticated channels similarly to our work.
Below, we review the constructions of QPKE from OWFs proposed in the recent works.

\ifnum\cameraready=0
The construction by Morimae and Yamakawa~\cite{EPRINT:MorYam22c} is highly simple.
A (mixed state) public key of their construction is of the form $(\ct_0,\ct_1)$, where $\ct_b$ is an encryption of $b$ by a symmetric key encryption scheme.
The encryption algorithm with input message $b$ simply outputs $\ct_b$.

Coladangelo~\cite{cryptoeprint:2023/282} constructed a QPKE scheme with
quantum public keys and quantum ciphertexts from pseudorandom functions (PRFs), which are constructed from OWFs.
The public key is 
\begin{equation}
\ket{\pk}\coloneqq\sum_y(-1)^{\PRF_k(y)}\ket{y},
\end{equation}
and the secret key is $k$.
The ciphertext for the plaintext $m$ is
\begin{equation}
(Z^x\ket{\pk}=\sum_y(-1)^{x\cdot y+\PRF_k(y)}\ket{y},
r,r\cdot x\oplus m),
\end{equation}
where $r$ is chosen uniformly at random.

 Barooti, Grilo, Huguenin-Dumittan, Malavolta, Sattath, Vu, and Walter~\cite{cryptoeprint:BGHMSVW23} constructed three QPKE schemes: (1) CCA secure QPKE with
quantum public keys and classical ciphertexts from OWFs (2) CCA1\footnote{Afther the adversary received a challenge ciphertext, they cannot access the decryption oracle.} secure QPKE with quantum public keys and ciphertexts from pseudorandom function-like states generators, (3) CPA secure QPKE with quantum public keys and classical ciphertexts from pseudo-random function-like states with proof of destruction.
All constructions considers security under the encryption oracle.
We review their construction based on OWFs.

Their construction is hybrid encryption of CPA secure QPKE (the KEM part) and CCA secure classical symmetric key encryption (the DEM part).
The public key is
\begin{equation}
\ket{\pk}\coloneqq   \sum_x\ket{x}\ket{\PRF_k(x)},
\end{equation}
and the secret key is $k$.
The encryption algorithm first measures $\ket{\pk}$ in the computational basis to get
$(x,\PRF_k(x))$ and outputs $(x,\mathsf{SKE}.\Enc(\PRF_k(x),m))$
as the ciphertext for the plaintext $m$, where $\mathsf{SKE}.\Enc$
is the encryption algorithm of a symmetric key encryption scheme.
\fi

We finally compare Quantum Key Distribution (QKD)~\cite{BB84} with our notion of QPKE. QKD also enables us to establish secure communication over an untrusted quantum channel assuming that an authenticated classical channel is available similarly to our QPKE. An advantage of QKD is that it is information theoretically secure and does not need any computational assumption. 
On the other hand, it has disadvantages that it must be interactive and parties must record secret information for each session. Thus, it is incomparable to the notion of QPKE.

\subsection{Concurrent Work}
A concurrent and independent work by Malavolta and Walter~\cite{NewMW23} constructs a two-round quantum key exchange protocol from OWFs. Their underlying idea is similar to our IND-pkT-CPA secure construction. Indeed, the technical core of their work is a construction a QPKE scheme that is secure against adversaries that only see one copy of the quantum public key. A nice feature of their scheme is that it satisfies everlasting security. That is, as long as the adversary is quantum polynomial-time when tampering with the public key, it cannot recover any information of the encrypted message even if it has an unbounded computational power later. They also show how to extend the scheme to satisfy security in the many-copy setting at the cost of sacrificing everlasting security. This gives an alternative construction of IND-pkT-CPA secure QPKE scheme from OWFs using our terminology.
On the other hand, they do not consider CCA security, and  
our CPA-to-CCA compiler is unique to this work.

\subsection{Open Problems}
In our construction, public keys are quantum states. It is an open problem whether
QPKE with classical public keys are possible from OWFs.
Another interesting open problem is whether we can construct QPKE defined in this work from an even weaker assumption than OWFs such as pseudorandom states generators.


In our model of QPKE, a decryption error 
may be caused by tampering attacks on the quantum public key.
To address this issue, we introduce the security notion we call decryption error detectability that guarantees that a legitimate receiver of a ciphertext can notice if the decrypted message is different from the message intended by the sender.
We could consider even stronger variant of decryption error detectability that requires that a sender can notice if a given quantum public key does not provide decryption correctness.
It is an open problem to construct a QPKE scheme satisfying such a stronger decryption error detectability.

The notion of IND-pkT-CCA security is defined with respect to a classical decryption oracle. In fact, this is inherent for our proof technique. We leave it open to construct a tamper-resilient QPKE scheme that resists attacks with a quantumly-accessible decryption oracle.

%% file: technical_overview.tex
\newcommand{\INDpkTA}{\textrm{IND-pkT-CPA}}
\newcommand{\INDCVApkTA}{\textrm{IND-pkT-CVA}}
\newcommand{\INDCCApkTA}{\textrm{IND-pkT-CCA}}
\newcommand{\INDoneCCApkTA}{\textrm{IND-pkT-1CCA}}
\newcommand{\INDpkTCPA}{\textrm{IND-pkT-CPA}}
\newcommand{\INDpkTCVA}{\textrm{IND-pkT-CVA}}
\newcommand{\INDpkTCCA}{\textrm{IND-pkT-CCA}}
\newcommand{\INDpkToneCCA}{\textrm{IND-pkT-1CCA}}
\newcommand{\cv}{\mathsf{cv}}
\newcommand{\onekey}[1]{#1^{(1)}}

\newcommand{\CVA}{\mathsf{CVA}}
\newcommand{\oneCCA}{\mathsf{1CCA}}
\renewcommand{\CCA}{\mathsf{CCA}}
\newcommand{\onecca}{\mathsf{1cca}}
\newcommand{\cva}{\mathsf{cva}}
\newcommand{\Test}{\mathsf{Test}}
\newcommand{\sigvklen}{n}
\newcommand{\sigvk}{\mathsf{sigvk}}
\newcommand{\sigk}{\mathsf{sigk}}
\newcommand{\Good}{\mathtt{Good}}
\newcommand{\Forge}{\mathtt{Forge}}
\newcommand{\DecError}{\mathtt{DecError}}

\newcommand{\TMAC}{\mathsf{TMAC}}
\newcommand{\tmac}{\mathsf{tmac}}
\newcommand{\TKGen}{\mathsf{TKGen}}
\newcommand{\token}{\mathsf{token}}
\newcommand{\mk}{\mathsf{mk}}

\renewcommand{\sig}{\mathsf{sig}}

\newcommand{\Onekey}{\mathsf{1Key}}
\newcommand{\okey}{\mathsf{1key}}
\newcommand{\Mkey}{\mathsf{MKey}}
\newcommand{\snum}{\mathsf{snum}}
\newcommand{\keylist}{\mathsf{KL}}

\section{Technical Overview}
\ifnum\llncs=0
We provide a technical overview of our work.
\fi
\subsection{Definition of QPKE}
\paragraph{\bf Syntax.}
We define QPKE that can be used in the setting where quantum unauthenticated channels and classical authenticated channels are available.
To this end, we introduce the following two modifications to the previous definitions.
\begin{itemize}
\item The secret key generation algorithm outputs a classical verification key together with the secret key.
\item The verification key is given to the encryption algorithm together with a quantum public key and a message so that the encryption algorithm can check the validity of the given quantum public key.
\end{itemize}
Concretely, in our definition, a QPKE scheme consists of four algorithms $(\SKGen,\PKGen,\allowbreak \Enc,\Dec)$.
$\SKGen$ is a classical secret key generation algorithm that is given the security parameter and outputs a classical secret key $\sk$ and a classical verification key $\vk$.
$\PKGen$ is a quantum public key generation algorithm that takes as input the classical secret key $\sk$ and outputs a quantum public key $\pk$.
$\Enc$ is a quantum encryption algorithm that takes as inputs the classical verification key $\vk$, a quantum public key $\pk$, and a plaintext $\msg$, and outputs a classical ciphertext $\ct$.
Finally, $\Dec$ is a classical decryption algorithm that takes as input the classical secret key and a ciphertext, and outputs the decryption result.

The above definitions for QPKE assume each quantum public key is used to encrypt only a single message and might be consumed. 
We also introduce a notion of recyclable QPKE where the encryption algorithm given a quantum public key outputs a ciphertext together with a classical state that can be used to encrypt a message many times.
In this overview, we mainly focus on non-recyclable QPKE for simplicity.

\paragraph{\bf IND-pkT-CPA security.}
IND-pkT-CPA security roughly guarantees that indistinguishability holds even if messages are encrypted by a public key $\pk^\prime$ tampered with by an adversary as long as the encryption is done with the correct verification key $\vk$.
Formally, IND-pkT-CPA security is defined using the following security experiment played by an adversary $\cA$.
The experiment first generates classical secret key and verification key pair $(\sk,\vk)\la\SKGen(1^\secp)$ and $m$ copies of the quantum public key $\pk_1,\ldots,\pk_m\la\PKGen(\sk)$.
Then, $\cA$ is given the classical verification key $\vk$ and $m$ quantum public keys $\pk_1,\ldots,\pk_m$, and outputs a tampered quantum public key $\pk^\prime$ and a pair of challenge plaintexts $(\msg_0,\msg_1)$.
The experiment generates the challenge ciphertext using the adversarially generated quantum public key, that is, $\ct^*\la\Enc(\vk,\pk^\prime,\msg_b)$, where $b\la\bit$.
Finally, $\cA$ is given $\ct^*$ and outputs the guess for $b$.
IND-pkT-CPA security guarantees that any efficient quantum adversary cannot guess $b$ significantly better than random guessing in this experiment.

IND-pkT-CPA security captures the setting 
where the classical verification key is sent via a classical authenticated channel and thus everyone can obtain correct verification key, but a quantum public key is sent via an unauthenticated quantum channel and thus can be tampered with by an adversary.
Especially, it captures an adversary $\cA$ who steals a quantum public key $\pk$ sent to a user, replace it with a tampered one $\pk^\prime$, and try to break the secrecy of a message encrypted by $\pk^\prime$.

To capture wide range of usage scenarios, we give multiple copies of the quantum public keys $\pk_1,...,\pk_m$ to $\cA$.
We also consider a relaxed security notion where an adversary is given a single quantum public key and denote it as $\onekey{\INDpkTA}$.

\paragraph{\bf Decryption error detectability.}
We also define a security notion related to the correctness notion that we call decryption error detectability.
It roughly guarantees that a legitimate receiver of a ciphertext can notice if the decrypted message is different from the message intended by the sender.
Such a decryption error could occur frequently in our setting as a result of the tampering attacks on the quantum public key sent via an unauthenticated quantum channel.
Note that our definition of QPKE requires a ciphertext of QPKE to be a classical string and we assume all classical information is sent through a classical authenticated channel.
Thus, similarly to the verification key, we can assume that ciphertexts can be sent without being tampered.

\subsection{IND-pkT-CPA Secure Construction}
We provide the technical overview for IND-pkT-CPA secure construction.

\paragraph{\bf Duality between distinguishing and swapping.} 
Our construction is inspired by the duality between distinguishing and swapping shown by Aaronson, Atia, and Susskind~\cite{AAS20} and its cryptographic applications by Hhan, Morimae, and Yamakawa~\cite{EC:HhaMorYam23}.\footnote{In the main body, we do not explicitly use any result of \cite{AAS20,EC:HhaMorYam23} though our analysis is similar to theirs.} We first review their idea. Let $\ket{\psi}$ and $\ket{\phi}$ be orthogonal states. \cite{AAS20} showed that $\ket{\psi}+\ket{\phi}$ and $\ket{\psi}-\ket{\phi}$ are computationally indistinguishable\footnote{We often omit normalization factors.} if and only if one cannot efficiently ``swap'' $\ket{\psi}$ and $\ket{\phi}$ with a non-negligible advantage, i.e., for any efficiently computable unitary $U$, $|\bra{\phi}U\ket{\psi}+\bra{\psi}U\ket{\phi}|$ is negligible. 
Based on the above result, \cite{EC:HhaMorYam23} suggested to use $\ket{\psi}+(-1)^b \ket{\phi}$ as an encryption of a plaintext $b\in \bit$. By the result of \cite{AAS20}, its security is reduced to the hardness of swapping  $\ket{\psi}$ and $\ket{\phi}$. 

\paragraph{\bf Basic one-time SKE.}
We can construct one-time SKE scheme with quantum ciphertext using the above duality between distinguishing and swapping as follows.
A secret decryption key is $(x_0,x_1)$ for uniformly random bit strings $x_0,x_1\in \bit^\secp$, and the corresponding secret encryption key is
\begin{equation}
    \ket{0}\ket{x_0}+\ket{1}\ket{x_1}.
\end{equation}
Then, when encrypting a plaintext $b\in\bit$, it transforms the secret encryption key into the ciphertext
\begin{equation}
    \ket{0}\ket{x_0}+(-1)^b \ket{1}\ket{x_1}.
\end{equation}

One-time indistinguishability of this scheme is somewhat obvious because the adversary has no information of $x_0$ or $x_1$ besides the ciphertext, but let us analyze it using the idea of \cite{AAS20} to get more insights. Suppose that the above scheme is insecure, i.e., $\ket{0}\ket{x_0}+\ket{1}\ket{x_1}$ and $\ket{0}\ket{x_0}-\ket{1}\ket{x_1}$ are computationally distinguishable with a non-negligible advantage. Then, by the result of \cite{AAS20}, there is an efficient unitary $U$ that swaps  $\ket{0}\ket{x_0}$ and $\ket{1}\ket{x_1}$ with a non-negligible advantage. By using this unitary, let us  consider the following procedure:
\begin{enumerate}
    \item Given a state $\ket{0}\ket{x_0}\pm  \ket{1}\ket{x_1}$, measure it in the computational basis to get $\ket{\alpha}\ket{x_{\alpha}}$ for random $\alpha\in \bit$.
    \item Apply the unitary $U$ to $\ket{\alpha}\ket{x_{\alpha}}$ and measure it in the computational basis.
\end{enumerate}
Since $U$ swaps  $\ket{0}\ket{x_0}$ and $\ket{1}\ket{x_1}$ with a non-negligible advantage, the probability that the outcome of the second measurement is $\ket{\alpha\oplus 1}\ket{x_{\alpha\oplus1}}$ is non-negligible. 
This yields the following observation: If one can efficiently distinguish $\ket{0}\ket{x_0}+\ket{1}\ket{x_1}$ and $\ket{0}\ket{x_0}-\ket{1}\ket{x_1}$, then one can efficiently compute both $x_0$ and $x_1$ from $\ket{0}\ket{x_0}\pm  \ket{1}\ket{x_1}$. On the other hand, it is easy to show that one cannot compute both $x_0$ and $x_1$ from $\ket{0}\ket{x_0}\pm \ket{1}\ket{x_1}$ with a non-negligible probability by a simple information theoretical argument. Thus, the above argument implies one-time indistinguishability of the above construction.

\paragraph{\bf Extension to $\onekey{\INDpkTA}$ secure QPKE with quantum ciphertext.}
We show how to extend the above SKE scheme into an $\onekey{\INDpkTA}$ secure QPKE scheme with quantum ciphertext.
One natural approach is to use the secret encryption key $\ket{0}\ket{x_0}+\ket{1}\ket{x_1}$ as a quantum public key.
However, it does not work since the adversary for $\onekey{\INDpkTA}$ who is given $\ket{0}\ket{x_0}+\ket{1}\ket{x_1}$ as the public key can replace it with $\ket{0}\ket{x'_0}+\ket{1}\ket{x'_1}$ for $x'_0,x'_1$ of its choice. 
To fix this issue, we partially authenticate a quantum public key by using classical digital signatures. 
Concretely, the secret key generation algorithm $\SKGen$ generates a signing key and verification key pair $(\sk,\vk)$ of a digital signature scheme, and use them as the secret key and verification key of the QPKE scheme.
Then, public key generation algorithm $\PKGen$ takes as input $\sk$ and outputs a quantum public key
\begin{equation}
    \ket{0}\ket{\sigma(0)}+\ket{1}{\ket{\sigma(1)}}, 
\end{equation}
where $\sigma(\alpha)$ is a signature for $\alpha\in \bit$ by the signing key $\sk$. 
Here, we assume that the signature scheme has a deterministic signing algorithm.
The encryption algorithm $\Enc$ that is given $\vk$, a quantum public key, and a plaintext $b\in\bit$ first coherently verifies using $\vk$ the validity of the signatures in the second register of the public key and aborts if the verification rejects.
Otherwise, $\Enc$ generates the ciphertext by encoding the plaintext $b$ into the phase of the quantum public key as before.

The $\onekey{\INDpkTA}$ security of the construction is analyzed as follows.
We assume that the digital signature scheme satisfies strong unforgeability, i.e., given message-signature pairs $(\msg_1,\sigma_1),...,(\msg_n,\sigma_n)$, no efficient adversary can output $(\msg^\ast,\sigma^\ast)$ such that $(\msg^\ast,\sigma^\ast)\neq (\msg_i,\sigma_i)$ for all $i\in [n]$.\footnote{At this point, two-time security (where $n=2$) suffices but we finally need to allow $n$ to be an arbitrary polynomial.}
Then, no matter how the adversary who is given a single correctly generated quantum public key tampers with it, if it passes the verification in $\Enc$, the state after passing the verification is negligibly close to a state in the form of 
\begin{equation}
c_0\ket{0}\ket{\sigma(0)}\ket{\Psi_0}
+c_1\ket{1}\ket{\sigma(1)}\ket{\Psi_1}
\label{overview_pk}
\end{equation}
with some complex coefficients $c_0$ and $c_1$,
and some states $\ket{\Psi_0}$ and $\ket{\Psi_1}$ over the adversary's register (except for a negligible probability).
The encryption of a plaintext $b\in\bit$ is to apply $Z^b$ on the first qubit of \cref{overview_pk}.
The cipertext generated under the tampered public key is therefore
\begin{equation}
c_0\ket{0}\ket{\sigma(0)}\ket{\Psi_0}
+(-1)^bc_1\ket{1}\ket{\sigma(1)}\ket{\Psi_1}.
\end{equation}
By a slight extension of the analysis of the above SKE scheme, we show that if one can efficiently distinguish $c_0\ket{0}\ket{\sigma(0)}\ket{\Psi_0}
+c_1\ket{1}\ket{\sigma(1)}\ket{\Psi_1}$ and $c_0\ket{0}\ket{\sigma(0)}\ket{\Psi_0}
-c_1\ket{1}\ket{\sigma(1)}\ket{\Psi_1}$, then one can efficiently compute both $\sigma(0)$ and $\sigma(1)$. On the other hand, recall that the adversary is only given one copy of the public key $\ket{0}\ket{\sigma(0)}+\ket{1}\ket{\sigma(1)}$. We can show that it is impossible to compute both $\sigma(0)$ and $\sigma(1)$ from this state by the strong unforgeability as follows. By \cite[Lemma 2.1]{C:BonZha13},  the probability to output both $\sigma(0)$ and $\sigma(1)$ 
is only halved even if $\ket{0}\ket{\sigma(0)}+\ket{1}\ket{\sigma(1)}$ is measured in the computational basis before given to the adversary. After the measurement, the adversary's input collapses to a classical state $\ket{\alpha}\ket{\sigma(\alpha)}$ for random $\alpha\in \bit$, in which case the adversary can output $\sigma(\alpha\oplus1)$ only with a negligible probability by the strong unforgeability. 
Combining the above, security of the above scheme under tampered public keys is proven. 

\paragraph{\bf Achiving $\INDpkTA$ security.}
The above QPKE scheme satisfies $\onekey{\INDpkTA}$ security, but does not satisfy $\INDpkTA$ security where the adversary is given multiple copies of quantum public keys.
If the adversary is given two copies of the quantum public key, by measuring each public key in the computational basis,
the adversary can learn both $\sigma(0)$ and $\sigma(1)$ with probability 1/2.
In order to extend the scheme into $\INDpkTA$ security, we introduce a classical randomness for each public key generation.  
Specifically, a public key is 
\begin{equation}
    (r, \ket{0}\ket{\sigma(0,r)}+\ket{1}{\ket{\sigma(1,r)}})
\end{equation}
where 
$r\in \bit^\secp$ is chosen uniformly at random for every execution of the public key generation algorithm,
and
$\sigma(\alpha,r)$ is a signature for $\alpha \concat r$.\footnote{$\alpha\concat r$ is the concatenation of two bit strings $\alpha$ and $r$.} 
An encryption of a plaintext $b\in \bit$ is 
\begin{equation}
    (r, \ket{0}\ket{\sigma(0,r)}+(-1)^b\ket{1}{\ket{\sigma(1,r)}}).
\end{equation}
Since each quantum public key uses different $r$,  
security of this scheme holds even if the adversary obtains arbitrarily many public keys. 

\paragraph{\bf Making ciphertext classical.} 
The above constructions has quantum ciphertext, but our definition explicitly requires that a QPKE scheme have a classical cipheretxt.
We observe that the ciphertext of the above schemes can be made classical easily. In the $\INDpkTA$ secure construction, the ciphertext contains a quantum state $\ket{0}\ket{\sigma(r,0)}+(-1)^b\ket{1}{\ket{\sigma(r,1)}}$. Suppose that we measure this state in Hadamard basis and let $d$ be the measurement outcome. Then an easy calculation shows that we have 
\begin{equation}
    b=d\cdot (0\concat \sigma(0,r)\oplus 1\concat \sigma(1,r)).
\end{equation}
Thus, sending $(r,d)$ as a ciphertext is sufficient for the receiver who has the decryption key to recover the plaintext $b$. Moreover, this variant is at least as secure as the original one with quantum ciphertexts since the Hadamard-basis measurement only loses information of the ciphertext. 

\paragraph{\bf Achieving recyclability.}
Given that we achieve classical ciphertext property, it is rather straightforward to transform the construction into recyclable one where the encryption algorithm outputs a classical state that can be used to encrypt many plaintexts. The transformation uses standard hybrid encryption technique. Concretely, the encryption algorithm first generates a key $K$ of a SKE scheme, encrypt each bit of $K$ by the above non-recyclable scheme in a bit-by-bit manner, and encrypt the plaintext $\msg$ by the symmetric key encryption scheme under the key $K$. The final ciphertext is $(\ct,\ct_{\ske})$, where $\ct$ is the ciphertext of $K$ by the non-recyclable scheme and $\ct_{\ske}$ is the ciphertext of $\msg$ by the SKE scheme. The encryption algorithm outputs a classical state $(\ct,K)$ together with the ciphertext. The encryptor can reuse the state when it encrypts another message later.\footnote{The idea to achieve the recyclability by the hybrid encryption technique was also used in one of the constructions in \cite{TCC:BGHMSVW23}.}

\paragraph{\bf Adding decryption error detectability.} 
So far, we are only concerned with IND-pkTA security. 
On the other hand, the schemes presented in the previous paragraphs do not satisfy decryption error detectability. (See \Cref{def:dec_err} for formal definition.)  
Fortunately, there is a simple generic conversion that adds decryption error detectability while preserving IND-pkTA security by using digital signatures. The idea is that the encryption algorithm first generates a signature for the message under a signing key generated by itself, encrypts both the original message and signature under the building block scheme, and outputs the ciphertexts along with the verification key for the signature scheme in the clear. Then, the decryption algorithm can verify that the decryption result is correct as long as it is a valid message-signature pair (except for a negligible probability).

\subsection{CPA-to-CCA Transformation}\label{sec:overview_CCA}
We now explain how to transform $\INDpkTA$ secure QPKE scheme into $\INDCCApkTA$ secure one using OWFs.

\paragraph{\bf Definition of $\INDCCApkTA$ security.}
$\INDCCApkTA$ security is defined by adding the following two modifications to the security experiment for $\INDpkTA$ security.
\begin{itemize}
\item Throughout the experiment, the adversary can get access to the decryption oracle that is given a ciphertext $\ct$ and returns $\Dec(\sk,\ct)$ if $\ct\ne\ct^*$ and $\bot$ otherwise.
\item The adversary is given the $1$-bit leakage information that the challenge ciphertext is decrypted to $\bot$ or not.
\end{itemize}
As discussed in \cref{sec:discussion}, the second modification is needed to support a natural usage scenario of QPKE.
For simplicity, we will ignore this second modification for now and proceed the overview as if $\INDCCApkTA$ security is defined by just adding the decryption oracle to the security experiment for $\INDpkTA$ security.

We also define a weaker variant of $\INDCCApkTA$ security where the adversary is allowed to make only a single query to the decryption oracle.
We denote it as $\INDoneCCApkTA$ security.
We consider a relaxed variant of $\INDCCApkTA$ security and $\INDoneCCApkTA$ security where the adversary is given only a single copy of quantum public key.
We denote them as $\onekey{\INDCCApkTA}$ security and $\onekey{\INDoneCCApkTA}$ security respectively, similarly to $\onekey{\INDpkTA}$ security.

\paragraph{\bf $\INDoneCCApkTA$ from $\INDpkTA$.}
In classical cryptography, the CCA security where the number of decryption query is a-priori bounded to $q$ is called $q$-bounded-CCA security.
It is known that any CPA secure classical PKE scheme can be transformed into $q$-bounded-CCA secure one using only a digital signature scheme~\cite{AC:CHHIKP07}.
We show that by using a similar technique, we can transform an $\INDpkTA$ secure QPKE scheme into an $\INDoneCCApkTA$ secure one using only a digital signature scheme.

\paragraph{\bf Boosting $1$-bounded-CCA into full-fledged CCA.}
Classically, it is not known how to boost bounded-CCA security into CCA security without using additional assumption, and as a result, ``general transformation from CPA to CCA'' is a major open question in classical public key cryptography.
Surprisingly, we show that $1$-bounded-CCA security can be boosted into CCA security for QPKE assuming only OWFs.
More specifically, we show that $\onekey{\INDoneCCApkTA}$ secure QPKE can be transformed into $\onekey{\INDCCApkTA}$ secure one assuming only OWFs.

The key component in the transformation is tokenized message authentication code (MAC)~\cite{EPRINT:BehSatShi21}.
Tokenized MAC is a special MAC scheme where we can generate a quantum MAC token using the secret MAC key.
The quantum MAC token can be used to generate a valid signature only once.
In other words, an adversary who is given a single quantum MAC token cannot generate valid signatures for two different messages.
Tokenized MAC can be realized using only OWFs~\cite{EPRINT:BehSatShi21}.

The high level idea is to design CCA secure scheme so that a public key contains quantum MAC token and an adversary can generate a valid ciphertext only when it consumes the MAC token, which ensures that the adversary can make only one meaningful decryption query and CCA security is reduced to $1$-bounded-CCA security.
Consider the following construction of a QPKE scheme $\CCA$ based on $\onekey{\INDoneCCApkTA}$ secure QPKE scheme $\oneCCA$ and tokenized MAC scheme $\TMAC$.
The secret key of $\CCA$ consists of the secret keys of $\oneCCA$ and $\TMAC$, and the verification key of $\CCA$ is that of $\oneCCA$.
A quantum public key of $\CCA$ consists of that of $\oneCCA$ and a MAC token of $\TMAC$.
The encryption algorithm of $\CCA$ first generates a ciphertext $\onecca.\ct$ of $\oneCCA$ and then generates a signature $\tmac.\sigma$ for the message $\onecca.\ct$ by consuming the MAC token contained in the public key. The resulting ciphertext is $(\onecca.\ct,\tmac.\sigma)$.
The decryption algorithm of $\CCA$ that is given the ciphertext $(\onecca.\ct,\tmac.\sigma)$ first checks validity of $\tmac.\sigma$ by using the secret MAC key included in the secret key.
If it passes, the decryption algorithm decrypts $\onecca.\ct$ by using the secret key of $\oneCCA$.

In the experiment of $\onekey{\INDCCApkTA}$ security for $\CCA$, we can ensure that an adversary can make at most one decryption query whose result is not $\bot$ by the power of $\TMAC$, as we want.
However, the adversary in fact can make one critical query $(\onecca.\ct^*,\tmac.\sigma^\prime)$, where $\onecca.\ct^*$ is the first component of the challenge ciphertext, which allows the adversary to obtain the challenge bit.
This attack is possible due to the fact that the adversary is allowed to tamper the quantum public key.\footnote{
More specifically, the attack is done as follows.
The adversary is given a quantum public key $(\onecca.\pk,\token)$ where $\onecca.\pk$ is a public key of $\oneCCA$ and $\token$ is a MAC token of $\TMAC$.
The adversary generates another token $\token^\prime$ of $\TMAC$ by itself and sends $(\onecca.\pk,\token^\prime)$ to the challenger as the tempered public key.
Since there is no validity check on the MAC token in the encryption algorithm, this tampered public key is not rejected and the challenge ciphertext $(\onecca.\ct^*,\tmac.\sigma^*)$ is generated. 
Given the challenge ciphertext, the adversary generates a signature $\tmac.\sigma^\prime$ for $\onecca.\ct^*$ using $\token$ contained in the given un-tampered public key and queries $(\onecca.\ct^*,\tmac.\sigma^\prime)$ to the decryption oracle.
Since $\tmac.\sigma^\prime$ is a valid signature generated using the correct token, this query is successful and the adversary obtains the challenge bit.
}
Fortunately, this attack can be prevented by using a digital signature scheme and tying the two components $\onecca.\ct^*$ and $\tmac.\sigma^\prime$ together.
Once this issue is fixed, we can successfully reduce the $\onekey{\INDCCApkTA}$ security of the construction to the $\onekey{\INDoneCCApkTA}$ security of $\oneCCA$, since now the adversary can make only one non-critical decryption query.

\paragraph{\bf Upgrading $\onekey{\INDCCApkTA}$ to $\INDCCApkTA$.}
We can easily transform an $\onekey{\INDCCApkTA}$ secure QPKE scheme into an $\INDCCApkTA$ secure one.
The transformation is somewhat similar to the one from $\onekey{\INDpkTA}$ secure scheme to $\INDpkTA$ secure one.
We bundle multiple instances of $\onekey{\INDCCApkTA}$ secure scheme each of which is labeled by a classical random string.
The transformation uses pseudorandom functions and digital signatures both of which are implied by OWFs.

\paragraph{\bf How to deal with $1$-bit leakage ``the challenge is decrypted to $\bot$ or not''.}
So far, we ignore the fact that our definition of $\INDCCApkTA$ security allows the adversary to obtain $1$-bit leakage information whether the challenge is decrypted to $\bot$ or not.
We introduce an intermediate notion between $\INDpkTA$ security and $\INDCCApkTA$ security that we call $\INDCVApkTA$ security where the adversary is given the $1$-bit leakage information but is not allowed to get access to the decryption oracle.
We then show that an $\INDpkTA$ secure QPKE scheme can be transformed into $\INDCVApkTA$ secure one using the cut-and-choose technique.
Moreover, we show that the above construction strategy towards $\INDCCApkTA$ secure construction works even if the adversaries are given the $1$-bit leakage information, if we start with $\INDCVApkTA$ secure scheme.

\paragraph{\bf Some Remarks.}
We finally provide some remarks.
\begin{description}
\item[Recyclability:] Similarly to $\INDpkTA$ secure scheme, we consider recyclable variant for $\INDCCApkTA$ secure one. We show that a recyclable $\INDCCApkTA$ secure QPKE scheme can be constructed from non-recyclable one using the hybrid encryption technique similarly to $\INDpkTA$ secure construction. 
\item[Strong decryption error detectability:] In the proof of the construction from $\onekey{\INDoneCCApkTA}$ secure scheme to $\onekey{\INDCCApkTA}$ secure one, we use the underlying scheme's decryption error detectability. The proof of CCA security is sensitive to decryption errors, and it turns out that decryption error detectability that only provides security guarantee against computationally bounded adversaries is not sufficient for this part. Thus, we introduce statistical variant of decryption error detectability that we call strong decryption error detectability.
We also prove that our $\INDCVApkTA$ secure construction based on the cut-and-choose technique achieves strong decryption error detectability, and the subsequent transformations preserve it.
\end{description}

%% file: preliminaries.tex

\newcommand{\ODec}[1]{O_{\mathsf{Dec},#1}}

\section{Preliminaries}\label{sec:preliminaries}
\subsection{Basic Notations}
\label{sec:basic_notations}

We use the standard notations of quantum computing and cryptography.
We use $\secp$ as the security parameter.
For any set $S$, $x\gets S$ means that an element $x$ is sampled uniformly at random from the set $S$.
We write $\negl$ to mean a negligible function.
PPT stands for (classical) probabilistic polynomial-time and QPT stands for quantum polynomial-time.
For an algorithm $A$, $y\gets A(x)$ means that the algorithm $A$ outputs $y$ on input $x$.
For two bit strings $x$ and $y$, $x\|y$ means the concatenation of them.
For simplicity, we sometimes omit the normalization factor of a quantum state.
(For example, we write $\frac{1}{\sqrt{2}}(|x_0\rangle+|x_1\rangle)$ just as
$|x_0\rangle+|x_1\rangle$.)
$I\coloneqq|0\rangle\langle0|+|1\rangle\langle 1|$ is the two-dimensional identity operator.
For the notational simplicity, we sometimes write $I^{\otimes n}$ just as $I$ when
the dimension is clear from the context.

\input{primitives}

\subsection{Lemma by Boneh and Zhandry}
In this paper, we use the following lemma by Boneh and Zhandry~\cite{C:BonZha13}.

\begin{lemma}[{\cite[Lemma 2.1]{C:BonZha13}}]
\label{lem:BZ}
Let $A$ be a quantum algorithm, and let $\Pr[x]$ be the probability that $A$ outputs $x$. Let
$A'$ be another quantum algorithm obtained from $A$ by pausing $A$ at an arbitrary stage of execution,
performing a partial measurement that obtains one of $k$ outcomes, and then resuming $A$. 
Let $\Pr'[x]$ be the probability that $A'$ outputs $x$. Then $\Pr'[x] \ge \Pr[x]/k$.
\end{lemma}

%% file: primitives.tex
\subsection{Digital Signatures}
\begin{definition}[Digital signatures] \label{def:sEUF-CMA} 
A digital signature scheme is a set of algorithms $(\Gen,\Sign,\Ver)$ such that
\begin{itemize}
    \item 
    $\Gen(1^\secp)\to(k,\vk):$ It is a PPT algorithm that, on input the security parameter $\secp$, outputs
    a signing key $k$ and a verification key $\vk$.
    \item
    $\Sign(k,\msg)\to\sigma:$
    It is a PPT algorithm that, on input the message $\msg$ and $k$, outputs a signature $\sigma$.
    \item
   $\Ver(\vk,\msg,\sigma)\to\top/\bot:$ 
   It is a deterministic classical polynomial-time algorithm that, on input $\vk$, $\msg$, and $\sigma$, outputs $\top/\bot$.
\end{itemize}
We require the following correctness and strong EUF-CMA security.

\paragraph{\bf Correctness:}
For any $\msg$,
\begin{equation}
   \Pr[\top\gets\Ver(\vk,\msg,\sigma):
   (k,\vk)\gets\Gen(1^\secp),
   \sigma\gets\Sign(k,\msg)
   ]\ge1-\negl(\secp). 
\end{equation}

\paragraph{\bf Strong EUF-CMA security:}
For any QPT adversary $\cA$ with classical oracle access to the signing oracle $\Sign(k,\cdot)$,
\ifnum\llncs=0
\begin{equation}
   \Pr[(\msg^\ast,\sigma^\ast)\notin \mathcal{Q}~\land~\top\gets\Ver(\vk,\msg^*,\sigma^*):
   (k,\vk)\gets\Gen(1^\secp),
   (\msg^\ast,\sigma^\ast)\gets\cA^{\Sign(k,\cdot)}(\vk)
   ]\le\negl(\secp), 
\end{equation}
\else
\begin{align}
   \Pr\left[(\msg^\ast,\sigma^\ast)\notin \mathcal{Q}~\land~\top\gets\Ver(\vk,\msg^*,\sigma^*):
   \begin{array}{r}
   (k,\vk)\gets\Gen(1^\secp)\\
   (\msg^\ast,\sigma^\ast)\gets\cA^{\Sign(k,\cdot)}(\vk)
   \end{array}
   \right]\notag
   \le\negl(\secp), 
\end{align}
\fi
where $\mathcal{Q}$ is the set of message-signature pairs returned by the signing oracle. 
\end{definition}

\begin{remark}
Without loss of generality, we can assume that $\Sign$ is deterministic.
(The random seed used for $\Sign$ can be generated by applying a PRF on the message signed, and the key of PRF is appended to the signing key.)
\end{remark}

\begin{theorem}[{\cite[Sec. 6.5.2]{Book:Goldreich04}}]\label{thm:sig_from_OWF}
Strong EUF-CMA secure digital signatures exist if OWFs exist.
\end{theorem}

\subsection{Pseudorandom Functions}

\begin{definition}[Pseudorandom functions (PRFs)]
A keyed function $\{\PRF_K: \cX\rightarrow \cY\}_{K\in \bit^\secp}$ that is computable in classical deterministic polynomial-time is a quantum-query secure pseudorandom function if 
for any QPT adversary $\cA$ with quantum access to the evaluation oracle $\PRF_K(\cdot)$,
\begin{equation}
   |\Pr[1\gets\cA^{\PRF_K(\cdot)}(1^\secp)] 
   -\Pr[1\gets\cA^{H(\cdot)}(1^\secp)] |\le\negl(\secp),
\end{equation}
where $K\gets\bit^\secp$ and $H:\cX\rightarrow \cY$ is a function chosen uniformly at random.
\end{definition}

As we can see, we consider PRFs that is secure even if an adversary can get access to the oracles in superposition, which is called quantum-query secure PRFs.
We use the term PRFs to indicate quantum-query secure PRFs in this work.

\begin{theorem}[{\cite{FOCS:Zhandry12}}]\label{rem:PRF}
(Quantum-query secure) PRFs exist if OWFs exist.
\end{theorem}

\subsection{Symmetric Key Encryption}

\begin{definition}[Symmetric Key Encryption (SKE)]
A (classical) symmetric key encryption (SKE) scheme with message space $\bit^\ell$ is a set of algorithms $(\Enc,\Dec)$ such that
\begin{itemize}
    \item
    $\Enc(K,\msg)\to\ct:$
    It is a PPT algorithm that, on input $K\in \bit^\secp$ and the message $\msg\in \bit^\ell$, outputs a ciphertext $\ct$.
    \item
   $\Dec(K,\ct)\to\msg':$ 
   It is a deterministic classical polynomial-time algorithm that, on input $K$ and $\ct$, outputs $\msg'$.
\end{itemize}
We require the following correctness.
\paragraph{\bf Correctness:}
For any $\msg\in\bit^\ell$,
\begin{equation}
   \Pr[\msg\gets\Dec(K,\ct):
   K\gets\bit^\secp, 
   \ct\gets\Enc(K,\msg)
   ] = 1.
\end{equation}
\end{definition}

\begin{definition}[IND-CPA Security]
For any QPT adversary $\cA$ with classical oracle access to the encryption oracle $\Enc(K,\cdot)$,
\begin{equation}
   \Pr\left[b\gets\cA(\ct^*,\st)^{\Enc(K,\cdot)}:
   \begin{array}{r}
   K\gets\bit^\secp\\
   (\msg_0,\msg_1,\st)\gets \cA^{\Enc(K,\cdot)}(1^\secp)\\
   b\gets \bit\\ 
   \ct^*\gets\Enc(K,\msg_b)
   \end{array}
   \right] \le \frac{1}{2}+\negl(\secp).
\end{equation}   
\end{definition}

\begin{theorem}[\cite{JACM:GolGolMic86,SIAMCOMP:HILL99}]
IND-CPA secure SKE exists if OWFs exist.
\end{theorem}

\begin{definition}[IND-CCA Security]
For any QPT adversary $\cA$ with classical oracle access to the encryption oracle $\Enc(K,\cdot)$,
\begin{equation}
   \Pr\left[b\gets\cA(\ct^*,\st)^{\Enc(K,\cdot),\ODec{2}(\cdot)}:
   \begin{array}{r}
   K\gets\bit^\secp\\
   (\msg_0,\msg_1,\st)\gets \cA^{\Enc(K,\cdot),\ODec{1}(\cdot)}(1^\secp)\\
   b\gets \bit\\ 
   \ct^*\gets\Enc(K,\msg_b)
   \end{array}
   \right] \le \frac{1}{2}+\negl(\secp).
\end{equation}   
Here, $\ODec{1}(\ct)$ returns $\Dec(K,\ct)$ for any $\ct$.
$\ODec{2}$ behaves identically to $\ODec{1}$ except that $\ODec{2}$ returns $\bot$ to the input $\ct^*$.
\end{definition}

\begin{theorem}[\cite{JC:BelNam08}]
IND-CCA secure SKE exists if OWFs exist.
\end{theorem}

%% file: definitions.tex
\section{Definition of QPKE}
\label{sec:definition}

In this section, we define QPKE. 

\begin{definition}[Quantum Public-Key Encryption (QPKE)]\label{def:QPKE}
A quantum public-key encryption scheme with message space $\bit^\ell$ is a set of algorithms 
$(\SKGen,\PKGen,\Enc,\Dec)$ such that
\begin{itemize}
    \item 
    $\SKGen(1^\secp)\to (\sk,\vk):$
    It is a PPT algorithm that, on input the security parameter $\secp$, outputs
    a classical secret key $\sk$ and a classical verification key $\vk$. 
    \item
    $\PKGen(\sk)\to \pk:$
    It is a QPT algorithm that, on input $\sk$, outputs
    a quantum public key $\pk$.
    \item
    $\Enc(\vk,\pk,\msg)\to \ct:$ 
    It is a QPT algorithm that, on input $\vk$, $\pk$, and a plaintext $\msg\in\bit^\ell$,  outputs a classical ciphertext $\ct$. 
    \item
    $\Dec(\sk,\ct)\to \msg':$ 
    It is a classical deterministic polynomial-time algorithm that, on input $\sk$ and $\ct$, outputs $\msg'\in\bit^\ell \cup \{\bot\}$.
\end{itemize}
We require the following correctness and IND-pkTA security.

\paragraph{\bf Correctness:}
For any $\msg\in\bit^\ell$,
\ifnum\llncs=0
\begin{equation}
   \Pr[\msg\gets\Dec(\sk,\ct):
   (\sk,\vk)\gets\SKGen(1^\secp),
   \pk\gets\PKGen(\sk),
   \ct\gets\Enc(\vk,\pk,\msg)
   ] \ge 1-\negl(\secp).
\end{equation}
\else
\begin{equation}
   \Pr\left[\msg\gets\Dec(\sk,\ct):
   \begin{array}{r}
   (\sk,\vk)\gets\SKGen(1^\secp) \\
   \pk\gets\PKGen(\sk) \\
   \ct\gets\Enc(\vk,\pk,\msg)
   \end{array}
   \right] \ge 1-\negl(\secp).
\end{equation}
\fi

\paragraph{\bf IND-pkT-CPA Security:}
For any polynomial $m$, and any QPT adversary $\cA$,
\begin{equation}
   \Pr\left[b\gets\cA(\ct^\ast,\st):
   \begin{array}{r}
   (\sk,\vk)\gets\SKGen(1^\secp)\\
   \pk_1,...,\pk_m\gets\PKGen(\sk)^{\otimes m}\\
   (\pk',\msg_0,\msg_1,\st)\gets\cA(\vk,\pk_1,...,\pk_m)\\
   b\gets \bit\\ 
   \ct^\ast\gets\Enc(\vk,\pk',\msg_b)
   \end{array}
   \right] \le \frac{1}{2}+\negl(\secp).
\end{equation}
Here, $\pk_1,...,\pk_m\gets\PKGen(\sk)^{\otimes m}$
means that $\PKGen$ is executed $m$ times and $\pk_i$ is the output of the $i$th execution of $\PKGen$.
$\st$ is a quantum internal state of $\cA$, which can be entangled with $\pk'$.
\end{definition}
As we discussed in~\cref{sec:discussion}, the above definition does not require the quantum public key $\pk$ to be a pure state.

We also define a security notion related to the correctness notion that we call decryption error detectability.

\begin{definition}[Decryption error detectability]
\label{def:dec_err}
We say that a QPKE scheme has decryption error detectability if for 
any polynomial $m$, 
and any QPT adversary $\cA$,
\begin{equation}
   \Pr\left[
   \msg'\neq \bot~\land~ 
   \msg'\neq \msg :
   \begin{array}{r}
   (\sk,\vk)\gets\SKGen(1^\secp)\\
   \pk_1,...,\pk_m\gets\PKGen(\sk)^{\otimes m}\\
   (\pk',\msg)\gets\cA(\vk,\pk_1,...,\pk_m)\\
   \ct\gets\Enc(\vk,\pk',\msg)\\
   \msg' \gets\Dec(\sk,\ct)
   \end{array}
   \right] \le \negl(\secp).
\end{equation}
\end{definition}

It is easy to see that we can generically add decryption error detectability by letting the sender generate a signature for the message under a signing key generated by itself, encrypt the concatenation of the message and signature, and send the ciphertext along with the verification key of the signature to the receiver.
The receiver can check that there is no decryption error (except for a negligible probability) if the decryption result is a valid message-signature pair. That is, we have the following theorem.  
\begin{theorem}\label{thm:add_decryption_error_detectability}
If there exist OWFs and a QPKE scheme that satisfies correctness and IND-pkT-CPA security, there exists a QPKE scheme that satisfies correctness, IND-pkT-CPA security, and decryption error detectability. 
\end{theorem}
We omit the proof since it is straightforward by the construction explained above. Since we have this theorem, we focus on constructing QPKE that satisfies correctness and IND-pkT-CPA security.  

%% file: construction.tex
\section{Construction of QPKE}
\label{sec:construction}
In this section, we construct a QPKE scheme that satisfies correctness and IND-pkT-CPA security (but not decryption error detectability) from 
strong EUF-CMA secure digital signatures. 
The message space of our construction is $\bit$, but it can be extended to be arbitrarily many bits by parallel repetition. 
Let $(\Gen,\Sign,\Ver)$ be a strong EUF-CMA secure digital signature scheme with a deterministic $\Sign$ algorithm and message space $\bit^u$ for $u=\omega(\log \secp)$. 

Our construction of QPKE is as follows.
\begin{itemize}
    \item 
    $\SKGen(1^\secp)\to(\sk,\vk):$
    Run $(k,\vk)\gets\Gen(1^\secp)$.
    Output $\sk\seteq k$. Output $\vk$. 
    \item
    $\PKGen(\sk)\to\pk:$
    Parse $\sk=k$.
    Choose $r\gets\bit^u$.
    By running $\Sign$ coherently, generate the state
    \begin{align}
    \ket{\psi_r}\seteq\ket{0}_\regA\otimes\ket{\Sign(k,0\|r)}_\regB
    +\ket{1}_\regA\otimes\ket{\Sign(k,1\|r)}_\regB
    \end{align}
    over registers $(\regA,\regB)$.
    Output 
    $
    \pk\coloneqq(r,\ket{\psi_r}).
    $
   \item
   $\Enc(\vk,\pk,b)\to\ct:$
   Parse
    $\pk=(r,\rho)$, where $\rho$ is a quantum state over registers $(\regA,\regB)$. 
    The $\Enc$ algorithm consists of the following three steps. 
   \begin{enumerate}
    \item
    It coherently checks the signature in $\rho$. In other words, it applies the unitary
   \begin{equation}
      U_{r,\vk}\ket{\alpha}_\regA\ket{\beta}_\regB\ket{0...0}_\regD
      =\ket{\alpha}_\regA\ket{\beta}_\regB\ket{\Ver(\vk,\alpha\|r,\beta)}_\regD
   \end{equation} 
    on $\rho_{\regA,\regB}\otimes\ket{0...0}\bra{0...0}_\regD$,\footnote{$\regC$ is skipped, because $\regC$ will be used later.} and
    measures the register $\regD$ in the computational basis. If the result is $\bot$, it outputs $\ct\seteq\bot$ and halts.
    If the result is $\top$, 
    it goes to the next step.
    \item
    It applies $Z^b$ on the register $\regA$. 
    \item
    It measures all qubits in the registers $(\regA,\regB)$ in the Hadamard basis to get the result $d$.
    It outputs
    $
    \ct\seteq(r,d).
    $
    \end{enumerate}
   \item
   $\Dec(\sk,\ct)\to b':$
   Parse $\sk=k$ and
    $\ct=(r,d)$.
    Output 
    \begin{equation}
    b'\seteq d\cdot 
    (0\|\Sign(k,0\|r) \oplus 1\|\Sign(k,1\|r)).
    \end{equation}
\end{itemize}

\begin{theorem}\label{thm:IND-pkTA_QPKE_from_SIG}
If $(\Gen,\Sign,\Ver)$ is a strong EUF-CMA secure digital signature scheme, then
the QPKE scheme $(\SKGen,\PKGen,\Enc,\Dec)$ above is correct and IND-pkT-CPA secure.
\end{theorem}

The correctness is straightforward. 
First, the state over the registers $(\regA,\regB)$ is
$\ket{\psi_r}$ if $\pk$ was not tampered with and the first step of $\Enc$ algorithm got $\top$.
Second, in that case, the state becomes 
\begin{align}
\ket{0}\ket{\Sign(k,0\|r)}+(-1)^b\ket{1}\ket{\Sign(k,1\|r)} 
\end{align}
after the second step of $\Enc$ algorithm.
Finally, because in that case $d$ obtained in the third step of $\Enc$ algorithm satisfies
    \begin{equation}
    b=d\cdot 
    (0\|\Sign(k,0\|r)
    \oplus 1\|\Sign(k,1\|r)),
    \end{equation}
    we have $b'=b$.

We prove IND-pkT-CPA security in the next section.

%% file: security_detailed.tex
\section{Proof of IND-pkT-CPA Security}\label{sec:security}
In this section, we show IND-pkT-CPA security of our construction to complete the proof of~\cref{thm:IND-pkTA_QPKE_from_SIG}.
The outline of the proof is as follows.
The security game for the IND-pkT-CPA security of our QPKE (Hybrid 0) is given in \cref{hyb0}.
We introduce two more hybrids, Hybrid 1 (\cref{hyb1}) and Hybrid 2 (\cref{hyb2}).
Hybrid 1 is the same as Hybrid 0 except that the challenger does not do the Hadamard-basis measurement
in the third step of $\Enc$ algorithm, and the challenger sends the adversary $r$ and the state over the registers
$(\regA,\regB)$.
Hybrid 2 is the same as Hybrid 1 except that the adversary outputs two bit strings $\mu_0,\mu_1$ and
the adversary wins if $\mu_0=\Sign(k,0\|r)$ and $\mu_1=\Sign(k,1\|r)$. 
The formal proof is as follows.

Assume that the IND-pkT-CPA security of our construction is broken by a QPT adversary
$\cA$. It means the QPT adversary $\cA$ wins Hybrid 0 with a non-negligible advantage.
Then, it is clear that there is another QPT adversary $\cA'$ that wins Hybrid 1  
with a non-negligible advantage. ($\cA'$ has only to do the Hadamard-basis measurement by itself.)
From the $\cA'$, we can construct a QPT adversary $\cA''$ that wins 
Hybrid 2 with a non-negligible probability 
by using the idea of \cite{EC:HhaMorYam23}. (For details, see \cref{sec:B}).
Finally, we show in \cref{sec:Zhandry}
that no QPT adversary can win Hybrid 2 except for a negligible probability.
We thus have the contradiction, and 
therefore our QPKE is IND-pkT-CPA secure.

\protocol{Hybrid 0}
{Hybrid 0}
{hyb0}
{
\begin{enumerate}
    \item \label{item:initial}
    $\cC$ runs $(k,\vk)\gets\Gen(1^\secp)$.
    $\cC$ sends $\vk$ to $\cA$.
    \item
    \label{rs}
    $\cC$ chooses $r_1,...,r_m\gets\bit^u$.
    \item
    \label{send}
    $\cC$ sends $\{(r_i,\ket{\psi_{r_i}})\}_{i=1}^m$
    to the adversary $\cA$, where
    \begin{align}
       \ket{\psi_{r_i}}&\coloneqq
       \ket{0}\otimes\ket{\Sign(k,0\|r_i)}
       +\ket{1}\otimes\ket{\Sign(k,1\|r_i)}.
    \end{align}
    \item
    \label{Ar}
    $\cA$ generates a quantum state over registers $(\regA,\regB,\regC)$. 
    ($(\regA,\regB)$ corresponds to the quantum part of $\pk'$, and $\regC$ corresponds to $\st$.)
    $\cA$ sends a bit string $r$ and the registers $(\regA,\regB)$ to $\cC$.
    $\cA$ keeps the register $\regC$.
    \item
    \label{sign}
    $\cC$ coherently checks the signature in the state sent from $\cA$.
    If the result is $\bot$, it sends $\bot$ to $\cA$ and halts.
    If the result is $\top$, it goes to the next step. 
   
    \item
    \label{bchosen}
    $\cC$ chooses $b\gets \bit$.
    $\cC$ applies $Z^b$ on the register $\regA$.

    \item
    $\cC$ measures all qubits in $(\regA,\regB)$ in the Hadamard basis to get the result $d$.
    $\cC$ sends $(r,d)$ to $\cA$.
    \item
    \label{last}
    $\cA$ outputs $b'$.
    If $b'=b$, $\cA$ wins.
\end{enumerate}
}

\protocol{Hybrid 1}
{Hybrid 1}
{hyb1}
{
\begin{itemize}
    \item[1.-6.]
    All the same as \cref{hyb0}.
    \item[7.]
    $\cC$ does not do the Hadamard-basis measurement,
    and $\cC$ sends $r$ and registers $(\regA,\regB)$ to $\cA$.
    \item[8.]
    The same as \cref{hyb0}.
\end{itemize}
}

\protocol{Hybrid 2}
{Hybrid 2}
{hyb2}
{
\begin{itemize}
    \item[1.-7.]
    All the same as \cref{hyb1}.
    \item[8.]
    $\cA$ outputs $(\mu_0,\mu_1)$.
    If $\mu_0=\Sign(k,0\|r)$ and $\mu_1=\Sign(k,1\|r)$, $\cA$ wins.
\end{itemize}
}

\subsection{From Distinguishing to Outputting Two Signatures}
\label{sec:B}
We present the construction of $\cA''$.
Assume that there exists a QPT adversary
$\cA'$ and a polynomial $p$ such that 
\begin{equation}
|\Pr[1\gets\cA'\mid b=0]-\Pr[1\gets\cA'\mid b=1]|\ge\frac{1}{p(\secp)}
\label{advantage}
\end{equation}
in Hybrid 1 (\cref{hyb1}) for all $\secp\in I$ with an infinite set $I$.
From the $\cA'$, we construct a QPT adversary $\cA''$ such that
\begin{equation}
\Pr[(\Sign(k,0\|r),\Sign(k,1\|r))\gets\cA'']\ge\frac{1}{q(\secp)}
\end{equation}
in Hybrid 2 (\cref{hyb2}) with a polynomial $q$ for infinitely many $\secp$.

Let $t\coloneqq (k,\vk,r_1,...,r_m,r)$, and $\Pr[t]$
be the probability that $t$ is generated in \cref{item:initial}, \cref{rs}, and \cref{Ar} in the game of \cref{hyb1}. 
Let $\mathsf{Good}$ be the event that
$\cC$ gets $\top$ in \cref{sign} in the game of \cref{hyb1}.
Let $\mathsf{Bad}$ be the event that
$\mathsf{Good}$ does not occur.
Then, from \cref{advantage}, we have
\begin{align}
\frac{1}{p(\secp)}&\le
\Big|\sum_t\Pr[t] \Pr[\mathsf{Good} \mid t]\Pr[1\gets\cA' \mid t,\mathsf{Good},b=0]\\
&+\sum_t\Pr[t] \Pr[\mathsf{Bad} \mid t]\Pr[1\gets\cA' \mid t,\mathsf{Bad},b=0]\nonumber\\
&-\sum_t\Pr[t] \Pr[\mathsf{Good} \mid t]\Pr[1\gets\cA' \mid t,\mathsf{Good},b=1]\\
&-\sum_t\Pr[t] \Pr[\mathsf{Bad} \mid t]\Pr[1\gets\cA' \mid t,\mathsf{Bad},b=1]\Big|\\
&\le
\sum_t\Pr[t] \Pr[\mathsf{Good} \mid t]
\Big|\Pr[1\gets\cA' \mid t,\mathsf{Good},b=0]
-\Pr[1\gets\cA' \mid t,\mathsf{Good},b=1]\Big|\nonumber\\
&+\sum_t\Pr[t] \Pr[\mathsf{Bad} \mid t]
\Big|\Pr[1\gets\cA' \mid t,\mathsf{Bad},b=0]
-\Pr[1\gets\cA' \mid t,\mathsf{Bad},b=1]\Big|\\
&=
\sum_t\Pr[t] \Pr[\mathsf{Good} \mid t]
\Big|\Pr[1\gets\cA' \mid t,\mathsf{Good},b=0]
-\Pr[1\gets\cA' \mid t,\mathsf{Good},b=1]\Big|
\end{align}
for all $\secp\in I$,
because if $\mathsf{Bad}$ occurs, $\cA'$ gets only $\bot$ which contains no information about $b$.
(Here, we often abuse notation to just write $t$ to mean the event that $t$ is generated 
in \cref{item:initial}, \cref{rs}, and \cref{Ar}.)
Therefore, if we define
\begin{equation}
   T_\secp\coloneqq\Big\{t:
   \Pr[\mathsf{Good} \mid t]
\cdot \Big|\Pr[1\gets\cA' \mid t,\mathsf{Good},b=0]
-\Pr[1\gets\cA' \mid t,\mathsf{Good},b=1]\Big|\ge\frac{1}{2p(\secp)}
   \Big\}, 
\end{equation}
we have, for all $\secp\in I$,
\begin{align}
\Pr[\mathsf{Good} \mid t]\ge\frac{1}{4p(\secp)}
\label{goodnonnegl}
\end{align}
and
\begin{align}
\Big|\Pr[1\gets\cA' \mid t,\mathsf{Good},b=0]
-\Pr[1\gets\cA' \mid t,\mathsf{Good},b=1]\Big|
\ge\frac{1}{2p(\secp)}
\label{dif}
\end{align}
for any $t\in T_\secp$ and
$
\sum_{t\in T_\secp}\Pr[t]\ge\frac{1}{2p(\secp)}.
$

Let $\ket{\phi_b^{t,good}}$ be the state over the registers $(\regA,\regB,\regC)$
immediately before \cref{last} of \cref{hyb1} given that 
$t$ is generated,
$\mathsf{Good}$ occurred,
and $b$ is chosen in \cref{bchosen} of \cref{hyb1}.
We can show the following lemma. (Its proof is given later.)
\begin{lemma}
\label{why}
If $(\Gen,\Sign,\Ver)$ is strong EUF-CMA secure, 
there exists a subset $T_\secp'\subseteq T_\secp$ such that the following is satisfied
for all $\secp\in I'$, where $I'\coloneqq\{\secp\in I:\secp\ge\secp_0\}$ with a certain $\secp_0$.
\begin{itemize}
\item
$\sum_{t\in T_\secp'}\Pr[t]\ge\frac{1}{4p(\secp)}$.
\item
For any $t\in T_\secp'$, $\ket{\phi_b^{t,good}}$ is close to a state 
    \begin{equation}
\ket{\tilde{\phi}_b^{t,good}}\coloneqq    c_0\ket{0}_\regA\ket{\Sign(k,0\|r)}_\regB\ket{\Psi_0}_{\regC}    
    +(-1)^bc_1\ket{1}_\regA\ket{\Sign(k,1\|r)}_\regB\ket{\Psi_1}_{\regC}   
    \label{whystate}
    \end{equation}
within the trace distance $\frac{1}{p^{10}(\secp)}$,
where $c_0$ and $c_1$ are some complex coefficients such that $|c_0|^2+|c_1|^2=1$, 
and $\ket{\Psi_0}$ and $\ket{\Psi_1}$ are some normalized states.
\end{itemize}
\end{lemma}

Now let us fix $t\in T_\secp'$.
Also, assume that $\mathsf{Good}$ occurred.
Because $T_\secp'\subseteq T_\secp$, it means that $t\in T_\secp$.
Then,
from \cref{dif},
\begin{align}
\Big|\Pr[1\gets\cA' \mid t,\mathsf{Good},b=0]
-\Pr[1\gets\cA' \mid t,\mathsf{Good},b=1]\Big|=\Delta
\label{advantage2}
\end{align}
for a non-negligible $\Delta\ge\frac{1}{2p(\secp)}$ for all $\secp\in I$.
Without loss of generality, we can assume that in \cref{last} of \cref{hyb1},
$\cA'$ applies a unitary $V$ on the state $\ket{\phi_b^{t,good}}$, and
measures the register $\regA$ in the computational basis to get $b'\in\bit$.
By \cref{advantage2}
we have
\begin{align}
V\ket{\phi_0^{t,good}}&=\sqrt{p}\ket{1}_{\regA}\ket{\nu_1}_{\regB,\regC}+\sqrt{1-p}\ket{0}_{\regA}\ket{\nu_0}_{\regB,\regC} \label{eq:V_phi_zero}\\    
V\ket{\phi_1^{t,good}}&=\sqrt{1-p+\Delta}\ket{0}_{\regA}\ket{\xi_0}_{\regB,\regC}+\sqrt{p-\Delta}\ket{1}_{\regA}\ket{\xi_1}_{\regB,\regC}    \label{eq:V_phi_one} 
\end{align}
for some real number $p$ and some normalized states $\ket{\nu_0},\ket{\nu_1},\ket{\xi_0},\ket{\xi_1}$.
(This is because any state can be written as $p\ket{1}\ket{\nu_1}+\sqrt{1-p}\ket{0}\ket{\nu_0}$
with some $p$ and normalized states $\ket{\nu_0},\ket{\nu_1}$, and due to \cref{advantage2}, the coefficients
of $\ket{1}\ket{\xi_1}$ has to be $\sqrt{p-\Delta}$.)
If we define $W$ as
$
   W\coloneqq V^\dagger (Z\otimes I) V, 
$
we have
\begin{align}
|\bra{\tilde{\phi}_b^{t,good}} W \ket{\tilde{\phi}_b^{t,good}}   
-\bra{\phi_b^{t,good}} W \ket{\phi_b^{t,good}}|\le\frac{2}{p^{10}(\secp)}   
\label{eq:2p10}
\end{align}
for all $\secp\in I'$
from \cref{why}.
Therefore,
\begin{align}
&   |
   c_0^*c_1\bra{0}\bra{\Sign(k,0\|r)}\bra{\Psi_0}W\ket{1}\ket{\Sign(k,1\|r)}\ket{\Psi_1}\\
   &+c_0c_1^*\bra{1}\bra{\Sign(k,1\|r)}\bra{\Psi_1}W\ket{0}\ket{\Sign(k,0\|r)}\ket{\Psi_0}|\\
  &=\frac{1}{4} |
  (\bra{\tilde{\phi}_0^{t,good}}+\bra{\tilde{\phi}_1^{t,good}})W(\ket{\tilde{\phi}_0^{t,good}}-\ket{\tilde{\phi}_1^{t,good}})\\
  &+(\bra{\tilde{\phi}_0^{t,good}}-\bra{\tilde{\phi}_1^{t,good}})W(\ket{\tilde{\phi}_0^{t,good}}+\ket{\tilde{\phi}_1^{t,good}})
  | \\
   &=\frac{1}{2} |
  \bra{\tilde{\phi}_0^{t,good}}W\ket{\tilde{\phi}_0^{t,good}}
  -\bra{\tilde{\phi}_1^{t,good}}W\ket{\tilde{\phi}_1^{t,good}}|\\
  &\ge\frac{1}{2} |
  \bra{\phi_0^{t,good}}W\ket{\phi_0^{t,good}}
  -\bra{\phi_1^{t,good}}W\ket{\phi_1^{t,good}}
  |-\frac{2}{p^{10}(\secp)}\label{eq:rewrite}\\
   \begin{split}
   &=
  \frac{1}{2}\left|\left(\sqrt{p}\bra{1}\bra{\nu_1}+\sqrt{1-p}\bra{0}\bra{\nu_0}\right)
  \left(-\sqrt{p}\ket{1}\ket{\nu_1}+\sqrt{1-p}\ket{0}\ket{\nu_0}\right)\right.\\
  &\left.~~~-
    \left(\sqrt{1-p+\Delta}\bra{0}\bra{\xi_0}+\sqrt{p-\Delta}\bra{1}\bra{\xi_1}\right) 
  \left(\sqrt{1-p+\Delta}\ket{0}\ket{\xi_0}-\sqrt{p-\Delta}\ket{1}\ket{\xi_1}\right)\right|
  -\frac{2}{p^{10}(\secp)}
  \end{split} \label{eq:expand}
  \\
  &=\frac{1}{2}\left|-p+(1-p)-(1-p+\Delta)+(p-\Delta)\right|-\frac{2}{p^{10}(\secp)}\\
  & =\Delta-\frac{2}{p^{10}(\secp)}\\ 
  & \ge\frac{1}{2p(\secp)}-\frac{2}{p^{10}(\secp)}\\ 
  & \ge\frac{1}{p(\secp)}
\end{align}
for all $\secp\in I'$.
Here, \Cref{eq:rewrite} follows from \cref{eq:2p10}, and 
\Cref{eq:expand} follows from \Cref{eq:V_phi_zero,eq:V_phi_one} and the definition of $W$. 
From the triangle inequality and the facts that $|c_0|\le1$ and $|c_1|\le1$,
\begin{align}
&\frac{1}{p(\secp)}\le
   |c_1|\cdot|\bra{0}\bra{\Sign(k,0\|r)}\bra{\Psi_0}W\ket{1}\ket{\Sign(k,1\|r)}\ket{\Psi_1}|\\
   &+|c_0|\cdot|\bra{1}\bra{\Sign(k,1\|r)}\bra{\Psi_1}W\ket{0}\ket{\Sign(k,0\|r)}\ket{\Psi_0}|
\end{align}
for all $\secp\in I'$.
Then,
\begin{align}
\frac{1}{2p(\secp)}\le
   |c_1|\cdot|\bra{0}\bra{\Sign(k,0\|r)}\bra{\Psi_0}W\ket{1}\ket{\Sign(k,1\|r)}\ket{\Psi_1}|
\end{align}
or
\begin{align}
&\frac{1}{2p(\secp)}\le
  |c_0|\cdot|\bra{1}\bra{\Sign(k,1\|r)}\bra{\Psi_1}W\ket{0}\ket{\Sign(k,0\|r)}\ket{\Psi_0}|
\end{align}
holds for all $\secp\in I'$.
Assume that the latter holds.
(The following proof can be easily modified even if the former holds.)
Then
\begin{align}
\frac{1}{4p^2(\secp)}&\le
   |c_0|^2\cdot|\bra{1}\bra{\Sign(k,1\|r)}\bra{\Psi_1}W\ket{0}\ket{\Sign(k,0\|r)}\ket{\Psi_0}|^2\\
   &\le
   |c_0|^2\cdot\|(I\otimes\bra{\Sign(k,1\|r)}\otimes I)W\ket{0}\ket{\Sign(k,0\|r)}\ket{\Psi_0}\|^2\label{delta}
\end{align}
for all $\secp\in I'$.
With this $W$, we construct the QPT adversary $\cA''$ 
as is shown in \cref{A''}.
\protocol{$\cA''$}
{$\cA''$}
{A''}
{
\begin{enumerate}
    \item 
    Simulate $\cA'$ in steps 1.-7. of \cref{hyb2}. 
    If $\bot$ is sent from $\cC$, output $\bot$ and halt. 
    \label{A''1}
    \item
    Measure the register $\regA$ in the computational basis. If the result is 1, output $\bot$ and halt.
    If the result is 0, measure the register $\regB$ of the post-measurement state in the computational basis to get the measurement result $\mu_0$. 
    \label{A''2}
    \item
    Apply $W$ on the post-measurement state 
    and measure the register $\regB$ in the computational basis to get the result $\mu_1$.
    \label{A''3}
    \item
    Output $(\mu_0,\mu_1)$.
\end{enumerate}
}

We show that $\cA''$
wins the game of \cref{hyb2} with a non-negligible probability for infinitely many $\secp$.
The probability that $t\in T_\secp'$ and $\mathsf{Good}$ occur in \cref{A''1} of \cref{A''}
is at least $\frac{1}{16p^2(\secp)}$ for all $\secp\in I'$, because
of the following reasons. First,
$\sum_{t\in T_\secp'}\Pr[t]\ge\frac{1}{4p(\secp)}$ for all $\secp\in I'$ from \cref{why}.
Second, because $t\in T_\secp'$ means $t\in T_\secp$, $\Pr[\mathsf{Good} \mid t]\ge\frac{1}{4p(\secp)}$ for all $\secp\in I$ from
\cref{goodnonnegl}.

Assume that $t\in T_\secp'$ and $\mathsf{Good}$ occur.
If $\cA''$ does the operations in \cref{A''2} and \cref{A''3} on $\ket{\tilde{\phi}_b^{t,good}}$,
the probability that $(\mu_0,\mu_1)=(\Sign(k,0\|r),\Sign(k,1\|r))$ is at least $\frac{1}{4p^2(\secp)}$ for all $\secp\in I'$
from \cref{delta}.
From \cref{why}, the trace distance between
$\ket{\phi_b^{t,good}}$ and $\ket{\tilde{\phi}_b^{t,good}}$ is at most $\frac{1}{p^{10}(\secp)}$ for all $\secp\in I'$.
Therefore,
if $\cA''$ does the operations in \cref{A''2} and \cref{A''3} on $\ket{\phi_b^{t,good}}$,
the probability that $(\mu_0,\mu_1)=(\Sign(k,0\|r),\Sign(k,1\|r))$ is at least $\frac{1}{4p^2(\secp)}-\frac{1}{p^{10}(\secp)}$ for all $\secp\in I'$.
Hence, the overall probability that $\cA''$ outputs $(\mu_0,\mu_1)=(\Sign(k,0\|r),\Sign(k,1\|r))$ is
non-negligible for infinitely many $\secp$.

We prove~\cref{why} to complete this subsection.
\begin{proof}[Proof of \cref{why}]
Fix $t\in T_\secp$.
Immediately before the coherent signature test in \cref{sign} of \cref{hyb1}, the entire state over the registers 
$(\regA,\regB,\regC)$
is generally written as 
$
\sum_{\alpha,\beta}d_{\alpha,\beta}\ket{\alpha}_\regA\ket{\beta}_\regB\ket{\Lambda_{\alpha,\beta}}_{\regC},
$
where $d_{\alpha,\beta}$ are some complex coefficients such that $\sum_{\alpha,\beta}|d_{\alpha,\beta}|^2=1$, 
and $\ket{\Lambda_{\alpha,\beta}}$ are some normalized states.
Define the set
\begin{align}
   S\coloneqq\{(\alpha,\beta):\Ver(\vk,\alpha\|r,\beta)=\top\wedge \beta\neq \Sign(k,\alpha\|r)\}. 
\end{align}
The (unnormalized) state after obtaining $\top$ in the coherent signature test in \cref{sign} of \cref{hyb1}
is
\begin{align}
&d_{0,\Sign(k,0\|r)}\ket{0}_\regA\ket{\Sign(k,0\|r)}_\regB
\ket{\Lambda_{0,\Sign(k,0\|r)}}_{\regC}\nonumber\\
&+d_{1,\sign(k,1\|r)}\ket{1}_\regA\ket{\Sign(k,1\|r)}_\regB
\ket{\Lambda_{1,\Sign(k,1\|r)}}_{\regC}\nonumber\\
&+\sum_{(\alpha,\beta)\in S}d_{\alpha,\beta}\ket{\alpha}_\regA\ket{\beta}_\regB\ket{\Lambda_{\alpha,\beta}}_{\regC}\label{unnormalized}.
\end{align}

Define 
\begin{equation}
T_\secp'\coloneqq\Big\{t\in T_\secp:\sum_{(\alpha,\beta)\in S}|d_{\alpha,\beta}|^2\le\frac{1}{4p^{21}(\secp)}\Big\}.
\end{equation}
If 
$
\sum_{t\in T_\secp\setminus T_\secp'}\Pr[t]\ge\frac{1}{4p(\secp)} 
$
for infinitely many $\secp\in I$, 
it contradicts
the strong EUF-CMA security of the digital signature scheme.
Therefore,
$
\sum_{t\in T_\secp\setminus T_\secp'}\Pr[t]\le\frac{1}{4p(\secp)} 
$
for all $\secp\in I'$, where $I'\coloneqq\{\secp\in I: \secp\ge\secp_0\}$ with a certain $\secp_0$.
This means that
\begin{align}
\sum_{t\in T_\secp'}\Pr[t]&\ge\sum_{t\in T_\secp}\Pr[t]-\frac{1}{4p(\secp)}
\ge\frac{1}{2p(\secp)}-\frac{1}{4p(\secp)}
=\frac{1}{4p(\secp)}
\end{align}
for all $\secp\in I'$.

Moreover, because $t\in T_\secp'$ means $t\in T_\secp$, $\Pr[\mathsf{Good} \mid t]\ge\frac{1}{4p(\secp)}$ for all $\secp\in I$ from \cref{goodnonnegl}.
Therefore,
for any $t\in T_\secp'$,
\begin{align}
|d_{0,\Sign(k,0\|r)}|^2    
+|d_{1,\Sign(k,1\|r)}|^2    
+\sum_{(\alpha,\beta)\in S}|d_{\alpha,\beta}|^2
\ge\frac{1}{4p(\secp)}
\end{align}
for all $\secp\in I$.
If we renormalize the state of \cref{unnormalized} and apply $Z^b$, 
we have
\begin{align}
&\ket{\phi_b^{t,good}}\\
&=\frac{d_{0,\Sign(k,0\|r)}}{\sqrt{|d_{0,\Sign(k,0\|r)}|^2+|d_{1,\Sign(k,1\|r)}|^2+\sum_{(\alpha,\beta)\in S}|d_{\alpha,\beta}|^2}}\ket{0}_\regA\ket{\Sign(k,0\|r)}_\regB
\ket{\Lambda_{0,\Sign(k,0\|r)}}_{\regC}\\
&+(-1)^b\frac{d_{1,\sign(k,1\|r)}}{\sqrt{|d_{0,\Sign(k,0\|r)}|^2+|d_{1,\Sign(k,1\|r)}|^2+\sum_{(\alpha,\beta)\in S}|d_{\alpha,\beta}|^2}}\ket{1}_\regA\ket{\Sign(k,1\|r)}_\regB
\ket{\Lambda_{1,\Sign(k,1\|r)}}_{\regC}\\
&+Z^b\frac{\sum_{(\alpha,\beta)\in S}d_{\alpha,\beta}}{\sqrt{|d_{0,\Sign(k,0\|r)}|^2+|d_{1,\Sign(k,1\|r)}|^2+\sum_{(\alpha,\beta)\in S}|d_{\alpha,\beta}|^2}}\ket{\alpha}_\regA\ket{\beta}_\regB
\ket{\Lambda_{\alpha,\beta}}_{\regC}.
\end{align}
For any $t\in T_\secp'$,
its trace distance to
the state
\begin{align}
&\frac{d_{0,\Sign(k,0\|r)}}{\sqrt{|d_{0,\Sign(k,0\|r)}|^2+|d_{1,\Sign(k,1\|r)}|^2}}\ket{0}_\regA\ket{\Sign(k,0\|r)}_\regB
\ket{\Lambda_{0,\Sign(k,0\|r)}}_{\regC}\\
&+(-1)^b\frac{d_{1,\sign(k,1\|r)}}{\sqrt{|d_{0,\Sign(k,0\|r)}|^2+|d_{1,\Sign(k,1\|r)}|^2}}\ket{1}_\regA\ket{\Sign(k,1\|r)}_\regB
\ket{\Lambda_{1,\Sign(k,1\|r)}}_{\regC}
\end{align}
is less than $\frac{1}{p^{10}(\secp)}$ for all $\secp\in I$.
\end{proof}

\subsection{No QPT Adversary Can Output Two Signatures}
\label{sec:Zhandry}
Here we show that no QPT adversary can win Hybrid 2 (\cref{hyb2}) with a non-negligible probability.
We first give an intuitive argument for the proof, and them give a precise proof.

Intuitive argument for the proof is as follows.
First, note that the probability that all $\{r_i\}_{i=1}^m$ are distinct in \cref{rs} in \cref{hyb2} is at least $1-\negl(\secp)$.
Therefore, we can assume  that all $\{r_i\}_{i=1}^m$ are distinct with a negligible loss in the adversary's winning probability.
If $r\notin\{r_i\}_{i=1}^m$, it is clear that $\cA$ cannot win the game of \cref{hyb2} except for a negligible probability.
The reason is that $\cA$ cannot find $\Sign(k,0\|r)$ or $\Sign(k,1\|r)$ except for a negligible probability due to the security of the digital signature scheme.
Therefore, we assume that $r$ is equal to one of the $\{r_i\}_{i=1}^m$.

Assume that, in the game of \cref{hyb2}, $\cC$ is replaced with $\cC'$ who is the same as $\cC$ except that
it measures the first qubit of $\ket{\psi_r}$ in the computational basis
before sending the states in \cref{send}.
Let $s\in\bit$ be the measurement result.
Then, for any QPT adversary $\cA$, the probability that $\cA$ wins
the game of \cref{hyb2}
is
negligible. The reason is that $\cA$ cannot find $\Sign(k,{s\oplus 1}\|r)$ 
except for a negligible probability due to the strong EUF-CMA security of the digital signature scheme.
From \cref{lem:BZ}, 
we therefore have
\begin{align}
   \Pr[(\Sign(k,0\|r),\Sign(k,1\|r))\gets\cA \mid \cC]&\le 2  
   \Pr[(\Sign(k,0\|r),\Sign(k,1\|r))\gets\cA \mid \cC']\label{ZB}\\
   &\le \negl(\secp), 
\end{align}
where the left-hand-side of \cref{ZB} is the probability that $\cA$ outputs
   $(\Sign(k,0\|r),\Sign(k,1\|r))$ with the challenger $\cC$,
   and the right-hand-side is that with the challenger $\cC'$.

We give a precise proof below.
Let $\mathsf{Alg}$ be an algorithm that, on input $(r_1,...,r_m)$, simulates $\cC$ and $\cA$
in \cref{hyb2} and outputs $(r,\mu_0,\mu_1)$.
The probability that $\cA$ wins in the game of \cref{hyb2} is
  \begin{align}
&\frac{1}{2^{um}}\sum_{r_1,...,r_m}\sum_r \Pr[(r,\Sign(k,0\|r),\Sign(k,1\|r))\gets\mathsf{Alg}(r_1,...,r_m)]      \\
&=
\frac{1}{2^{um}}\sum_{(r_1,...,r_m)\in R}\sum_r \Pr[(r,\Sign(k,0\|r),\Sign(k,1\|r))\gets\mathsf{Alg}(r_1,...,r_m)]      \\
&+\frac{1}{2^{um}}\sum_{(r_1,...,r_m)\notin R}\sum_r \Pr[(r,\Sign(k,0\|r),\Sign(k,1\|r))\gets\mathsf{Alg}(r_1,...,r_m)]      \\
&\le
\frac{1}{2^{um}}\sum_{(r_1,...,r_m)\in R}\sum_r \Pr[(r,\Sign(k,0\|r),\Sign(k,1\|r))\gets\mathsf{Alg}(r_1,...,r_m)]      
+\frac{1}{2^{um}}\sum_{(r_1,...,r_m)\notin R}      \\
&\le
\frac{1}{2^{um}}\sum_{(r_1,...,r_m)\in R}\sum_r \Pr[(r,\Sign(k,0\|r),\Sign(k,1\|r))\gets\mathsf{Alg}(r_1,...,r_m)]      
+\frac{(m-1)m}{2^{u}}\\
&=
\frac{1}{2^{um}}\sum_{(r_1,...,r_m)\in R}\sum_{r\in\{r_i\}_{i=1}^m} 
\Pr[(r,\Sign(k,0\|r),\Sign(k,1\|r))\gets\mathsf{Alg}(r_1,...,r_m)]      \\
&+\frac{1}{2^{um}}\sum_{(r_1,...,r_m)\in R}\sum_{r\notin\{r_i\}_{i=1}^m} 
\Pr[(r,\Sign(k,0\|r),\Sign(k,1\|r))\gets\mathsf{Alg}(r_1,...,r_m)]      
+\frac{(m-1)m}{2^{u}}\\
&\le
\frac{1}{2^{um}}\sum_{(r_1,...,r_m)\in R}\sum_{r\in\{r_i\}_{i=1}^m} 
\Pr[(r,\Sign(k,0\|r),\Sign(k,1\|r))\gets\mathsf{Alg}(r_1,...,r_m)]      \\
&+\frac{1}{2^{um}}\sum_{(r_1,...,r_m)\in R}\negl(\secp) 
+\frac{(m-1)m}{2^{u}}\label{EUC1}\\
&\le
\frac{1}{2^{um}}\sum_{(r_1,...,r_m)\in R}\sum_{r\in\{r_i\}_{i=1}^m} 
\Pr[(r,\Sign(k,0\|r),\Sign(k,1\|r))\gets\mathsf{Alg}(r_1,...,r_m)]      \\
&+\negl(\secp) 
+\frac{(m-1)m}{2^{u}}\\
&\le
\frac{1}{2^{um}}\sum_{(r_1,...,r_m)\in R}\sum_{r\in\{r_i\}_{i=1}^m} 
2\mbox{$\Pr$}'[(r,\Sign(k,0\|r),\Sign(k,1\|r))\gets\mathsf{Alg}(r_1,...,r_m)]      \label{BZ}\\
&+\negl(\secp) 
+\frac{(m-1)m}{2^{u}}\\
&\le
\frac{1}{2^{um}}\sum_{(r_1,...,r_m)\in R}\sum_{r\in\{r_i\}_{i=1}^m} 
\negl(\secp) 
+\negl(\secp) 
+\frac{(m-1)m}{2^{u}}\label{EUC2}\\
&\le
\negl(\secp) 
+\negl(\secp) 
+\frac{(m-1)m}{2^{u}}\\
&=\negl(\secp).
  \end{align} 
  Here, $R\coloneqq\{(r_1,...,r_m):\mbox{All of them are distinct}\}$.
  In \cref{EUC1}, we have used the strong EUF-CMA security of the digital signature scheme.
  $\Pr'$ is the probability that,
  in $\mathsf{Alg}$, $\cC$ is replaced with $\cC'$ who is the same as $\cC$ except that
it measures the first qubit of $\ket{\psi_r}$ in the computational basis
before sending the states in \cref{send}.
\cref{BZ} comes from \cref{lem:BZ}. 
\cref{EUC2} is from the strong EUF-CMA security of the digital signature scheme.

%% file: def_cca.tex

\section{Definition of Chosen Ciphertext Security}\label{sec:def_cca}

In this section, we define CCA security for QPKE and related security notions. \ifnum\llncs=1(Though we already defined IND-pkT-CCA security in \Cref{def_mainbody_indccapkta}, we restate it here again for the reader's convenience.) \fi
We start with an intermediate notion between CPA security and CCA security that we call security against challenge validity attack (CVA).

\begin{definition}[$\INDCVApkTA$ security]\label{def_indcvpkta}
For any polynomial $m$, and any QPT adversary $\cA$, we have
\begin{equation}
   \Pr\left[b\gets\cA(\ct^\ast,\cv,\st):
   \begin{array}{r}
   (\sk,\vk)\gets\SKGen(1^\secp)\\
   \pk_1,...,\pk_m\gets\PKGen(\sk)^{\otimes m}\\
   (\pk',\msg_0,\msg_1,\st)\gets\cA(\vk,\pk_1,...,\pk_m)\\
   b\gets \bit\\ 
   \ct^\ast\gets\Enc(\vk,\pk',\msg_b)\\
   \cv:=0\textrm{~if~}\Dec(\sk,\ct^*)=\bot\textrm{~and otherwise~}\cv:=1
   \end{array}
   \right] \le \frac{1}{2}+\negl(\secp).
\end{equation}
Here, $\pk_1,...,\pk_m\gets\PKGen(\sk)^{\otimes m}$
means that $\PKGen$ is executed $m$ times and $\pk_i$ is the output of the $i$th execution of $\PKGen$.
$\st$ is a quantum internal state of $\cA$, which can be entangled with $\pk'$.
\end{definition}

We then define CCA security for QPKE.

\begin{definition}[$\INDCCApkTA$ security]\label{def_indccapkta}
For any polynomial $m$, and any QPT adversary $\cA$, we have
\begin{equation}
   \Pr\left[b\gets\cA^{\ODec{2}(\cdot)}(\ct^\ast,\cv,\st):
   \begin{array}{r}
   (\sk,\vk)\gets\SKGen(1^\secp)\\
   \pk_1,...,\pk_m\gets\PKGen(\sk)^{\otimes m}\\
   (\pk',\msg_0,\msg_1,\st)\gets\cA^{\ODec{1}(\cdot)}(\vk,\pk_1,...,\pk_m)\\
   b\gets \bit\\ 
   \ct^\ast\gets\Enc(\vk,\pk',\msg_b)\\
   \cv:=0\textrm{~if~}\Dec(\sk,\ct^*)=\bot\textrm{~and otherwise~}\cv:=1
   \end{array}
   \right] \le \frac{1}{2}+\negl(\secp).
\end{equation}
Here, $\pk_1,...,\pk_m\gets\PKGen(\sk)^{\otimes m}$
means that $\PKGen$ is executed $m$ times and $\pk_i$ is the output of the $i$th execution of $\PKGen$.
$\st$ is a quantum internal state of $\cA$, which can be entangled with $\pk'$.
Also, $\ODec{1}(\ct)$ returns $\Dec(\sk,\ct)$ for any $\ct$.
$\ODec{2}$ behaves identically to $\ODec{1}$ except that $\ODec{2}$ returns $\bot$ to the input $\ct^*$.
\end{definition}

\begin{definition}[$\INDoneCCApkTA$ security]\label{def_ind1ccapkta}
$\INDoneCCApkTA$ security is defined in the same way as $\INDCCApkTA$ security except that in the security game we require that the total number of $\cA$'s query to $\ODec{1}$ and $\ODec{2}$ is at most one.
\end{definition}

\paragraph{Security under single public key.}
For $X\in\{\INDpkTA,\INDCVApkTA,\INDCCApkTA,\allowbreak \INDoneCCApkTA\}$ security, we define $\onekey{X}$ security as its variant where the number of public keys given to the adversary is fixed to one.
Note that $\onekey{X}$ security is implied by $X$ security for any $X\in\{\INDpkTA,\INDCVApkTA,\INDCCApkTA,\allowbreak\INDoneCCApkTA\}$.

\medskip

We also define statistical variant of decryption error detectability that is useful to achieve CCA security with our transformations.

\begin{definition}[Strong decryption error detectability]
\label{def:dec_err_new}
We say that a QPKE scheme has strong decryption error detectability if for 
any $\sk^\prime,\vk^\prime,\pk^\prime$, and $\msg$, we have
\begin{equation}
   \Pr\left[
   \msg'\neq \bot~\land~ 
   \msg'\neq \msg :
   \begin{array}{r}
   \ct\gets\Enc(\vk^\prime,\pk^\prime,\msg)\\
   \msg^\prime \gets\Dec(\sk^\prime,\ct)
   \end{array}
   \right] \le \negl(\secp).
\end{equation}
\end{definition}

%% file: transformations_cca.tex

\section{Transformations Achieving Chosen Ciphertext Security}\label{sec:cca_const}

In this section, we present the transformation from CPA secure QPKE scheme to CCA secure one. 
\ifnum\llncs=1 We remark that the transformation makes use of the intermediate security notions defined in \Cref{sec:def_cca}. \fi
The transformation consists of the following four subroutines.
\begin{enumerate}
\item Transformation from $\INDpkTCPA$ secure one to $\INDpkTCVA$ secure one presented in~\cref{sec:indcvapkta}.
\item Transformation from $\onekey{\INDpkTCVA}$ secure one to $\onekey{\INDpkToneCCA}$ secure one presented in~\cref{sec:indoneccapkta}.
\item Transformation from $\onekey{\INDpkToneCCA}$ secure one to $\onekey{\INDpkTCCA}$ secure one presented in \cref{sec:indccapkta}.
\item Transformation from $\onekey{\INDpkTCCA}$ secure one to $\INDpkTCCA$ secure one presented in \cref{sec:onekey_mkey}.
\end{enumerate}
Below, we first introduce the notion of tokenized MAC~\cite{EPRINT:BehSatShi21} in \cref{sec:ccaprep} that is used in the third transformation, and then provide each transformations.

\subsection{Preparations}\label{sec:ccaprep}
\begin{definition}[Tokenized MAC \cite{EPRINT:BehSatShi21}]\label{def:TMAC}
A tokenized MAC scheme with the message space $\bit^\ell$ is a set of algorithms 
$(\SKGen,\TKGen,\Sign,\Ver)$ such that
\begin{itemize}
    \item 
    $\SKGen(1^\secp)\to \sk:$
    It is a PPT algorithm that, on input the security parameter $\secp$, outputs
    a classical secret key $\sk$. 
    \item
    $\TKGen(\sk)\to \token:$
    It is a QPT algorithm that, on input $\sk$, outputs
    a quantum signing token $\token$.
    \item
    $\Sign(\token,\msg)\to \sigma:$ 
    It is a QPT algorithm that, on input $\token$ and a message $\msg\in\bit^\ell$,  outputs a classical signature $\sigma$. 
    \item
    $\Ver(\sk,\msg,\sigma)\to \top/\bot:$ 
    It is a classical deterministic polynomial-time algorithm that, on input $\sk$, $\msg$, and $\sigma$, outputs $\top$ or $\bot$.
\end{itemize}
We require the following correctness and unforgeability.

\paragraph{\bf Correctness:}
For any $\msg$,
\ifnum\llncs=0
\begin{equation}
   \Pr[\top\gets\Ver(\sk,\msg,\sigma):
   \sk\gets\SKGen(1^\secp),
   \token\gets\TKGen(\sk),
   \sigma\gets\Sign(\token,\msg)
   ]\ge1-\negl(\secp). 
\end{equation}
\else
\begin{equation}
   \Pr\left[\top\gets\Ver(\sk,\msg,\sigma):
   \begin{array}{r}
   \sk\gets\SKGen(1^\secp)\\
   \token\gets\TKGen(\sk)\\
   \sigma\gets\Sign(\token,\msg)
   \end{array}
   \right]\ge1-\negl(\secp). 
\end{equation}
\fi

\paragraph{\bf Unforgeability:}
For any QPT adversary $\cA$ with classical oracle access to the verification oracle $\Ver(\sk,\cdot,\cdot)$,
\begin{equation}
   \Pr\left[
   \begin{array}{r}
   \msg_1\ne\msg_2\\
   \land~ \top\gets\Ver(\sk,\msg_1,\sigma_1)\\
   \land~ \top\gets\Ver(\sk,\msg_2,\sigma_2)
    \end{array}
   :
   \begin{array}{r}
   \sk\gets\SKGen(1^\secp),\\
   \token\gets\TKGen(\sk),\\
   (\msg_a,\sigma_a)_{a\in[2]}\gets\cA^{\Ver(\sk,\cdot,\cdot)}(\token)\\
   \end{array}
   \right]\le\negl(\secp). 
\end{equation}
\end{definition}

\medskip

\begin{theorem}[\cite{EPRINT:BehSatShi21}]
Tokenized MAC exists if OWFs exist.
\end{theorem}

Note that the unforgeability in the above definition is weaker than that in the original definition by \cite{EPRINT:BehSatShi21}.
We use this weaker definition that is sufficient for our purpose for ease of exposition.

\subsection{$\INDCVApkTA$ Secure QPKE via Cut-and-Choose}\label{sec:indcvapkta}

We show a generic construction of $\INDCVApkTA$ secure QPKE from $\INDpkTA$ secure QPKE using the cut-and-choose technique.  

Let $\QPKE=(\QPKE.\SKGen,\QPKE.\PKGen,\QPKE.\Enc,\QPKE.\Dec)$ be a QPKE scheme with message space $\bit^\ell$. Then we construct a QPKE scheme $\CVA=(\CVA.\SKGen,\CVA.\PKGen,\CVA.\Enc,\CVA.\Dec)$ with message space $\bit^\ell$ as follows:
\begin{itemize}
    \item 
    $\CVA.\SKGen(1^\secp)\to (\sk,\vk):$
    Run $(\sk_i,\vk_i)\gets \QPKE.\SKGen(1^\secp)$ for every $i\in[4\secp]$. 
    Output $\sk:=(\sk_i)_{i\in[4\secp]}$ and $\vk:=(\vk_i)_{i\in[4\secp]}$.
    \item
    $\CVA.\PKGen(\sk)\to \pk:$
    Parse $\sk:=(\sk_i)_{i\in[4\secp]}$.
    Run $\pk_i \gets \QPKE.\PKGen(\sk_i)$ for $i\in[4\secp]$ and outputs $\pk:=(\pk_i)_{i\in[4\secp]}$. 
    \item
    $\CVA.\Enc(\vk,\pk,\msg)\to \ct:$
    Parse $\vk:=(\vk_i)_{i\in[4\secp]}$ and $\pk:=(\pk_i)_{i\in[4\secp]}$.
    Generate a random $2\secp$ size subset $\Test$ of $[4\secp]$. 
    Generate $u_i\la\bit^\ell$, run $\ct_i \gets \QPKE.\Enc(\vk_i,\pk_i,u_i)$ for every $i\in[4\secp]$.
    Set $v_i:=u_i$ if $i\in\Test$ and $v_i:=u_i\oplus\msg$ otherwise.
    Output $\ct:=(\Test, (\ct_i,v_i)_{i\in[4\secp]})$.
    \item
    $\CVA.\Dec(\sk,\ct)\to \msg:$
    Parse $\sk:=(\sk_i)_{i\in[4\secp]}$ and $\ct=(\Test,(\ct_i,v_i)_{i\in[4\secp]})$.
    Output $\bot$ if $\QPKE.\Dec(\sk_i,\ct_i)\ne v_i$ for some $i\in\Test$.
    Otherwise, run $u_i\la\QPKE.\Dec(\sk_i,\ct_i)$ and compute $\msg_i:=v_i\oplus u_i$ for every $i\in[4\secp]\setminus\Test$, and output most frequently appeared $\msg$. (If there are multiple such $\msg$, output the lexicographically first one.)
\end{itemize}

\paragraph{Correctness.}
Correctness of $\CVA$ immediately follows from correctness of $\QPKE$. 

\paragraph{Strong decryption error detectability.}
Let $(\sk^\prime,\vk^\prime,\pk^\prime,\msg)$ be any tuple of a secret key, verification key, public key, and message, where $\sk^\prime:=(\sk_i^\prime)_{i\in[4\secp]}$, $\vk^\prime:=(\vk_i^\prime)_{i\in[4\secp]}$, $\pk^\prime:=(\pk_i^\prime)_{i\in[4\secp]}$, and $\msg\in\bit^\ell$.
Suppose we pick $u_i\la\bit^\ell$, generate $\ct_i\la\QPKE.\Enc(\vk_i^\prime,\pk_i^\prime,u_i)$, and compute $u_i^\prime\la\QPKE.\Dec(\sk^\prime_i,\ct_i)$ for every $i\in[4\secp]$.
We consider the following two cases.
\begin{itemize}
\item The first case is $u_i\ne u_i^\prime$ for more than $\secp$ indices. In this case, a randomly chosen $2\secp$ size subset $\Test$ includes at least one index $i$ such that $u_i\ne u_i^\prime$ and thus the decryption result of $\ct:=(\Test,(\ct_i,v_i)_{i\in[4\secp]})$ for randomly chosen $\Test$ is $\bot$ with overwhelming probability, where $v_i:=u_i$ if $i\in\Test$ and otherwise $v_i:=u_i\oplus\msg$.
\item The second case is $u_i\ne u_i^\prime$ for less than $\secp$ indices.
In this case, for every choice of $\Test$, $\msg$ occupies the majority among $\msg_i:=\msg\oplus u_i \oplus u_i^\prime$ for $i\in[4\secp]\setminus\Test$.
Thus, the decryption result of $\ct:=(\Test, (\ct_i,v_i)_{i\in[4\secp]})$ is either $\bot$ or $\msg$, regardless of the choice of $\Test$.
\end{itemize}
This proves the strong decryption error detectability of $\CVA$.

\paragraph{$\INDCVApkTA$ security.}
We prove that if $\QPKE$ satisfies $\INDpkTA$ security, then $\CVA$ satisfies $\INDCVApkTA$ security.
We consider the following games.

\begin{description}
\item[$\hybi{0}$:]This is the original security experiment for the $\INDCVApkTA$ security of $\CVA$ played between $\cA$ and the challenger.
The detailed description is as follows.
\begin{enumerate}
\item The challenger generates $(\sk_i,\vk_i)\gets \QPKE.\SKGen(1^\secp)$ for every $i\in[4\secp]$, and set $\sk:=(\sk_i)_{i\in[4\secp]}$ and $\vk:=(\vk_i)_{i\in[4\secp]}$. The challenger generates $\pk_i \gets \QPKE.\PKGen(\sk_i)$ for $i\in[4\secp]$ and set $\pk:=(\pk_i)_{i\in[4\secp]}$.

\item The challenger runs $(\pk^\prime,\msg_0,\msg_1,\st)\la\cA(\vk,\pk)$.

\item The challenger parses $\pk^\prime:=(\pk_i^\prime)_{i\in[4\secp]}$ and picks $b\la\bit$. The challenger generates $\ct^*$ as follows.
\begin{itemize}
\item Generate a random $2\secp$ size subset $\Test^*$ of $[4\secp]$. 
\item Generate $u_i^*\la\bit^\ell$, run $\ct_i^* \gets \QPKE.\Enc(\vk_i,\pk_i^\prime,u_i^*)$ for every $i\in[4\secp]$.
\item Set $v_i^*:=u_i^*$ if $i\in\Test^*$ and $v_i^*:=u_i^*\oplus\msg_b$ otherwise.
\item Output $\ct^*:=(\Test^*, (\ct_i^*,v_i^*)_{i\in[4\secp]})$.
\end{itemize}
The challenger also sets $\cv:=0$ if $\CVA.\Dec(\sk,\ct^*)=\bot$ and otherwise sets $\cv:=1$.

\item The challenger runs $b^\prime\la\cA(\cv,\ct^*,\st)$.
The challenger outputs $1$ if $b=b^\prime$ and otherwise outputs $0$.

\end{enumerate}

\item[$\hybi{1}$:] This is the same as $\hybi{0}$ except that the challenger generates $\ct_i^*\la\QPKE.\Enc(\vk_i,\pk_i^\prime,0^\ell)$ for every $i\in[4\secp]\setminus\Test^*$.
\end{description}

We can prove $\abs{\Pr[1\la\hybi{0}]-\Pr[1\la\hybi{1}]}=\negl(\secp)$ using the $\INDpkTA$ security of $\QPKE$ with respect to instances such that the corresponding index $i$ is not included in $\Test^*$.
Note that the reduction needs to know whether $\CVA.\Dec(\sk,\ct^*)=\bot$ or not.
This is possible since it can be computed with only $\sk_i$ for $i\in\Test^*$, which is generated by the reduction itself.

In $\hybi{1}$, the challenge bit $b$ is completely hidden from the view of $\cA$ since $b$ is masked by $u_i$ for $i\in[4\secp]\setminus\Test^*$.
Thus, we have $\Pr[1\la\hybi{2}]=\frac{1}{2}$.
From the above discussions, $\CVA$ satisfies $\INDCVApkTA$ security.

\subsection{$\INDoneCCApkTA$ Secure QPKE from $\INDCVApkTA$ Secure One}\label{sec:indoneccapkta}

We show how to construct $\onekey{\INDoneCCApkTA}$ secure QPKE from $\onekey{\INDCVApkTA}$ secure one.  

Let $\CVA=(\CVA.\SKGen,\CVA.\PKGen,\CVA.\Enc,\CVA.\Dec)$ be a QPKE scheme with message space $\bit^\ell$ and 
$\SIG=(\SIG.\Gen,\SIG.\Sign,\SIG.\Ver)$ be a digital signature scheme whose verification key is of length $\sigvklen$. Then we construct a QPKE scheme $\oneCCA=(\oneCCA.\SKGen,\oneCCA.\PKGen,\oneCCA.\Enc,\oneCCA.\Dec)$ with message space $\bit^\ell$ as follows, where for a verification key $\sigvk$ of $\SIG$ and an integer $i\in[\sigvklen]$, $\sigvk[i]$ denotes the $i$-th bit of $\sigvk$:
\begin{itemize}
    \item 
    $\oneCCA.\SKGen(1^\secp)\to (\sk,\vk):$
    Run $(\cva.\sk_{i,\alpha},\cva.\vk_{i,\alpha})\gets \CVA.\SKGen(1^\secp)$ for every $i\in[\sigvklen]$ and $\alpha\in\bit$. 
    Output $\sk:=(\cva.\sk_{i,\alpha})_{i\in[\sigvklen],\alpha\in\bit}$ and $\vk:=(\cva.\vk_{i,\alpha})_{i\in[\sigvklen],\alpha\in\bit}$.
    \item
    $\oneCCA.\PKGen(\sk)\to \pk:$
    Parse $\sk:=(\cva.\sk_{i,\alpha})_{i\in[\sigvklen],\alpha\in\bit}$.
    Run $\cva.\pk_{i,\alpha} \gets \CVA.\PKGen(\cva.\sk_{i,\alpha})$ for every $i\in[\sigvklen]$ and $\alpha\in\bit$. Output $\pk:=(\cva.\pk_{i,\alpha})_{i\in[\sigvklen],\alpha\in\bit}$. 
    \item
    $\oneCCA.\Enc(\vk,\pk,\msg)\to \ct:$
    Parse $\vk:=(\cva.\vk_{i,\alpha})_{i\in[\sigvklen],\alpha\in\bit}$ and $\pk:=(\cva.\pk_{i,\alpha})_{i\in[\sigvklen],\alpha\in\bit}$.
    Run $(\sigvk,\sigk)\la\SIG.\Gen(1^\secp)$.
    Generate $u_i\la\bit^\ell$ for every $i\in[\sigvklen-1]$ and set $u_{\sigvklen}:=\msg\oplus\bigoplus_{i\in[\sigvklen-1]}u_i$.
    Run $\cva.\ct_i \gets \CVA.\Enc(\cva.\vk_{i,\sigvk[i]},\cva.\pk_{i,\sigvk[i]},u_i)$ for every $i\in[\sigvklen]$.
    Run $\sigma\la\SIG.\Sign(\sigk,\cva.\ct_1\|\cdots\|\cva.\ct_\sigvklen)$.
    Output $\ct:=(\sigvk, (\cva.\ct_i)_{i\in[\sigvklen]},\sigma)$.
    \item
    $\oneCCA.\Dec(\sk,\ct)\to \msg:$
    Parse $\sk:=(\cva.\sk_{i,\alpha})_{i\in[\sigvklen],\alpha\in\bit}$ and $\ct=(\sigvk,(\cva.\ct_i)_{i\in[\sigvklen]},\sigma)$.
    Output $\bot$ if $\SIG.\Ver(\sigvk,\cva.\ct_1\|\cdots\|\cva.\ct_{\sigvklen},\sigma)=\bot$ and otherwise go to the next step.
    Run $u_i\la\CVA.\Dec(\cva.\sk_{i,\sigvk[i]},\cva.\ct_i)$ for every $i\in[\sigvklen]$, and output $\bot$ if $u_i=\bot$ for some $i\in[\sigvklen]$.
    Otherwise, output $\bigoplus_{i\in[\sigvklen]}u_i$.
\end{itemize}

\paragraph{Correctness and strong decryption error detectability.}
The correctness and the strong decryption error detectability of $\oneCCA$ immediately follow from those of $\CVA$ and the correctness of $\SIG$. 

\paragraph{$\onekey{\INDoneCCApkTA}$ security.}
We prove that if $\CVA$ satisfies $\onekey{\INDCVApkTA}$ security and $\SIG$ satisfies strong unforgeability, then $\oneCCA$ satisfies $\onekey{\INDoneCCApkTA}$ security.

Let $\cA$ be any QPT adversary attacking the $\onekey{\INDoneCCApkTA}$ security of $\oneCCA$.
Without loss of generality, we assume that $\cA$ makes exactly one decryption query.
We proceed the proof using a sequence of experiments.

\begin{description}
\item[$\hybi{0}$:]This is the original security experiment for the $\onekey{\INDoneCCApkTA}$ security of $\oneCCA$ played between $\cA$ and the challenger.
The detailed description is as follows.
\begin{enumerate}
\item The challenger generates $(\cva.\sk_{i,\alpha},\cva.\vk_{i,\alpha})\gets \CVA.\SKGen(1^\secp)$ for every $i\in[\sigvklen]$ and $\alpha\in\bit$, and sets $\sk:=(\cva.\sk_{i,\alpha})_{i\in[\sigvklen],\alpha\in\bit}$ and $\vk:=(\cva.\vk_{i,\alpha})_{i\in[\sigvklen],\alpha\in\bit}$. The challenger generates $\cva.\pk_{i,\alpha} \gets \CVA.\PKGen(\cva.\sk_{i,\alpha})$ for $i\in[\sigvklen]$ and $\alpha\in\bit$, 
and sets $\pk:=(\cva.\pk_{i,\alpha})_{i\in[\sigvklen],\alpha\in\bit}$.

\item The challenger runs $(\pk^\prime,\msg_0,\msg_1,\st)\la\cA(\vk,\pk)^{\ODec{1}(\cdot)}$, where $\ODec{1}(\ct)$ behaves as follows.
\begin{itemize}
\item Parse $\ct=(\sigvk,(\cva.\ct_i)_{i\in[\sigvklen]},\sigma)$.
\item Output $\bot$ if $\SIG.\Ver(\sigvk,\cva.\ct_1\|\cdots\|\cva.\ct_{\sigvklen},\sigma)=\bot$ and otherwise go to the next step.
\item Run $u_i\la\CVA.\Dec(\cva.\sk_{i,\sigvk[i]},\cva.\ct_i)$ for every $i\in[\sigvklen]$, and output $\bot$ if $u_i=\bot$ for some $i\in[\sigvklen]$.
\item Otherwise, output $\bigoplus_{i\in[\sigvklen]}u_i$.
\end{itemize}

\item The challenger parses $\pk^\prime:=(\cva.\pk_{i,\alpha}^\prime)_{i\in[\sigvklen],\alpha\in\bit}$ and picks $b\la\bit$. The challenger generates $\ct^*$ as follows.
\begin{itemize}
\item Run $(\sigvk^*,\sigk^*)\la\SIG.\Gen(1^\secp)$.
\item Generate $u_i^*\la\bit^\ell$ for every $i\in[\sigvklen-1]$ and set $u_{\sigvklen}^*:=\msg_b\oplus\bigoplus_{i\in[\sigvklen-1]}u_i^*$.
\item Run $\cva.\ct_i^* \gets \CVA.\Enc(\cva.\vk_{i,\sigvk^*[i]},\cva.\pk_{i,\sigvk^*[i]}^\prime,u_i^*)$ for every $i\in[\sigvklen]$.
\item Run $\sigma^*\la\SIG.\Sign(\sigk^*,\cva.\ct_1^*\|\cdots\|\cva.\ct_\sigvklen^*)$ and set $\ct^*:=(\sigvk^*,(\cva.\ct_i^*)_{i\in[\sigvklen]},\sigma^*)$.
\end{itemize}
The challenger also sets $\cv:=0$ if $\oneCCA.\Dec(\sk,\ct^*)=\bot$ and otherwise sets $\cv:=1$.

\item The challenger runs $b^\prime\la\cA(\cv,\ct^*,\st)^{\ODec{2}(\cdot)}$, where $\ODec{2}$ behaves exactly in the same way as $\ODec{1}$ except that $\ODec{2}$ given $\ct$ returns $\bot$ if $\ct=\ct^*$.
The challenger outputs $1$ if $b=b^\prime$ and otherwise outputs $0$.

\end{enumerate}
\item[$\hybi{1}$:] This is the same as $\hybi{0}$ except that the challenger generates the key pair $(\sigvk^*,\sigk^*)$ of $\SIG$ that is used to generate the challenge ciphertext at the beginning of the game, and $\ODec{1}$ and $\ODec{2}$ behave as follows.
\begin{itemize}
\item Parse $\ct=(\sigvk,(\cva.\ct_i)_{i\in[\sigvklen]},\sigma)$.
\item Output $\bot$ if $\sigvk=\sigvk^*$ and otherwise go to the next step.
\item Output $\bot$ if $\SIG.\Ver(\sigvk,\cva.\ct_1\|\cdots\|\cva.\ct_{\sigvklen},\sigma)=\bot$ and otherwise go to the next step.
\item Run $u_i\la\CVA.\Dec(\cva.\sk_{i,\sigvk[i]},\cva.\ct_i)$ for every $i\in[\sigvklen]$, and output $\bot$ if $u_i=\bot$ for some $i\in[\sigvklen]$.
\item Otherwise, output $\bigoplus_{i\in[\sigvklen]}u_i$.
\end{itemize}
\end{description}

We define the following two events.
\begin{description}
\item[$\Forge_{j,1}$:]In $\hybi{j}$, $\cA$ queries $\ct=(\sigvk,(\cva.\ct_i)_{i\in[\sigvklen]},\sigma)$ to $\ODec{1}$ such that $\sigvk=\sigvk^*$ and $\SIG.\Ver(\sigvk,\allowbreak\cva.\ct_1\|\cdots\|\cva.\ct_{\sigvklen},\sigma)=\top$.
\item[$\Forge_{j,2}$:]In $\hybi{j}$, $\cA$ queries $\ct=(\sigvk,(\cva.\ct_i)_{i\in[\sigvklen]},\sigma)$ to $\ODec{2}$ such that $\sigvk=\sigvk^*$, $\ct\ne\ct^*$, and $\SIG.\Ver(\sigvk,\allowbreak\cva.\ct_1\|\cdots\|\cva.\ct_{\sigvklen},\sigma)=\top$.
\end{description}
We also let $\Forge_j=\Forge_{j,1}\lor\Forge_{j,2}$.
$\hybi{0}$ and $\hybi{1}$ are identical games unless the events $\Forge_0$ and $\Forge_1$ happen in $\hybi{0}$ and $\hybi{1}$, respectively.
Thus, we have $\Pr[1\la\hybi{0}\land\lnot\Forge_0]=\Pr[1\la\hybi{1}\land\lnot\Forge_1]$ and $\Pr[\Forge_0]=\Pr[\Forge_1]$.
Then, we have
\begin{align}
\abs{\Pr[1\la\hybi{0}]-\Pr[1\la\hybi{1}]}
&\le \abs{\Pr[1\la\hybi{0}\land\Forge_0]-\Pr[1\la\hybi{1}\land\Forge_1]}\\
&+\abs{\Pr[1\la\hybi{0}\land\lnot\Forge_0]-\Pr[1\la\hybi{1}\land\lnot\Forge_1]}\\
&=\abs{\Pr[1\la\hybi{0}\land\Forge_0]-\Pr[1\la\hybi{1}\land\Forge_1]}\\
&\le \Pr[\Forge_1]\cdot\abs{\Pr[1\la\hybi{0}|\Forge_0]-\Pr[1\la\hybi{1}|\Forge_1]}\\
&\le \Pr[\Forge_1]\\
&\le\Pr[\Forge_{1,1}]+\Pr[\Forge_{1,2}].
\end{align}
From the strong unforgeability of $\SIG$, we have $\Pr[\Forge_{1,1}]\le\negl(\secp)$ and $\Pr[\Forge_{1,2}]\le\negl(\secp)$, and thus obtain $\abs{\Pr[1\la\hybi{0}]-\Pr[1\la\hybi{1}]}\le\negl(\secp)$.

\begin{description}
\item[$\hybi{2}$:] This is the same as $\hybi{1}$ except that the challenger generates 
\begin{align}
\cva.\ct_{i^*}^*\la\CVA.\Enc(\cva.\vk_{i^*,\sigvk^*[i^*]},\cva.\pk_{i^*,\sigvk^*[i^*]}^\prime,0^\ell)
\end{align}
for randomly chosen $i^*\la[\sigvklen]$.
\end{description}

To estimate $\abs{\Pr[1\la\hybi{1}]-\Pr[1\la\hybi{2}]}$, we construct the following adversary $\cB$ that uses $\cA$ and attacks the $\onekey{\INDCVApkTA}$ security of $\CVA$.
\begin{enumerate}
\item Given $(\cva.\vk,\cva.\pk)$, $\cB$ generates $(\sigvk^*,\sigk^*)\la\SIG.\Gen(1^\secp)$, picks $i^*\la[\sigvklen]$, and sets $\cva.\vk_{i^*,\sigvk^*[i^*]}:=\cva.\vk$ and $\cva.\pk_{i^*,\sigvk^*[i^*]}:=\cva.\pk$.
$\cB$ generates $(\cva.\sk_{i,\alpha},\cva.\vk_{i,\alpha})\gets \CVA.\SKGen(1^\secp)$ and $\cva.\pk_{i,\alpha} \gets \CVA.\PKGen(\cva.\sk_{i,\alpha})$ for every $(i,\alpha)\in[\sigvklen]\times\bit\setminus\{(i^*,\sigvk^*[i^*])\}$.
$\cB$ sets $\vk:=(\cva.\vk_{i,\alpha})_{i\in[\sigvklen],\alpha\in\bit}$ and $\pk:=(\cva.\pk_{i,\alpha})_{i\in[\sigvklen],\alpha\in\bit}$.

\item $\cB$ runs $(\pk^\prime,\msg_0,\msg_1,\st)\la\cA(\pk,\vk)^{\ODec{1}(\cdot)}$, where $\ODec{1}(\ct)$ is simulated as follows.
\begin{itemize}
\item Parse $\ct:=(\sigvk,(\cva.\ct_i)_{i\in[\sigvklen]},\sigma)$.
\item If $\sigvk=\sigvk^*$ or $\SIG.\Ver(\sigvk,\cva.\ct_1\|\cdots\|\cva.\ct_{\sigvklen},\sigma)=\bot$, return $\bot$ with probability $\frac{1}{\sigvklen}$ and abort with output $\beta^\prime=0$ with probability $\frac{\sigvklen-1}{\sigvklen}$.
Otherwise, go to the next step.
\item If $i^*$ is not the smallest index $i$ such that $\sigvk[i]\ne\sigvk^*[i]$, abort with output $\beta^\prime=0$. Otherwise, go to the next step.
\item Return $\oneCCA.\Dec(\sk,\ct)$. Note that in this case, $\cB$ can compute $\oneCCA.\Dec(\sk,\ct)$ by using $\cva.\sk_{i,\alpha}$ for $(i,\alpha)\in[\sigvklen]\times\bit\setminus\{(i^*,\sigvk^*[i^*])\}$.
\end{itemize}

\item $\cB$ parses $\pk^\prime:=(\cva.\pk_{i,\alpha}^\prime)_{i\in[\sigvklen],\alpha\in\bit}$ and picks $b\la\bit$. $\cB$ generates $\ct^*$ as follows.
\begin{itemize}
\item Generate $u_i^*\la\bit^\ell$ for every $i\in[\sigvklen]\setminus\{i^*\}$ and set $u_{i^*}^*:=\msg_b\oplus\bigoplus_{i\in[\sigvklen]\setminus\{i^*\}}u_i^*$. 
\item Run $\cva.\ct_i^* \gets \CVA.\Enc(\cva.\vk_{i,\sigvk^*[i]},\cva.\pk_{i,\sigvk^*[i]},u_i^*)$ for every $i\in[\sigvklen]\setminus\{i^*\}$.
\item Output $(\cva.\pk_{i^*,\sigvk^*[i^*]}^\prime,u_{i^*}^*,0^\ell,\st_\cB)$, where $\st_\cB$ includes all information $\cB$ knows at this point.
\item Obtain $(\ct_{i^*}^*,\cv_{i^*},\st_{\cB})$. Set $\cva.\ct^*_{i^*}\coloneqq \ct_{i^*}^*$.
\item Generate $\sigma^*\la\SIG.\Sign(\sigk^*,\cva.\ct_1^*\|\cdots\|\cva.\ct_\sigvklen^*)$ and set $\ct^*:=(\sigvk^*,(\cva.\ct_i^*)_{i\in[\sigvklen]},\sigma^*)$.
\end{itemize}
$\cB$ also sets $\cv=0$ if $\cv_{i^*}=0$.
Otherwise, $\cB$ sets $\cv=1$ if and only if $\CVA.\Dec(\cva.\sk_{i,\sigvk^*[i]},\cva.\ct_{i}^*)\ne\bot$ holds for every $i\in[\sigvklen]\setminus\{i^*\}$.
\item $\cB$ runs $b^\prime\la\cA(\ct^*,\cv,\st)^{\ODec{2}(\cdot)}$, where $\ODec{2}$ is simulated in exactly the same way as $\ODec{1}$.
$\cB$ outputs $\beta^\prime=1$ if $b=b^\prime$ and otherwise outputs $\beta^\prime=0$.

\end{enumerate}

We define $\Good$ as the event that $\cB$ does not abort when simulating decryption oracles.
We also let the challenge bit in the security experiment played by $\cB$ be $\beta$.
$\cB$ aborts with probability $\sigvklen-1/\sigvklen$ regardless of the value of $\beta$, that is, $\Pr[\Good|\beta=0]=\Pr[\Good|\beta=1]=1/\sigvklen$ holds.
Then, $\cB$'s advantage is calculated as follows.

\begin{align}
\abs{\Pr[\beta'=1|\beta=0]-\Pr[\beta'=1|\beta=1]}
&=\abs{\Pr[b=b^\prime \land \Good|\beta=0]-\Pr[b=b^\prime \land \Good|\beta=1]}\\
&= \frac{1}{\sigvklen}\abs{\Pr[b=b^\prime | \beta=0 \land \Good]-\Pr[b=b^\prime | \beta=1 \land \Good]}\\
&=\frac{1}{\sigvklen}\abs{\Pr[1\la\hybi{1}]-\Pr[1\la\hybi{2}]}.
\end{align}

The second line follows from the fact that we have $\Pr[\Good|\beta=0]=\Pr[\Good|\beta=1]=1/\sigvklen$ as stated above.
The third line follows since $\cB$ perfectly simulates $\hybi{1}$ (resp. $\hybi{2}$) conditioned that $\beta=0$ (resp. $\beta=1$) and the event $\Good$ occurs.
Thus, from the $\onekey{\INDCVApkTA}$ security of $\CVA$, we have $\abs{\Pr[1\la\hybi{1}]-\Pr[1\la\hybi{2}]}\le\negl(\secp)$.

Clearly, we have $\Pr[1\la\hybi{2}]=\frac{1}{2}$.
From the above discussions, $\oneCCA$ satisfies $\onekey{\INDoneCCApkTA}$ security.

\subsection{Boosting $\onekey{\INDoneCCApkTA}$ Security into $\onekey{\INDCCApkTA}$ Security}\label{sec:indccapkta}

We show how to transform $\onekey{\INDoneCCApkTA}$ secure QPKE into $\onekey{\INDCCApkTA}$ secure one using tokenized MAC.  

We construct a QPKE scheme $\CCA=(\CCA.\SKGen,\CCA.\PKGen,\CCA.\Enc,\CCA.\Dec)$ using the following building blocks.
\begin{itemize}
\item A QPKE scheme $\oneCCA=(\oneCCA.\SKGen,\oneCCA.\PKGen,\oneCCA.\Enc,\oneCCA.\Dec)$.
\item A tokenized MAC scheme $\TMAC=(\TMAC.\SKGen,\TMAC.\TKGen,\TMAC.\Sign,\TMAC.\Ver)$.
\item A signature scheme $\SIG=(\SIG.\Gen,\SIG.\Sign,\SIG.\Ver)$.
\end{itemize}

The construction of $\CCA$ is as follows.

\begin{itemize}
    \item 
    $\CCA.\SKGen(1^\secp)\to (\sk,\vk):$
    Run $(\onecca.\sk,\onecca.\vk)\gets \oneCCA.\SKGen(1^\secp)$ and $\mk\la\TMAC.\SKGen(1^\secp)$. 
    Output $\sk:=(\onecca.\sk,\mk)$ and $\vk:=\onecca.\vk$.
    \item
    $\CCA.\PKGen(\sk)\to \pk:$
    Parse $\sk:=(\onecca.\sk,\mk)$.
    Run $\onecca.\pk \gets \oneCCA.\PKGen(\onecca.\sk)$ and $\token\la\TMAC.\TKGen(\mk)$ and outputs $\pk:=(\onecca.\pk,\token)$. 
    \item
    $\CCA.\Enc(\vk,\pk,\msg)\to \ct:$
    Parse $\vk:=\onecca.\vk$ and $\pk:=(\onecca.\pk,\token)$.
    Run $(\sigvk,\sigk)\la\SIG.\Gen(1^\secp)$.
    Run $\onecca.\ct \gets \oneCCA.\Enc(\onecca.\vk,\onecca.\pk,\sigvk\|\msg)$.
    Run $\tmac.\sigma\la\TMAC.\Sign(\token,\allowbreak\onecca.\ct)$.
    Run $\sig.\sigma\la\SIG.\Sign(\sigk,\onecca.\ct\|\tmac.\sigma)$.
    Output $\ct:=(\onecca.\ct,\tmac.\sigma,\sig.\sigma)$.
    \item
    $\CCA.\Dec(\sk,\ct)\to \msg:$
    Parse $\sk:=(\onecca.\sk,\mk)$ and $\ct=(\onecca.\ct,\tmac.\sigma,\sig.\sigma)$.
    Output $\bot$ if $\TMAC.\Ver(\mk,\onecca.\ct,\tmac.\sigma)=\bot$ and otherwise go to the next step.
    Run $\sigvk\|\msg\la\oneCCA.\Dec(\onecca.\sk,\onecca.\ct)$, and output $\bot$ if $\sigvk\|\msg=\bot$ or $\SIG.\Ver(\sigvk,\onecca.\ct\|\tmac.\sigma,\sig.\sigma)=\bot$.
    Otherwise, output $\msg$.
\end{itemize}

\paragraph{Correctness and strong decryption error detectability.}
The correctness and the strong decryption error detectability of $\CCA$ immediately follow from those of $\oneCCA$ and the correctness of $\TMAC$ and $\SIG$. 

\paragraph{$\onekey{\INDCCApkTA}$ security.}
We prove that if $\oneCCA$ satisfies $\onekey{\INDoneCCApkTA}$ security, $\TMAC$ satisfies unforgeability, and $\SIG$ satisfies strong unforgeability, then $\CCA$ satisfies $\onekey{\INDCCApkTA}$ security.

Let $\cA$ be any QPT adversary attacking the $\onekey{\INDCCApkTA}$ security of $\CCA$.
We proceed the proof using a sequence of experiments.

\begin{description}
\item[$\hybi{0}$:]This is the original security experiment for the $\onekey{\INDCCApkTA}$ security of $\CCA$ played between $\cA$ and the challenger. The detailed description is as follows.
\begin{enumerate}
\item The challenger generates $(\onecca.\sk,\onecca.\vk)\gets \oneCCA.\SKGen(1^\secp)$ and $\mk\la\TMAC.\SKGen(1^\secp)$, and sets $\sk:=(\onecca.\sk,\mk)$ and $\vk:=\onecca.\vk$.
The challenger also generates $\onecca.\pk \gets \oneCCA.\PKGen(\onecca.\sk)$ and $\token\la\TMAC.\TKGen(\mk)$ and sets $\pk:=(\onecca.\pk,\token)$.

\item The challenger runs $(\pk^\prime,\msg_0,\msg_1,\st)\la\cA(\vk,\pk)^{\ODec{1}(\cdot)}$, where $\ODec{1}(\ct)$ behaves as follows.
\begin{itemize}
\item Parse $\ct=(\onecca.\ct,\tmac.\sigma,\sig.\sigma)$.
\item Output $\bot$ if $\TMAC.\Ver(\mk,\onecca.\ct,\tmac.\sigma)=\bot$ and otherwise go to the next step.
\item Run $\sigvk\|\msg\la\oneCCA.\Dec(\onecca.\sk,\onecca.\ct)$, and output $\bot$ if $\sigvk\|\msg=\bot$ or $\SIG.\Ver(\sigvk,\onecca.\ct\|\tmac.\sigma,\sig.\sigma)=\bot$. Otherwise, output $\msg$.
\end{itemize}
\item The challenger parses $\pk^\prime:=(\onecca.\pk^\prime,\token^\prime)$ and picks $b\la\bit$.
The challenger generates $\ct^*$ as follows.
\begin{itemize}
\item Run $(\sigvk^*,\sigk^*)\la\SIG.\Gen(1^\secp)$. 
\item Run $\onecca.\ct^* \gets \oneCCA.\Enc(\onecca.\vk,\onecca.\pk^\prime,\sigvk^*\|\msg_b)$.
\item Run $\tmac.\sigma^*\la\TMAC.\Sign(\token^\prime,\onecca.\ct^*)$.
\item Run $\sig.\sigma^*\la\SIG.\Sign(\sigk^*,\onecca.\ct^*\|\tmac.\sigma^*)$.
\item Set $\ct^*:=(\onecca.\ct^*,\tmac.\sigma^*,\sig.\sigma^*)$.
\end{itemize}
The challenger sets $\cv=0$ if $\CCA.\Dec(\sk,\ct^*)=\bot$ and $\cv=1$ otherwise.
\item The challenger runs $b^\prime\la\cA(\ct^*,\cv,\st)^{\ODec{2}(\cdot)}$, where $\ODec{2}(\ct)$ behaves in the same way as $\ODec{1}$ except that $\ODec{2}$ returns $\bot$ if $\ct=\ct^*$.
The challenger outputs $1$ if $b=b^\prime$ and otherwise outputs $0$.
\end{enumerate}
\item[$\hybi{1}$:] This is the same as $\hybi{0}$ except that $\ODec{2}$ given $\ct=(\onecca.\ct,\tmac.\sigma,\sig.\sigma)$ returns $\bot$ if $\onecca.\ct=\onecca.\ct^*$. 
\end{description}

We define the following events.
\begin{description}
\item[$\DecError_j$:]In $\hybi{j}$, It holds that $\oneCCA.\Dec(\onecca.\sk,\onecca.\ct^*)\notin\{\sigvk^*\|\msg_b,\bot\}$.
\item[$\Forge_j$:]In $\hybi{j}$, $\cA$ queries $\ct=(\onecca.\ct,\tmac.\sigma,\sig.\sigma)$ to $\ODec{2}$ such that $\onecca.\ct=\onecca.\ct^*$, $\ct\ne\ct^*$, and $\SIG.\Ver(\sigvk,\onecca.\ct\|\tmac.\sigma,\sig.\sigma)=\top$, where $\sigvk\|\msg\la\oneCCA.\Dec(\onecca.\sk,\onecca.\ct)$.
\end{description}
$\ODec{2}$ returns $\bot$ for a queried ciphertext $\ct=(\onecca.\ct,\tmac.\sigma,\sig.\sigma)$ such that $\onecca.\ct=\onecca.\ct^*$ in $\hybi{0}$, unless the event $\DecError_0$ or $\Forge_0$ occur.
Thus, we have $\abs{\Pr[1\la\hybi{0}]-\Pr[1\la\hybi{1}]}\le\Pr[\DecError_1\lor\Forge_1]\le\Pr[\DecError_1]+\Pr[\Forge_1]$.
We have $\Pr[\DecError_1]\le\negl(\secp)$ from the strong decryption error detectability of $\oneCCA$ and $\Pr[\Forge_1]\le\negl(\secp)$ from the strong unforgeability of $\SIG$.
From these, we obtain $\abs{\Pr[1\la\hybi{0}]-\Pr[1\la\hybi{1}]}\le\negl(\secp)$.

\begin{description}
\item[$\hybi{2}$:] This is the same as $\hybi{1}$ except that $\ODec{1}$ has a state $(s,t)$ that is initially set to $(0,\bot)$ and behaves as follows.
\begin{itemize}
\item Parse $\ct=(\onecca.\ct,\tmac.\sigma,\sig.\sigma)$.
\item If $s=1$, do the following
\begin{itemize}
\item Output $\bot$ if $\onecca.\ct\ne t$. Otherwise, go to the next step.
\item If $\TMAC.\Ver(\mk,\onecca.\ct,\tmac.\sigma)=\bot $ output $\bot$. Otherwise, go to the next step.
\item Run $\sigvk\|\msg\la\oneCCA.\Dec(\onecca.\sk,\onecca.\ct)$, and output $\bot$ if $\sigvk\|\msg=\bot$ or $\SIG.\Ver(\sigvk,\allowbreak\onecca.\ct\|\tmac.\sigma,\sig.\sigma)=\bot$. Otherwise, output $\msg$.
\end{itemize} 
\item If $s=0$, do the following.
\begin{itemize}
\item If $\TMAC.\Ver(\mk,\onecca.\ct,\tmac.\sigma)=\bot $ output $\bot$. Otherwise, update the state $(s,t)$ into $(1,\onecca.\ct)$ and go to the next step.
\item Run $\sigvk\|\msg\la\oneCCA.\Dec(\onecca.\sk,\onecca.\ct)$, and output $\bot$ if $\sigvk\|\msg=\bot$ or $\SIG.\Ver(\sigvk,\allowbreak\onecca.\ct\|\tmac.\sigma,\sig.\sigma)=\bot$. Otherwise, output $\msg$.
\end{itemize}
\end{itemize}
Also, the state $(s,t)$ is passed to $\ODec{2}$ at the end of the execution of $\ODec{1}$ and $\ODec{2}$ behaves in the same way as $\ODec{1}$ except that $\ODec{2}$ given $\ct=(\onecca.\ct,\tmac.\sigma,\sig.\sigma)$ returns $\bot$ if $\onecca.\ct=\onecca.\ct^*$.
\end{description}

From the unforgeability of $\TMAC$, we have $\abs{\Pr[1\la\hybi{1}]-\Pr[1\la\hybi{2}]}\le\negl(\secp)$.

In $\hybi{2}$, the decryption oracle decrypts at most one queried ciphertext $\ct=(\onecca.\ct,\tmac.\sigma,\sig.\sigma)$ such that $\onecca.\ct\ne\onecca.\ct^*$.
Then, from the $\onekey{\INDoneCCApkTA}$ security of $\oneCCA$, we have $\Pr[1\la\hybi{2}]\le\frac{1}{2}+\negl(\secp)$.
From the above discussions, $\CCA$ satisfies $\onekey{\INDCCApkTA}$ security.

\subsection{$\onekey{\INDCCApkTA}$ Security to $\INDCCApkTA$ Security}\label{sec:onekey_mkey}

We show how to construct $\INDCCApkTA$ secure QPKE from $\onekey{\INDCCApkTA}$ secure one.  

We construct a QPKE scheme $\Mkey=(\Mkey.\SKGen,\Mkey.\PKGen,\Mkey.\Enc,\Mkey.\Dec)$ using the following building blocks.
\begin{itemize}
\item A QPKE scheme $\Onekey=(\Onekey.\SKGen,\Onekey.\PKGen,\Onekey.\Enc,\Onekey.\Dec)$.
\item PRFs $\{\PRF_K\}_{K\in \bit^\secp}$.
\item A signature scheme $\SIG=(\SIG.\Gen,\SIG.\Sign,\SIG.\Ver)$.
\end{itemize}

The construction of $\Mkey$ is as follows.

\begin{itemize}
    \item 
    $\Mkey.\SKGen(1^\secp)\to (\sk,\vk):$
    Generate $K\la\bit^\secp$ and $(\sigvk,\sigk)\la\SIG.\Gen(1^\secp)$.
    Output $\sk:=(K,\sigk)$ and $\vk:=\sigvk$.
    \item
    $\Mkey.\PKGen(\sk)\to \pk:$
    Parse $\sk:=(K,\sigk)$ and generate $\snum\la\bit^\secp$.
    Run $r\la\PRF_K(\snum)$, $(\okey.\sk,\okey.\vk)\la\Onekey.\SKGen(1^\secp;r)$, and $\okey.\pk\la\Onekey.\PKGen(\okey.\sk)$.
    Run $\sigma\la\SIG.\Sign(\sigk,\snum\|\okey.\vk)$.
    Outputs $\pk:=(\snum,\okey.\vk,\okey.\pk,\sigma)$. 
    \item
    $\Mkey.\Enc(\vk,\pk,\msg)\to \ct:$
    Parse $\vk:=\sigvk$ and $\pk:=(\snum,\okey.\vk,\okey.\pk,\sigma)$.
    Output $\bot$ if $\SIG.\Ver(\sigvk,\snum\|\okey,\sigma)=\bot$ 
    and otherwise go to the next step.
    Run $\okey.\ct \gets \Onekey.\Enc(\okey.\vk,\okey.\pk,\msg)$.
    Output $\ct:=(\snum,\okey.\ct)$.
    \item
    $\Mkey.\Dec(\sk,\ct)\to \msg:$
    Parse $\sk:=(K,\sigk)$ and $\ct=(\snum,\okey.\ct)$.
    Run $r\la\PRF_K(\snum)$ and $(\okey.\sk,\okey.\vk)\la\Onekey.\SKGen(1^\secp;r)$.
    Output $\msg\la\Onekey.\Dec(\okey.\sk,\okey.\ct)$.
\end{itemize}

\paragraph{Correctness and strong decryption error detectability.}
The correctness and the strong decryption error detectability of $\Mkey$ immediately follow from those of $\Onekey$ and the correctness of $\SIG$. 

\paragraph{$\INDCCApkTA$ security.}
We prove that if $\Onekey$ satisfies $\onekey{\INDCCApkTA}$ security, $\{\PRF_K\}_{K\in \bit^\secp}$ is a secure PRF, and $\SIG$ satisfies strong unforgeability, then $\Mkey$ satisfies $\INDCCApkTA$ security.

Let $\cA$ be any QPT adversary attacking the $\INDCCApkTA$ security of $\Mkey$.
Let $m$ be a polynomial of $\secp$.
We proceed the proof using a sequence of experiments.

\begin{description}
\item[$\hybi{0}$:]This is the original security experiment for the $\INDCCApkTA$ security of $\Mkey$ played between $\cA$ and the challenger.
The detailed description is as follows.
\begin{enumerate}
\item The challenger generates $K\la\bit^\secp$ and $(\sigvk,\sigk)\la\SIG.\Gen(1^\secp)$, and sets $\sk:=(K,\sigk)$ and $\vk:=\sigvk$.
The challenger generates $\pk_i$ for every $i\in[m]$ as follows.
\begin{itemize}
\item Generate $\snum_i\la\bit^\secp$.
\item Run $r_i\la\PRF_K(\snum_i)$, $(\okey.\sk_i,\okey.\vk_i)\la\Onekey.\SKGen(1^\secp;r_i)$, and $\okey.\pk_i\la\Onekey.\PKGen(\okey.\sk_i)$.
\item Run $\sigma_i\la\SIG.\Sign(\sigk,\snum_i\|\okey.\vk_i)$.
\item Set $\pk_i:=(\snum_i,\okey.\vk_i,\okey.\pk_i,\sigma_i)$. 
\end{itemize}

\item The challenger runs $(\pk^\prime,\msg_0,\msg_1,\st)\la\cA(\vk,\pk_1,\cdots,\pk_m)^{\ODec{1}(\cdot)}$, where $\ODec{1}(\ct)$ behaves as follows.
\begin{itemize}
\item Parse $\ct=(\snum,\okey.\ct)$.
\item Run $r\la\PRF_K(\snum)$ and $(\okey.\sk,\okey.\vk)\la\Onekey.\SKGen(1^\secp;r)$.
\item Return $\msg\la\Onekey.\Dec(\okey.\sk,\okey.\ct)$.
\end{itemize}

\item The challenger parses $\pk^\prime:=(\snum^\prime,\okey.\vk^\prime,\okey.\pk^\prime,\sigma^\prime)$ and picks $b\la\bit$. The challenger generates $\ct^*$ as follows.
\begin{itemize}
\item Set $\ct^*:=\bot$ if $\SIG.\Ver(\sigvk,\snum^\prime\|\okey.\vk^\prime,\sigma^\prime)=\bot$ and otherwise go to the next step.
\item Run $\okey.\ct^* \gets \Onekey.\Enc(\okey.\vk^\prime,\okey.\pk^\prime,\msg_b)$.
\item Set $\ct^*:=(\snum^\prime,\okey.\ct^*)$.
\end{itemize}
The challenger also sets $\cv:=0$ if $\Mkey.\Dec(\sk,\ct^*)=\bot$ and otherwise sets $\cv:=1$.

\item The challenger runs $b^\prime\la\cA(\cv,\ct^*,\st)^{\ODec{2}(\cdot)}$, where $\ODec{2}$ behaves exactly in the same way as $\ODec{1}$ except that $\ODec{2}$ given $\ct$ returns $\bot$ if $\ct=\ct^*$.
The challenger outputs $1$ if $b=b^\prime$ and otherwise outputs $0$.

\end{enumerate}
\item[$\hybi{1}$:] This is the same as $\hybi{1}$ except that $\PRF_K(\cdot)$ is replaced with a truly random function.
\end{description}

From the security of $\PRF$, we have $\abs{\Pr[1\la\hybi{0}]-\Pr[1\la\hybi{1}]}\le\negl(\secp)$.

\begin{description}
\item[$\hybi{2}$:] This is the same as $\hybi{1}$ except that if $\snum^\prime\|\okey.\vk^\prime\ne \snum_i\|\okey.\vk_i$ for every $i\in[m]$, the challenger sets $\ct^*:=\bot$.
\end{description}

From the strong unforgeability of $\SIG$, we have $\abs{\Pr[1\la\hybi{1}]-\Pr[1\la\hybi{2}]}\le\negl(\secp)$.

\begin{description}
\item[$\hybi{3}$:] This is the same as $\hybi{2}$ except that if $\snum^\prime\|\okey.\vk^\prime= \snum_i\|\okey.\vk_i$ for some $i\in[m]$, the challenger generates 
$\okey.\ct^*\la\Onekey.\Enc(\okey.\vk_i,\okey.\pk^\prime,0^\ell)$.
\end{description}

To estimate $\abs{\Pr[1\la\hybi{2}]-\Pr[1\la\hybi{3}]}$, we construct the following adversary $\cB$ that uses $\cA$ and attacks the $\onekey{\INDCCApkTA}$ security of $\Onekey$.
\begin{enumerate}
\item Given $(\okey.\vk,\okey.\pk)$, $\cB$ picks $i^*\la[m]$,
generates $(\sigvk,\sigk)\la\SIG.\Gen(1^\secp)$, and sets $\vk:=\sigvk$.
$\cB$ then generates $\snum_{i^*}\la\bit^\secp$, sets $\okey.\vk_{i^*}:=\okey.\vk$ and $\okey.\pk_{i^*}:=\okey.\pk$, generates $\sigma_{i^*}\la\SIG.\Sign(\sigk,\snum_{i^*}\|\okey.\vk_{i^*})$, and sets $\pk_{i^*}:=(\snum_{i^*},\okey.\vk_{i^*},\okey.\pk_{i^*},\sigma_{i^*})$.
$\cB$ prepares an empty list $\keylist$.
$\cB$ does the following for every $i\in[m]\setminus\{i^*\}$.
\begin{itemize}
\item Generate $\snum_i\la\bit^\secp$ and $r_i\la\bit^\secp$.
\item Run $(\okey.\sk_i,\okey.\vk_i)\la\Onekey.\SKGen(1^\secp;r_i)$ and $\okey.\pk_i\la\Onekey.\PKGen(\okey.\sk_i)$.
\item Run $\sigma_i\la\SIG.\Sign(\sigk,\snum_i\|\okey.\vk_i)$.
\item Set $\pk_i:=(\snum_i,\okey.\vk_i,\okey.\pk_i,\sigma_i)$ and adds $(\snum_i,\okey.\sk_i)$ to $\keylist$. 
\end{itemize}

\item $\cB$ runs $(\pk^\prime,\msg_0,\msg_1,\st)\la\cA(\vk,\pk_1,\cdots,\pk_m)^{\ODec{1}(\cdot)}$, where $\ODec{1}(\ct)$ is simulated as follows.
\begin{itemize}
\item Parse $\ct:=(\snum,\okey.\ct)$.
\item If $\snum=\snum_{i^*}$, queries $\okey.\ct$ to its decryption oracle and forwards the answer to $\cA$.
Otherwise, go to the next step.
\item If there exists an entry of the form $(\snum,\okey.\sk)$ in $\keylist$, return $\Onekey.\Dec(\okey.\sk,\okey.\ct)$ to $\cA$.
Otherwise, go to the next step.
\item Generate $r\la\bit^\secp$ and $(\okey.\vk,\okey.\sk)\la\Onekey.\SKGen(1^\secp;r)$, and add $(\snum,\okey.\sk)$ to $\keylist$.
\item Return $\Onekey.\Dec(\okey.\sk,\okey.\ct)$ to $\cA$.
\end{itemize}

\item $\cB$ parses $\pk^\prime:=(\snum^\prime,\okey.\vk^\prime,\okey.\pk^\prime,\sigma^\prime)$ and picks $b\la\bit$. $\cB$ does the following.
\begin{itemize}
\item Set $\ct^*:=\bot$ if $\SIG.\Ver(\sigvk,\snum^\prime\|\okey.\vk^\prime,\sigma^\prime)=\bot$ and otherwise go to the next step.
\item Set $\ct^*:=\bot$ if $\snum^\prime\|\okey.\vk^\prime\ne \snum_i\|\okey.\vk_i$ for every $i\in[m]$. Otherwise, go to the next step.
\item Abort with $\beta^\prime:=0$ if $\snum^\prime\|\okey.\vk^\prime\ne \snum_{i^*}\|\okey.\vk_{i^*}$. Otherwise, go to the next step.
\item Output $(\okey.\pk^\prime,\msg_b,0^\ell,\st_{\cB})$, where $\st_{\cB}$ is all information that $\cB$ knows at this point.
\item Obtain $(\okey.\ct^*,\cv,\st_{\cB})$.
\item Set $\ct^*:=(\snum^\prime,\okey.\ct^*)$.
\end{itemize}

\item $\cB$ runs $b^\prime\la\cA(\cv,\ct^*,\st)^{\ODec{2}(\cdot)}$, where $\ODec{2}$ is simulated exactly in the same way as $\ODec{1}$ except that $\ODec{2}$ given $\ct$ returns $\bot$ if $\ct=\ct^*$.
$\cB$ outputs $\beta^\prime:=1$ if $b=b^\prime$ and otherwise outputs $\beta^\prime:=0$.

\end{enumerate}

We define $\Good$ as the event that $\cB$ does not abort when generating the challenge ciphertext.
Then, letting the challenge bit in the security experiment played by $\cB$ be $\beta$, $\cB$'s advantage is calculated as follows.

\begin{align}
\abs{\Pr[\beta'=1|\beta=0]-\Pr[\beta'=1|\beta=1]}
&=\abs{\Pr[b=b^\prime \land \Good|\beta=0]-\Pr[\beta'=1 \land \Good|\beta=1]}\\
&\ge \frac{1}{m}\abs{\Pr[b=b^\prime | \beta=0 \land \Good]-\Pr[\beta'=1 | \beta=1 \land \Good]}\\
&=\frac{1}{m}\abs{\Pr[1\la\hybi{2}]-\Pr[1\la\hybi{3}]}.
\end{align}

The second line follows from the fact that we have $\Pr[\Good|\beta=0]=\Pr[\Good|\beta=1]\ge\frac{1}{m}$.
The third line follows since $\cB$ perfectly simulates $\hybi{2}$ (resp. $\hybi{3}$) conditioned that $\beta=0$ (resp. $\beta=1$) and the event $\Good$ occurs.
Thus, from the $\onekey{\INDCCApkTA}$ security of $\CCA$, we have $\abs{\Pr[1\la\hybi{2}]-\Pr[1\la\hybi{3}]}\le\negl(\secp)$.

Clearly, we have $\Pr[1\la\hybi{3}]=\frac{1}{2}$.
From the above discussions, $\Mkey$ satisfies $\INDCCApkTA$ security.

%% file: submission_CCA.tex
\section{Chosen Ciphertext Security}
We extend IND-pkT-CPA security to the chosen ciphertext setting as follows.
\begin{definition}[$\INDCCApkTA$ security]\label{def_mainbody_indccapkta}
For any polynomial $m$, and any QPT adversary $\cA$, we have
\begin{equation}
   \Pr\left[b\gets\cA^{\ODec{2}(\cdot)}(\ct^\ast,\cv,\st):
   \begin{array}{r}
   (\sk,\vk)\gets\SKGen(1^\secp)\\
   \pk_1,...,\pk_m\gets\PKGen(\sk)^{\otimes m}\\
   (\pk',\msg_0,\msg_1,\st)\gets\cA^{\ODec{1}(\cdot)}(\vk,\pk_1,...,\pk_m)\\
   b\gets \bit\\ 
   \ct^\ast\gets\Enc(\vk,\pk',\msg_b)\\
   \cv:=0\textrm{~if~}\Dec(\sk,\ct^*)=\bot\textrm{~and otherwise~}\cv:=1
   \end{array}
   \right]
   \le \frac{1}{2}+\negl(\secp).
\end{equation}
Here, $\pk_1,...,\pk_m\gets\PKGen(\sk)^{\otimes m}$
means that $\PKGen$ is executed $m$ times and $\pk_i$ is the output of the $i$th execution of $\PKGen$.
$\st$ is a quantum internal state of $\cA$, which can be entangled with $\pk'$.
Also, $\ODec{1}(\ct)$ returns $\Dec(\sk,\ct)$ for any $\ct$.
$\ODec{2}$ behaves identically to $\ODec{1}$ except that $\ODec{2}$ returns $\bot$ to the input $\ct^*$.
\end{definition}

In the full version, we show a generic compiler that upgrades an IND-pkT-CPA secure QPKE scheme to IND-pkt-CCA secure one while preserving decryption error detectability by additionally using OWFs.  
Combined with our construction of IND-pkT-CPA secure QPKE scheme based on OWFs, we obtain the following theorem. 
\begin{theorem}
    If there exist OWFs, then there exists a QPKE scheme that satisfies IND-pkT-CCA security and decryption error detectability. 
\end{theorem}
Due to space limitation, we omit the proof.
Its overview can be found in \Cref{sec:overview_CCA} and the full proof can be found in the full version. 

%% file: reusable.tex
\section{Recyclable Variants}\label{sec:recyclable}
In the construction given in 
\ifnum\cameraready=0
\cref{sec:construction,sec:cca_const},
\else
\cref{sec:construction}, 
\fi
a quantum public key can be used to encrypt only one message and a sender needs to obtain a new quantum public key whenever it encrypts a message. This is not desirable from practical perspective. 
In this section, we define recyclable QPKE where a sender only needs to receive one quantum public key to send arbitrarily many messages, and then show how to achieve it. 

\subsection{Definitions}
The definition is similar to QPKE as defined in~\cref{def:QPKE} except that the encryption algorithm outputs a classical \emph{recycled} key that can be reused to encrypt messages many times.   


\begin{definition}[Recyclable QPKE]
A recyclable QPKE scheme with message space $\bit^\ell$ is a set of algorithms  $(\SKGen,\PKGen,\Enc,\rEnc,\Dec)$ such that
\begin{itemize}
    \item 
    $\SKGen(1^\secp)\to (\sk,\vk):$
    It is a PPT algorithm that, on input the security parameter $\secp$, outputs
    a classical secret key $\sk$ and a classical verification key $\vk$. 
    \item
    $\PKGen(\sk)\to \pk:$
    It is a QPT algorithm that, on input $\sk$, outputs
    a quantum public key $\pk$.
    \item
    $\Enc(\vk,\pk,\msg)\to (\ct,\rk):$ 
    It is a QPT algorithm that, on input $\vk$, $\pk$, and a plaintext $\msg\in\bit^\ell$, outputs a classical ciphertext $\ct$ and classical recycled key $\rk$.  
    \item
    $\rEnc(\rk,\msg)\to \ct:$ 
    It is a PPT algorithm that, on input $\rk$ and a plaintext $\msg\in\bit^\ell$, outputs a classical ciphertext $\ct$. 
    \item
    $\Dec(\sk,\ct)\to \msg':$ 
    It is a classical deterministic polynomial-time algorithm that, on input $\sk$ and $\ct$, outputs $\msg'\in\bit^\ell \cup \{\bot\}$.
\end{itemize}

We require the following correctness.

\paragraph{\bf Correctness:}
 For any $\msg,\msg'\in\bit^\ell$,
\begin{equation}
   \Pr\left[
   \msg\gets\Dec(\sk,\ct)
   \land\\ 
   \msg'\gets\Dec(\sk,\ct'):
   \begin{array}{r}
   (\sk,\vk)\gets\SKGen(1^\secp)\\
   \pk\gets\PKGen(\sk)\\
   (\ct,\rk)\gets\Enc(\vk,\pk,\msg)\\
   \ct'\gets\rEnc(\rk,\msg')
   \end{array}
   \right] \ge 1-\negl(\secp).
\end{equation}
\end{definition}

\begin{definition}[IND-pkT-CPA Security for Recyclable QPKE]
We require the followings.
\paragraph{\bf Security under quantum public keys:}
For any polynomial $m$, and any QPT adversary $\cA$,
\begin{equation}
   \Pr\left[b\gets\cA^{\rEnc(\rk,\cdot)}(\ct^\ast ,\st):
   \begin{array}{r}
   (\sk,\vk)\gets\SKGen(1^\secp)\\
   \pk_1,...,\pk_m\gets\PKGen(\sk)^{\otimes m}\\
   (\pk',\msg_0,\msg_1,\st)\gets\cA(\vk,\pk_1,...,\pk_m)\\
   b\gets \bit\\ 
   (\ct^\ast,\rk) \gets\Enc(\vk,\pk',\msg_b)
   \end{array}
   \right] \le \frac{1}{2}+\negl(\secp).
\end{equation}

\paragraph{\bf Security under recycled keys:} 
For any polynomial $m$, and any QPT adversary $\cA$,
\begin{equation}
   \Pr\left[b\gets\cA^{\rEnc(\rk,\cdot)}(\ct^\ast ,\st'):
   \begin{array}{r}
   (\sk,\vk)\gets\SKGen(1^\secp)\\
   \pk_1,...,\pk_m\gets\PKGen(\sk)^{\otimes m}\\
   (\pk',\msg,\st)\gets\cA(\vk,\pk_1,...,\pk_m)\\
   (\ct^\prime,\rk) \gets\Enc(\vk,\pk',\msg)\\
    (\msg_0,\msg_1,\st')\gets \cA^{\rEnc(\rk,\cdot)}(\ct^\prime,\st)\\ 
   b\gets \bit\\ 
   \ct^\ast \gets\rEnc(\rk,\msg_b)
   \end{array}
   \right] \le \frac{1}{2}+\negl(\secp).
\end{equation}
Here, $\pk_1,...,\pk_m\gets\PKGen(\sk)^{\otimes m}$
means that $\PKGen$ is executed $m$ times and $\pk_i$ is the output of the $i$th execution of $\PKGen$,  $\rEnc(\rk,\cdot)$ means a classically-accessible encryption oracle, 
and $\st$ and $\st'$ are quantum internal states of $\cA$, which can be entangled with $\pk'$.
\end{definition}

\begin{definition}[IND-pkT-CCA Security for Recyclable QPKE]
We require the followings.
\paragraph{\bf Security under quantum public keys:}
For any polynomial $m$, and any QPT adversary $\cA$,
\begin{equation}
   \Pr\left[b\gets\cA^{\rEnc(\rk,\cdot),\ODec{2}(\cdot)}(\ct^\ast ,\cv, \st):
   \begin{array}{r}
   (\sk,\vk)\gets\SKGen(1^\secp)\\
   \pk_1,...,\pk_m\gets\PKGen(\sk)^{\otimes m}\\
   (\pk',\msg_0,\msg_1,\st)\gets\cA^{\ODec{1}(\cdot)}(\vk,\pk_1,...,\pk_m)\\
   b\gets \bit\\ 
   (\ct^\ast,\rk) \gets\Enc(\vk,\pk',\msg_b)\\
   \cv:=0\textrm{~if~}\Dec(\sk,\ct^*)=\bot\textrm{~and otherwise~}\cv:=1
   \end{array}
   \right] \le \frac{1}{2}+\negl(\secp).
\end{equation}

\paragraph{\bf Security under recycled keys:} 
For any polynomial $m$, and any QPT adversary $\cA$,
\begin{equation}
   \Pr\left[b\gets\cA^{\rEnc(\rk,\cdot),\ODec{2}(\cdot)}(\ct^\ast ,\cv, \st'):
   \begin{array}{r}
   (\sk,\vk)\gets\SKGen(1^\secp)\\
   \pk_1,...,\pk_m\gets\PKGen(\sk)^{\otimes m}\\
   (\pk',\msg,\st)\gets\cA^{\ODec{1}(\cdot)}(\vk,\pk_1,...,\pk_m)\\
   (\ct^\prime,\rk) \gets\Enc(\vk,\pk',\msg)\\
    (\msg_0,\msg_1,\st')\gets \cA^{\rEnc(\rk,\cdot),\ODec{1}(\cdot)}(\ct^\prime,\st)\\ 
   b\gets \bit\\ 
   \ct^\ast \gets\rEnc(\rk,\msg_b)\\
   \cv:=0\textrm{~if~}\Dec(\sk,\ct^*)=\bot\textrm{~and otherwise~}\cv:=1
   \end{array}
   \right] \le \frac{1}{2}+\negl(\secp).
\end{equation}
Here, $\pk_1,...,\pk_m\gets\PKGen(\sk)^{\otimes m}$
means that $\PKGen$ is executed $m$ times and $\pk_i$ is the output of the $i$th execution of $\PKGen$,  $\rEnc(\rk,\cdot)$ means a classically-accessible encryption oracle, 
and $\st$ and $\st'$ are quantum internal states of $\cA$, which can be entangled with $\pk'$.
Also, $\ODec{1}(\ct)$ returns $\Dec(\sk,\ct)$ for any $\ct$.
$\ODec{2}$ behaves identically to $\ODec{1}$ except that $\ODec{2}$ returns $\bot$ to the input $\ct^*$.
\end{definition}

\subsection{Construction} 
We show a generic construction of recyclable QPKE from (non-recyclable) QPKE with classical ciphertexts and SKE via standard hybrid encryption.  

Let $\QPKE=(\QPKE.\SKGen,\QPKE.\PKGen,\QPKE.\Enc,\QPKE.\Dec)$ be a (non-recyclable) QPKE scheme with message space $\bit^\secp$ and $\SKE=(\SKE.\Enc,\SKE.\Dec)$ be an SKE scheme with message space $\bit^\ell$. Then we construct a recyclable QPKE scheme $\QPKE'=(\QPKE'.\SKGen,\QPKE'.\PKGen,\QPKE'.\Enc,\allowbreak \QPKE'.\rEnc,\QPKE'.\Dec)$ with message space $\bit^\ell$ as follows:
\begin{itemize}
    \item 
    $\QPKE'.\SKGen(1^\secp)\to (\sk',\vk'):$
    Run  $(\sk,\vk)\gets \QPKE.\SKGen(1^\secp)$ and output a secret key $\sk':=\sk$ and verification key $\vk':=\vk$. 
    \item
    $\QPKE'.\PKGen(\sk')\to \pk':$
    Run $\pk \gets \QPKE.\PKGen(\sk)$ and outputs $\pk':=\pk$. 
    \item
    $\QPKE'.\Enc(\vk',\pk',\msg)\to (\ct',\rk'):$
        Parse $\pk'=\pk$ and $\vk'=\vk$,  
        sample $K\gets \bit^\secp$,  run  $\ct \gets \QPKE.\Enc(\vk,\pk,K)$ and $\ct_{\ske}\gets \SKE.\Enc(K,\msg)$, and output a ciphertext $\ct':=(\ct,\ct_{\ske})$ and recycled key $\rk':=(K,\ct)$.
 \item
 $\QPKE'.\rEnc(\rk',\msg)\to \ct':$
Parse $\rk'=(K,\ct)$, run $\ct_{\ske}\gets \SKE.\Enc(K,\msg)$, and output a ciphertext $\ct':=(\ct,\ct_{\ske})$. 
    \item
    $\QPKE'.\Dec(\sk',\ct')\to \msg':$
    Parse $\ct'=(\ct,\ct_{\ske})$ and $\sk'=\sk$, run $K'\gets \QPKE.\Dec(\sk,\ct)$ and $\msg' \gets \SKE.\Dec(K',\ct_{\ske})$, and output $\msg'$. 
\end{itemize}

\paragraph{Correctness and decryption error detectability.}
Correctness of $\QPKE'$ immediately follows from correctness of $\QPKE$ and $\SKE$.
Also, the decryption error detectability of $\QPKE'$ directly follows from that of $\QPKE$.

\paragraph{IND-pkT-CPA security and IND-pkT-CCA security.}
If $\QPKE$ satisfies $\INDpkTCPA$ (resp. $\INDpkTCCA$) security and $\SKE$ satisfies IND-CPA (resp. IND-CCA) security, then $\QPKE'$ satisfies $\INDpkTCPA$ (resp. $\INDpkTCCA$) security.
The proofs can be done by standard hybrid arguments, thus omitted.



%% file: omitted_related_work.tex
\section{More Related Work and Open Problems}\label{sec:related_open}
\subsection{Related Works} 
The possibility that QPKE can be achieved from weaker assumptions was first pointed out by Gottesman~\cite{GottesmanPKE}, though he did not give any concrete construction.
The first concrete construction of QPKE was proposed by Kawachi, Koshiba, Nishimura, and Yamakami~\cite{EC:KKNY05}.
They formally defined the notion of QPKE with quantum public keys, and provided a construction satisfying it from a distinguishing problem of two quantum states.
Recently, Morimae and Yamakawa~\cite{EPRINT:MorYam22c} pointed out that QPKE defined by~\cite{EC:KKNY05} can be achieved from any classical or quantum symmetric key encryption almost trivially.
The constructions proposed in these two works have mixed state quantum public keys.
Then, subsequent works~\cite{cryptoeprint:2023/282,TCC:BGHMSVW23} independently studied the question whether QPKE with pure state quantum public keys can be constructed from OWFs or even weaker assumptions.

The definition of QPKE studied in the above works essentially assume that a sender can obtain intact quantum public keys.
As far as we understand, this requires unsatisfactory physical setup assumptions such as secure quantum channels or tamper-proof quantum hardware, regardless of whether the quantum public keys are pure states or mixed states.
In our natural setting where an adversary can touch the quantum channel where quantum public keys are sent, the adversary can easily attack the previous constructions by simply replacing the quantum public key on the channel with the one generated by itself that the adversary knows the corresponding secret key.
We need to take such adversarial behavior into consideration, unless we assume physical setup assumptions that deliver intact quantum public keys to the sender.
Our work is the first one that proposes a QPKE scheme secure in this natural setting assuming only classical authenticated channels that is the same assumption as classical PKE and can be implemented by digital signature schemes.
It is unclear if we could solve the problem in the previous constructions by using classical authenticated channels similarly to our work.
Below, we review the constructions of QPKE from OWFs proposed in the recent works.

The construction by Morimae and Yamakawa~\cite{EPRINT:MorYam22c} is highly simple.
A (mixed state) public key of their construction is of the form $(\ct_0,\ct_1)$, where $\ct_b$ is an encryption of $b$ by a symmetric key encryption scheme.
The encryption algorithm with input message $b$ simply outputs $\ct_b$.

Coladangelo~\cite{cryptoeprint:2023/282} constructed a QPKE scheme with
quantum public keys and quantum ciphertexts from pseudorandom functions (PRFs), which are constructed from OWFs.
The public key is 
\begin{equation}
\ket{\pk}\coloneqq\sum_y(-1)^{\PRF_k(y)}\ket{y},
\end{equation}
and the secret key is $k$.
The ciphertext for the plaintext $m$ is
\begin{equation}
(Z^x\ket{\pk}=\sum_y(-1)^{x\cdot y+\PRF_k(y)}\ket{y},
r,r\cdot x\oplus m),
\end{equation}
where $r$ is chosen uniformly at random.

 Barooti, Grilo, Huguenin-Dumittan, Malavolta, Sattath, Vu, and Walter~\cite{TCC:BGHMSVW23} constructed three QPKE schemes: (1) CCA secure QPKE with
quantum public keys and classical ciphertexts from OWFs (2) CCA1\footnote{Afther the adversary received a challenge ciphertext, they cannot access the decryption oracle.} secure QPKE with quantum public keys and ciphertexts from pseudorandom function-like states generators, (3) CPA secure QPKE with quantum public keys and classical ciphertexts from pseudo-random function-like states with proof of destruction.
All constructions considers security under the encryption oracle.
We review their construction based on OWFs.

Their construction is hybrid encryption of CPA secure QPKE (the KEM part) and CCA secure classical symmetric key encryption (the DEM part).
The public key is
\begin{equation}
\ket{\pk}\coloneqq   \sum_x\ket{x}\ket{\PRF_k(x)},
\end{equation}
and the secret key is $k$.
The encryption algorithm first measures $\ket{\pk}$ in the computational basis to get
$(x,\PRF_k(x))$ and outputs $(x,\mathsf{SKE}.\Enc(\PRF_k(x),m))$
as the ciphertext for the plaintext $m$, where $\mathsf{SKE}.\Enc$
is the encryption algorithm of a symmetric key encryption scheme.

We finally compare Quantum Key Distribution (QKD)~\cite{BB84} with our notion of QPKE. QKD also enables us to establish secure communication over an untrusted quantum channel assuming that an authenticated classical channel is available similarly to our QPKE. An advantage of QKD is that it is information theoretically secure and does not need any computational assumption. 
On the other hand, it has disadvantages that it must be interactive and parties must record secret information for each session. Thus, it is incomparable to the notion of QPKE.

\subsection{Open Problems}
In our construction, public keys are quantum states. It is an open problem whether
QPKE with classical public keys are possible from OWFs.
Another interesting open problem is whether we can construct QPKE defined in this work from an even weaker assumption than OWFs such as pseudorandom states generators.


In our model of QPKE, a decryption error 
may be caused by tampering attacks on the quantum public key.
To address this issue, we introduce the security notion we call decryption error detectability that guarantees that a legitimate receiver of a ciphertext can notice if the decrypted message is different from the message intended by the sender.
We could consider even stronger variant of decryption error detectability that requires that a sender can notice if a given quantum public key does not provide decryption correctness.
It is an open problem to construct a QPKE scheme satisfying such a stronger decryption error detectability.

%% file: pure.tex
\section{Pure State Public Key Variant}\label{sec:pure_pk_qpke}
As discussed in \Cref{sec:discussion}, we believe that the distinction between pure state and mixed state public keys is not important from a practical point of view. Nonetheless, it is a mathematically valid question if we can construct an IND-pkT-CPA secure QPKE scheme with pure state public keys. We give such a scheme based on the existence of quantum-secure OWFs 
by extending the construction given in \Cref{sec:construction}. 
For the ease of exposition, we first show a construction based on \emph{slightly superpolynomially secure} OWFs in \Cref{sec:construction_superpolynomial}. Then, we explain how to modify the scheme to base its security on standard polynomially secure OWFs in  \Cref{sec:construction_polynomial}.

\paragraph{Preparation.}

We define a fine-grained version of strong EUF-CMA security for digital signature schemes. 
\begin{definition}[$T$-strong EUF-CMA security]
A digital signature scheme $(\Gen,\Sign,\Ver)$ is $T$-strong EUF-CMA secure if the following holds: 
For any quantum adversary $\cA$ that runs in time $T$ and makes at most $T$ classical queries to the signing oracle $\Sign(k,\cdot)$,
\begin{equation}
   \Pr[(\msg^\ast,\sigma^\ast)\notin \mathcal{Q}~\land~\top\gets\Ver(\vk,\msg^*,\sigma^*):
   (k,\vk)\gets\Gen(1^\secp),
   (\msg^\ast,\sigma^\ast)\gets\cA^{\Sign(k,\cdot)}(\vk)
   ]\le T^{-1}, 
\end{equation}
where $\mathcal{Q}$ is the set of message-signature pairs returned by the signing oracle. 
\end{definition}
\begin{remark}
Strong EUF-CMA security defined in \Cref{def:sEUF-CMA} holds if and only if $T$-strong EUF-CMA security holds for all polynomials $T$.  
\end{remark}
\begin{remark}\label{rem:sig_from_OWF_superpoly}
We can show that there exists a $T$-strong EUF-CMA secure digital signature scheme for some $T=\secp^{\omega(1)}$ if slightly superpolynomially secure OWFs exist similarly to the proof of \Cref{thm:sig_from_OWF} in \cite[Sec. 6.5.2]{Book:Goldreich04}. Here, a superpolynomially secure OWF is a function $f$ for which there exists $T=\secp^{\omega(1)}$ such that any adversary with running time $T$ can invert $f$ with probability at most $T^{-1}$.
\end{remark}

We also need the following lemma in the security proof.
\begin{lemma}\label{lem:cannot_find_both}
For a function $H:\bit^{u+1}\rightarrow \bit^v$, let $\ket{\psi_H}:=\sum_{R\in \bit^{u+1}}\ket{R}\ket{H(R)}$. For any integer $m$ and (unbounded-time) quantum algorithm $\A$, 

\begin{align}
    \Pr_{H}[
    y_0=H(0\concat r)~\land~y_1=H(1\concat r)
    :(r,y_0,y_1)\gets \A(\ket{\psi_H}^{\otimes m})]\le (2m+1)^4(2^{-u}+2^{-v})
\end{align}
where $H$ is a uniformly random function from $\bit^{u+1}$ to $\bit^v$. 
\end{lemma}

We prove it using the result of \cite{EC:YamZha21}. We defer the proof to \Cref{sec:proof_lemma}.

\subsection{Construction from Slightly Superpolynomially Secure OWFs}\label{sec:construction_superpolynomial}
In this section, we construct a QPKE scheme that satisfies correctness and IND-pkT-CPA security (but not decryption error detectability) and has pure state public keys from 
$T$-strong EUF-CMA secure digital signatures for a superpolynomial $T$ and quantum-query secure PRFs. Note that they exist assuming the existence of slightly superpolynomially secure OWFs as noted in  \Cref{rem:sig_from_OWF_superpoly,rem:PRF}.  
The message space of our construction is $\bit$, but it can be extended to be arbitrarily many bits by parallel repetition.  
Let $(\Gen,\Sign,\Ver)$ be a $T$-strong EUF-CMA secure digital signature scheme with a deterministic $\Sign$ algorithm and message space $\bit^{u+v+1}$    
and $\{\PRF_K:\bit^{u+1} \rightarrow \bit^v\}_{K\in \bit^\secp}$ be a quantum-query secure PRF 
where 
$T=\secp^{\omega(1)}$, 
$u:=\lfloor (\log T)/2\rfloor$, and
$v=\omega(\log \secp)$.  

Then we construct a QPKE scheme $(\SKGen,\PKGen,\Enc,\Dec)$ as follows.
\begin{itemize}
    \item 
$\SKGen(1^\secp)\to(\sk,\vk):$ Run $(k,\vk)\gets\Gen(1^\secp)$ and sample $K\gets \bit^\secp$. Output $\sk\coloneqq (k,K)$ and $\vk$. 
    \item
    $\PKGen(\sk)\to\pk:$
    Parse $\sk=(k,K)$.
    Choose $r\gets\bit^u$.
    By running $\Sign$ and $\PRF$ coherently, generate the state 
    \begin{eqnarray}
    \ket{\psi_\sk}\coloneqq
    \sum_{r\in \bit^u}\ket{r}_{\regR}\otimes 
    \left( 
    \begin{array}{l}
\ket{0}_\regA\otimes\ket{y(0,r)}_\regB \otimes \ket{\sigma(0,r)}_\regC\\
    +\ket{1}_\regA\otimes\ket{y(1,r)}_\regB \otimes \ket{\sigma(1,r)}_\regC
      \end{array}
    \right)
    \end{eqnarray} 
    over registers $(\regR,\regA,\regB,\regC)$
    where 
    $y(b,r):=\PRF_K(b\concat r)$
    and
    $\sigma(b,r):=\Sign(k,b\concat r\concat y(b,r))$
    for $b\in \bit$ and $r\in \bit^u$. (We omit $K$ and $k$ from the notations for simplicity.)
    Output 
    \begin{align}
    \pk&\coloneqq\ket{\psi_\sk}.
    \end{align}
   \item
   $\Enc(\vk,\pk,b)\to\ct:$
   Parse
    $\pk=\rho$, where $\rho$ is a quantum state over registers $(\regR,\regA,\regB,\regC)$. 
    The $\Enc$ algorithm consists of the following three steps. 
   \begin{enumerate}
    \item
    It coherently checks the signature in $\rho$. In other words, it applies the unitary
   \begin{equation}   U_{\vk}\ket{r}_{\regR}\ket{\alpha}_\regA\ket{\beta}_{\regB}\ket{\gamma}_\regC\ket{0...0}_\regE
   =\ket{r}_{\regR}\ket{\alpha}_\regA\ket{\beta}_{\regB}\ket{\gamma}_\regC\ket{\Ver(\vk,\alpha\|r\|\beta,\gamma)}_\regE
   \end{equation} 
    on $\rho_{\regR,\regA,\regB,\regC}\otimes|0...0\rangle\langle0...0|_\regE$, and measures the register $\regE$ in the computational basis. 
    If the result is $\bot$, it outputs $\ct\coloneqq\bot$ and halts.
    If the result is $\top$, 
    it goes to the next step.
    \item
    It applies $Z^b$ on the register $\regA$. 
    \item
    It measures 
    $\regR$ in the computational basis to get $r$ and 
    all qubits in the registers $(\regA,\regB,\regC)$ 
    in the Hadamard basis to get the result $d$.
    It outputs
    \begin{equation}
    \ct\coloneqq(r,d).
    \end{equation}
    \end{enumerate}
   \item
   $\Dec(\sk,\ct)\to b':$
   Parse $\sk=(k,K)$ and
    $\ct=(r,d)$.
    Output 
    \begin{equation}
    b'\coloneqq d\cdot 
   (0\|y(0,r)\|\sigma(0,r) \oplus 1\|y(1,r)\|\sigma(1,r)).
    \end{equation}
\end{itemize}

\begin{theorem}\label{thm:IND-pkTA_QPKE_from_SIG_pure}
If 
$(\Gen,\Sign,\Ver)$ is a $T$-strong EUF-CMA secure digital signature scheme and $\{\PRF_K:\bit^{u+1} \rightarrow \bit^v\}_{K\in \bit^\secp}$ is a quantum-query secure PRF, 
then the QPKE scheme $(\SKGen,\PKGen,\Enc,\Dec)$ above is correct and satisfies IND-pkT-CPA security.
\end{theorem}
\begin{proof}[Proof (sketch)]
The correctness is easily seen similarly to the proof of \Cref{thm:IND-pkTA_QPKE_from_SIG}. 

For IND-pkT-CPA security, we only explain the differences from the proof of \Cref{thm:IND-pkTA_QPKE_from_SIG} since it is very similar. We define Hybrid $0,1$, and $2$ similarly to those in the proof of \Cref{thm:IND-pkTA_QPKE_from_SIG}. For clarity, we describe them in \cref{hyb0_pure,hyb1_pure,hyb2_pure}. 

Assume that the IND-pkT-CPA security of our construction is broken by a QPT adversary
$\cA$. It means the QPT adversary $\cA$ wins Hybrid 0 with a non-negligible advantage.
Then, it is clear that there is another QPT adversary $\cA'$ that wins Hybrid 1  
with a non-negligible advantage. ($\cA'$ has only to do the Hadamard-basis measurement by itself.)

From the $\cA'$, we construct a QPT adversary $\cA''$ that wins 
Hybrid 2 with a non-negligible probability based on a similar proof to that in   \cref{sec:B}. 
Indeed, the proof is almost identical once we show that 
any QPT adversary given polynomially many copies of the public key can output a valid signature for a message that is not of the form $b\|r\|y(b,r)$ only with a negligible probability. 
To prove this, we consider a reduction algorithm that queries signatures on \emph{all} messages of the form $b\|r\|y(b,r)$. Thus, the reduction algorithm makes $2^{u+1}$ queries and runs in time $2^u\cdot \poly(\secp)$. Since we have $2^{u+1}<T$ and $2^u\cdot \poly(\secp)<T$ for sufficiently large $\secp$ by $u=\lfloor (\log T)/2 \rfloor$, which in particular implies $2^u\le T^{1/2}$, and $T=\secp^{\omega(1)}$, the reduction enables us to prove the above property assuming the $T$-strong EUF-CMA security of the digital signature scheme.\footnote{In the proof for the mixed state public key version in \Cref{sec:B}, the reduction algorithm only needs to query signatures on $b\|r$ for $r$'s used in one of the public keys given to the adversary.  On the other hand, in the pure state public key case, each public key involves all $r$'s and thus the reduction algorithm needs to query signatures on superpolynomially many messages. This is why we need superpolynomial security for the digital signature scheme. 
}

Thus, we are left to prove that no QPT adversary can win Hybrid 2 with a non-negligible probability.  
Let Hybrid 2' be a hybrid that works similarly to Hybrid 2 except that $y(b,r)$ is defined as  $y(b,r):=H(b\concat r)$ for a uniformly random function $H$ instead of $\PRF$. By the quantum-query security of $\PRF$, the winning probabilities in Hybrid 2' and Hybrid 2 are negligibly close. Thus, it suffices to prove the winning probability in Hybrid 2' is negligible. This is proven by a straightforward reduction to \Cref{lem:cannot_find_both} noting that $\ket{\psi_\sk}$ with the modification of $y(b,r)$ as above can be generated from  $\ket{\psi_H}=\sum_{R\in \bit^{u+1}}\ket{R}\ket{H(R)}$ by coherently running $\Sign$. 
This completes the proof of IND-pkT-CPA security.

\protocol{Hybrid 0}
{Hybrid 0}
{hyb0_pure}
{
\begin{enumerate}
    \item \label{item:initial_pure}
    $\cC$ runs $(k,\vk)\gets\Gen(1^\secp)$.
    $\cC$ sends $\vk$ to $\cA$.
  \item
    $\cC$ sends $\ket{\psi_{\sk}}^{\otimes m}$
    to the adversary $\cA$, where
   \begin{eqnarray}
    \ket{\psi_\sk}\coloneqq
    \sum_{r\in \bit^u}\ket{r}\otimes 
    \left( 
    \begin{array}{l}
\ket{0}\otimes\ket{y(0,r)} \otimes \ket{\sigma(0,r)}\\
    +\ket{1}\otimes\ket{y(1,r)} \otimes \ket{\sigma(1,r)}
      \end{array}
    \right)
    \end{eqnarray} 
    \item
    \label{Ar_pure}
    $\cA$ generates a quantum state over registers $(\regR,\regA,\regB,\regC,\regD)$. 
    ($(\regR,\regA,\regB,\regC)$ corresponds to the quantum part of $\pk'$, and $\regD$ corresponds to $\st$.)
    $\cA$ sends the registers $(\regR,\regA,\regB,\regC)$ to $\cC$.
    $\cA$ keeps the register $\regD$.
    \item
    \label{sign_pure}
    $\cC$ coherently checks the signature in the state sent from $\cA$.
    If the result is $\bot$, it sends $\bot$ to $\cA$ and halts.
    If the result is $\top$, it goes to the next step. 
   
    \item
    \label{bchosen_pure}
    $\cC$ chooses $b\gets \bit$.
    $\cC$ applies $Z^b$ on the register $\regA$.
 \item
    $\cC$ measures $\regR$ in the computational basis to get $r$. 
    \item
    $\cC$ measures all qubits in $(\regA,\regB,\regC)$ in the Hadamard basis to get the result $d$.
    $\cC$ sends $(r,d)$ to $\cA$.
    \item
    \label{last_pure}
    $\cA$ outputs $b'$.
    If $b'=b$, $\cA$ wins.
\end{enumerate}
}

\protocol{Hybrid 1}
{Hybrid 1}
{hyb1_pure}
{
\begin{itemize}
    \item[1.-6.]
    All the same as \cref{hyb0_pure}.
    \item[7.]
    $\cC$ does not do the Hadamard-basis measurement,
    and $\cC$ sends $r$ and registers $(\regA,\regB,\regC)$ to $\cA$.
    \item[8.]
    The same as \cref{hyb0_pure}.
\end{itemize}
}

\protocol{Hybrid 2}
{Hybrid 2}
{hyb2_pure}
{
\begin{itemize}
    \item[1.-7.]
    All the same as \cref{hyb1_pure}.
    \item[8.]
    $\cA$ outputs $(\mu_0,\mu_1)$.
    If $\mu_0=y(0,r)\|\sigma(0,r)$ and $\mu_1=y(1,r)\|\sigma(1,r)$, $\cA$ wins.
\end{itemize}
}
\end{proof}
\begin{remark}
    We can add decryption error detectability by \Cref{thm:add_decryption_error_detectability} and extend it to recyclable QPKE by the construction of \Cref{sec:recyclable}. These extensions preserve the property that public keys are pure states. 
\end{remark}

\subsection{Construction from Polynomially Secure OWFs}\label{sec:construction_polynomial}
We explain how to extend the construction in \Cref{sec:construction_superpolynomial} to base security on standard polynomial hardness of OWFs. We rely on a similar idea to the ``on-the-fly-adaptation'' technique introduced in  \cite{C:DotSch15}.  The reason why we need superpolynomial security in \Cref{sec:construction_superpolynomial} is that the reduction algorithm for the transition from Hybrid 1 to 2 has to make $2^{u+1}\approx 2T^{1/2}$ signing queries for a superpolynomial $T$.  Suppose that we set $T$ to be a polynomial. i.e., $T=\secp^c$ for some constant $c$. Then, the reduction algorithm for the transition from Hybrid 1 to 2 works under polynomial security of the digital signature scheme. The problem, however, is that we cannot show that the winning probability in Hybrid 2 is negligible: It can be only bounded by $(2m+1)^4(2^{-u}+2^{-v})$, which is not negligible since $2^{u}\approx T^{1/2}=\secp^{c/2}$. On the other hand, we can make it arbitrarily small inverse-polynomial by making $c$ larger. Based on this observation, we can show the following: 
Let $(\SKGen_c,\PKGen_c,\Enc_c,\Dec_c)$ be the QPKE scheme given in \Cref{sec:construction_superpolynomial} where $T:=\secp^c$. Then, for any polynomials $p$ and $m$, 
there exists a constant $c$ such that 
any QPT adversary given $m$ copies of the quantum public key 
has an advantage to break IND-pkT-CPA security of $(\SKGen_c,\PKGen_c,\Enc_c,\Dec_c)$ at most $1/p(\secp)$ for all sufficiently large $\secp$. 

Then, our idea is to parallelly run  $(\SKGen_c,\PKGen_c,\Enc_c,\Dec_c)$ for $c=1,2,...,\secp$ where the encryption algorithm generates a $\secp$-out-of-$\secp$ secret sharing of the message and encrypts $c$-th share under $\Enc_c$.\footnote{In fact, it suffices to parallelly run $(\SKGen_c,\PKGen_c,\Enc_c,\Dec_c)$ for $c=1,2,...,\eta(\secp)$ for any super-constant function $\eta$.} 
Suppose that this scheme is not IND-pkT-CPA secure. Then, there is a polynomial $q$ and QPT adversary $\cA$ given $m=\poly(\secp)$ copies of the quantum public key that has an advantage to break the IND-pkT-CPA security at least $1/q(\secp)$ for infinitely many $\secp$. 
For each $c$, 
it is easy to construct a QPT adversary $\cA_c$ that breaks IND-pkT-CPA security of $(\SKGen_c,\PKGen_c,\Enc_c,\Dec_c)$ with the same advantage as $\cA$'s advantage.  
On the other hand, 
by the observation explained above, 
we can take a constant $c$ (depending on $q$ and $m$) such that any QPT adversary given $m$ copies of the public key has an advantage to break IND-pkT-CPA security of $(\SKGen_c,\PKGen_c,\Enc_c,\Dec_c)$ at most $1/2q(\secp)$ for all sufficiently large $\secp$.  
This is a contradiction. Thus, the above scheme is IND-pkT-CPA secure.

\subsection{Proof of \Cref{lem:cannot_find_both}}\label{sec:proof_lemma}
For proving \Cref{lem:cannot_find_both}, we rely on the following lemma shown by \cite{EC:YamZha21}. 
\begin{lemma}[{\cite[Theorem 4.2]{EC:YamZha21}}]\label{lem:YZ}
Let $H:\calX\rightarrow \calY$ be a uniformly random function. 
Let $\A$ be an (unbounded-time) randomized algorithm that makes $q$ quantum queries to $H$ and outputs a classical string $z$. 
Let $\mathcal{C}$ be an (unbounded-time) randomized algorithm that takes $z$ as input, 
makes $k$ classical queries to $H$, and outputs $\top$ or $\bot$.
Let $\B$ be the following algorithm that makes at most $k$ classical queries to $H$: 
\begin{description}
\item[$\B^{H}()$:]~ It does the following:
\begin{enumerate}
\item Choose a function $H':\calX\rightarrow \calY$ from a family of $2q$-wise independent hash functions.
    \item 
    For each $j\in[k]$, uniformly pick $(i_j,b_j)\in ([q]\times \bit) \cup \{(\bot,\bot)\}$ under the constraint that there does not exist $j\neq j'$ such that $i_j=i_{j'}\neq \bot$.
    \item Initialize a stateful oracle $\mathcal{O}$ to be a quantumly-accessible classical oracle that computes $H'$.
    \item Run $\A^{\mathcal{O}}()$ where $\mathcal{O}$ is simulated as follows.
   When $\A$ makes its $i$-th query, the oracle is simulated as follows:
    \begin{enumerate}
      \item If $i=i_j$ for some $j\in[k]$, 
        measure $\A$'s  query register to obtain $x'_j$, query $x'_j$ to the random oracle $H$ to obtain $H(x'_j)$, and do either of the following.
        \begin{enumerate}
        \item If $b_j=0$, reprogram 
        $\mathcal{O}$ to output $H(x'_j)$ on $x'_j$ and answer $\A$'s $i_j$-th query by using the reprogrammed oracle. 
        \item If $b_j=1$, answer  $\A$'s $i_j$-th query by using the oracle before the reprogramming and then reprogram 
        $\mathcal{O}$ to output $H(x'_j)$ on $x'_j$. 
        \end{enumerate}
     \item Otherwise, answer $\A$'s $i$-th query by just using the oracle $\mathcal{O}$ without any measurement or  reprogramming.
    \end{enumerate}
    \item Output whatever $\A$ outputs.
\end{enumerate}
\end{description}
Then we have 
\begin{align}
     \Pr_H[\mathcal{C}^H(z)=\top :z\gets \B^{H}()] \geq \frac{1}{(2q+1)^{2k}}\Pr_H[\mathcal{C}^H(z)=\top :z\gets \A^{H}()]. \label{eq:lifting}
\end{align}
\end{lemma}
\begin{remark}
    There are the following  differences  from \cite[Theorem 4.2]{EC:YamZha21} in the statement of the lemma:
    \begin{enumerate}
        \item They consider inputs to $\A$ and $\B$. We omit them because this suffices for our purpose. 
        \item They consider a more general setting where $\A$ and $\B$ interact with $\mathcal{C}$. We focus on the non-interactive setting. 
        \item They do not explicitly write how $\B$ works in the statement of \cite[Theorem 4.2]{EC:YamZha21}. But this is stated at the beginning of its proof. 
    \end{enumerate}
\end{remark}

Using the above lemma, it is easy to prove \Cref{lem:cannot_find_both}.
\begin{proof}[Proof of \Cref{lem:cannot_find_both}]
For an algorithm $\A$ in \Cref{lem:cannot_find_both}, 
let $\tilde{\A}$ be an oracle-aided algorithm that generates $m$ copies of $\ket{\psi_H}$ by making $m$ quantum queries to $H$ on uniform superpositions of inputs and then runs $\A(\ket{\psi_H}^{\otimes m})$. 
Let $\mathcal{C}$ be an oracle-aided algorithm that takes $z=(r,y_0,y_1)$ as input, makes two classical queries $0\|r$ and $1\|r$ to $H$, and outputs $\top$ if and only if $y_0=H(0\concat r)$ and $y_1=H(1\concat r)$.  
By \Cref{lem:YZ}, 
we have 
\begin{align}
     &\Pr_H[ y_0=H(0\concat r)~\land~y_1=H(1\concat r) :(r,y_0,y_1)\gets \tilde{\B}^{H}()]\\ \geq 
     &\frac{1}{(2m+1)^{4}}\Pr_H[ y_0=H(0\concat r)~\land~y_1=H(1\concat r) :(r,y_0,y_1)\gets \tilde{\A}^{H}()] \label{eq:lifting_application}
\end{align}
where $\tilde{\B}$ is to $\tilde{\A}$ as $\B$ (defined in \Cref{lem:YZ}) is to $\A$. 
By the definition of $\tilde{\B}$, it just makes at most two classical queries to $H$ on independently and uniformly random inputs $R_0,R_1$. The probability that we happen to have $\{R_0,R_1\}=\{0\| r,1\|r\}$ for some $r\in \bit^u$ is $2^{-u}$. Unless the above occurs, either  $H(0\concat r)$ or $H(1\concat r)$ is uniformly random to
$\tilde{\B}$ for all $r$, and thus 
the probability that its output satisfies $ y_0=H(0\concat r)$ and $y_1=H(1\concat r)$ is at most $2^{-v}$.
Thus, we have 
\begin{align}
    \Pr_H[ y_0=H(0\concat r)~\land~y_1=H(1\concat r) :(r,y_0,y_1)\gets \tilde{\B}^{H}()]\le 2^{-u}+2^{-v}. \label{eq:upper_bound_tilde_B}
\end{align}
Combining \Cref{eq:lifting_application,eq:upper_bound_tilde_B}, we obtain \Cref{lem:cannot_find_both}. 
\end{proof}